\documentclass[12pt, reqno]{amsart}

\usepackage[utf8]{inputenc}
\usepackage{amsthm}
\usepackage{caption}
\usepackage{amsmath}
\usepackage{amssymb}
\usepackage{mathrsfs}
\usepackage{mathtools}
\usepackage{rotating}
\usepackage{bbm}
\usepackage{adjustbox}
\usepackage{booktabs}
\usepackage[online,flushleft]{threeparttable}
\usepackage{multirow}
\usepackage{array}
\usepackage{import}
\newcolumntype{C}[1]{>{\centering\arraybackslash}p{#1}}

\usepackage[open,openlevel=1]{bookmark}
\usepackage{xcolor}
\usepackage{float}
\hypersetup{
colorlinks,
linkcolor={red!50!black},
citecolor={blue!50!black},
urlcolor={blue!80!black}
}

\usepackage{thmtools, thm-restate}  

\declaretheoremstyle[
spaceabove=6pt, spacebelow=6pt,
headfont=\normalfont\bfseries,
notefont=\mdseries, notebraces={(}{)},
bodyfont=\normalfont,
postheadspace=1em,
qed=$\blacksquare$
]{examplestyle}

\declaretheoremstyle[
spaceabove=6pt, spacebelow=6pt,
headfont=\normalfont\bfseries,
notefont=\mdseries, notebraces={(}{)},
bodyfont=\itshape,
postheadspace=1em
]{theorem}

\declaretheoremstyle[
spaceabove=6pt, spacebelow=6pt,
headfont=\normalfont\bfseries,
notefont=\mdseries, notebraces={(}{)},
bodyfont=\normalfont,
postheadspace=1em
]{assumption}

\declaretheoremstyle[
spaceabove=4pt, spacebelow=4pt,
headfont=\itshape\bfseries,
notefont=\mdseries, notebraces={(}{)},
bodyfont=\itshape,
postheadspace=0.2em,
qed=\qedsymbol
]{remark}

\declaretheorem[style=theorem]{theorem}
\declaretheorem[style=theorem, numbered=no, name=Theorem]{theorem*}
\declaretheorem[style=remark,name=Remark]{remark}

\declaretheorem[style=assumption,name=Assumption]{assumption}
\declaretheorem[style=assumption,name=Assumption 1', numbered=no]{assumption1prime}
\declaretheorem[style=assumption,name=Assumption 2', numbered=no]{assumption2prime}

\declaretheorem[style=theorem,name=Definition]{definition}
\declaretheorem[style=theorem, name=Lemma]{lemma}

\declaretheorem[numbered=no,style=definition,name=Question]{question*}  
\declaretheorem[style=definition,name=Definition, numbered=no]{definition*}  

\declaretheorem[style=examplestyle]{example}
\renewcommand\thmcontinues[1]{continued}

\newtheorem{thm}{Assumption}

\newcommand{\real}{\mathbb{R}}

\newcommand{\norm}[1]{\left\Vert #1\right\Vert }

\newcommand{\indicator}{\mathbbm{1}}

\newcommand{\ud}{\mathrm{d}}
\newcommand{\prob}{\mathbb{P}}
\newcommand{\expt}{\mathbb{E}}
\newcommand{\var}{\text{\normalfont Var}}

\newcounter{steps}

\DeclareMathOperator*{\argmax}{arg\,max}
\DeclareMathOperator*{\argmin}{arg\,min}

\usepackage[backend=biber, 
    style=authoryear, 
	  date=year,
		uniquename=false,
		uniquelist=false,
		maxcitenames=3,   
		maxbibnames=99,  
		doi=false,    
    isbn=false,   
    url=false,    
    eprint=false  
]{biblatex}

\AtEveryBibitem{
    \clearfield{note}      
    \clearfield{addendum}  
    \clearlist{publisher}
    \clearlist{location}
    \clearfield{issn}
    \clearfield{isbn}
}
\renewbibmacro{in:}{}

\usepackage[nodisplayskipstretch]{setspace}
\usepackage{geometry}
\usepackage{enumitem}
\setenumerate[1]{label=(\roman*)}
\setenumerate[2]{label=\normalfont{(\alph*)}}

\begin{document}

\title[]{Econometric Inference with Machine-Learned Proxies: \\ Partial Identification via Data Combination}


\author[]{Lixiong Li}
\address{Johns Hopkins University}
\email{lixiong.li@jhu.edu}
\thanks{ 
I thank Marc Henry and Yingyao Hu for their helpful feedbacks, and seminar participants at UPenn.\\
\noindent Address: Wyman Park Building 5th Floor, 3100 Wyman Park Drive, Baltimore, MD 21211 \\  Email: lixiong.li@jhu.edu \\
This version: \today
}

\keywords{}

\date{}

\dedicatory{Johns Hopkins University}

\begin{abstract}

Empirical researchers increasingly use upstream machine-learning (ML) methods to construct proxies for latent target variables from complex, unstructured data. A naive plug-in use of such proxies in downstream econometric models, however, can lead to biased estimation and invalid inference. This paper develops a framework for partial identification and inference in general moment models with ML-generated proxies. Our approach does not require restrictive assumptions on the upstream ML procedure, such as consistency or known convergence rates, nor does it require a complete validation sample containing all variables used in the downstream analysis. Instead, we assume access to two datasets: a downstream sample containing observed covariates and the proxy, and an auxiliary validation sample containing joint observations on the proxy and its target variable. We treat the proxy as a linking variable between these two samples, rather than as a literal noisy substitute for the latent target variable. Building on this idea, we develop a sharp identification strategy based on an unconditional optimal transport characterization and an inference procedure that controls asymptotic size using analytical critical values without resampling. Monte Carlo simulations show reliable size control and informative confidence sets across a range of predictive-accuracy scenarios.

\vspace{3em}

\noindent\textbf{Keywords:} machine learning proxies, data combination, partial identification, optimal transport, cross-fitting, nonclassical measurement error. 
\end{abstract}

\maketitle
\newpage

\section*{Introduction}

There is an increasing prevalence in economics and the social sciences of leveraging complex, unstructured data---such as text and imagery---to construct variables central to empirical analysis. In these contexts, machine learning (ML) algorithms are frequently employed to map these raw inputs onto proxies for latent constructs that would otherwise be inherently unobservable or prohibitively costly to quantify at scale. Such approaches have been instrumental in deriving measures of media slant \autocite{groseclose_measure_2005,gentzkow_what_2010}, firm-level political risk \autocite{hassan_firm-level_2019}, local economic activity \autocite{hu_illuminating_2022}, remote-work status in job postings \autocite{hansen_remote_2023}, and air pollution exposure \autocite{sager_clean_2025}. For broader overviews, see \textcite{mullainathan_machine_2017} and \textcite{gentzkow_text_2019}, along with the references therein.

This procedure is typically implemented in two steps. First, in an upstream step, a prediction rule $g$ is trained on a labeled dataset to map unstructured or high-dimensional data $X$ to a target variable $Z$. Second, in a downstream step, a researcher studies an economic model with observed covariates $W$, latent vector $Z$, and parameter $\theta$. Although $Z$ is not observed in the downstream step, the researcher does observe the unstructured inputs $X$, which allows the researcher to apply the prediction rule from the upstream step and construct the proxy $\hat{Z}=g(X)$. While we distinguish between these two steps for clarity, the distinction is conceptual, since the upstream ML training and the downstream empirical estimation may be conducted by the same researcher.

While ML-generated proxies have substantially expanded the scope of empirical research, incorporating them into downstream economic models poses a fundamental challenge. A naive plug-in approach that treats the ML-generated proxy $\hat{Z}$ as if it were the true latent variable $Z$ ignores both measurement error and the generated-regressor problem, and can therefore lead to biased estimation and invalid inference. This issue is compounded by two additional features of modern ML applications. First, the prediction rules used in practice are often highly complex, making it analytically difficult to characterize the statistical properties of $\hat{Z}$. In particular, convergence rates for $\hat{Z}$ are often unavailable, and in some settings even consistency of $\hat{Z}$ for $Z$ is unclear. Second, because the unstructured input $X$ may contain rich information not only about $Z$ but also about the observed covariates $W$, the resulting measurement error $Z-\hat{Z}$ is generally nonclassical: it may depend on $Z$, remain correlated with $W$ even conditional on $Z$, and even be endogenous to the downstream economic model.

In this paper, we propose a new framework to address these challenges by leveraging an auxiliary validation sample. We consider a setting in which the researcher has access to two distinct datasets. The first is a downstream sample containing observed covariates $W$, unstructured inputs $X$, and the proxy $\hat{Z}$ constructed from $X$ using the upstream prediction rule $g$. The second is an auxiliary validation sample containing observations on $Z$ and $X$, and therefore also on the proxy $\hat{Z}=g(X)$. Such validation data often arise naturally in practice, for example as a held-out sample from the upstream training data used to assess the predictive performance of the rule $g$. Importantly, we do not require the validation sample to contain $W$, nor do we assume any direct correspondence or individual-level matching between the observations in the validation and downstream samples.

The key idea of this paper is to reframe the role of $\hat{Z}$. Rather than treating $\hat{Z}$ as a literal but noisy substitute for $Z$, we view it as a low-dimensional summary of the raw unstructured input $X$ that links the validation and downstream samples. Because the validation sample contains joint observations on $(Z,\hat{Z})$, it allows us to learn the conditional distribution of $Z$ given $\hat{Z}$. Because $\hat{Z}$ is also observed in the downstream sample, this information can then be carried over to the downstream setting: for each observation $(W_i,\hat{Z}_i)$, although the realization of $Z_i$ remains unobserved, the validation sample informs the conditional distribution of $Z_i$ given $\hat{Z}_i$. In this sense, $\hat{Z}$ serves not merely as a noisy proxy for $Z$, but as the key linking variable that allows information about $Z$ from the validation sample to be incorporated into the downstream economic analysis. 

We then develop an identification strategy based on this idea. A natural way to implement the linking role of $\hat{Z}$ is through an optimal transport (OT) characterization conditional on $\hat{Z}$, as in \textcite{fan_partial_2025}. However, such a conditional OT strategy can be difficult to implement in practice, especially when $\hat{Z}$ is continuous or high-dimensional, because it requires estimating the conditional distribution of $Z$ given $\hat{Z}$ and solving a separate OT problem for each value in the support of $\hat{Z}$. To overcome this difficulty, we instead develop an \emph{unconditional} OT characterization based on the decoupling idea of \textcite{li_finite_2025}. Our approach replaces the continuum of conditional transport problems with a single OT problem formulated in terms of the unconditional distribution of $(W,\hat{Z})$ and the unconditional distribution of $(Z,\hat{Z})$, thereby rendering the method tractable in practice. Importantly, this unconditional characterization remains \emph{sharp} for the identified set of $\theta$: given the available data and maintained assumptions, the resulting bounds cannot be further tightened.

Finally, we develop an inference procedure for testing whether a candidate value of $\theta$ belongs to the identified set. Inference is challenging in this setting because OT problems typically exhibit nonstandard asymptotic behavior. A common workaround is to replace the OT problem with an entropy-regularized counterpart, but that approach is not suitable here: a fixed entropy penalty introduces regularization bias that invalidates the underlying OT characterization, whereas accommodating that bias through worst-case bounds leads to excessively conservative inference. We instead base our procedure on the Kantorovich dual formulation together with a sieve approximation to the relevant dual function space. The number of sieve terms is allowed to increase with the sample size, so that the approximation error, and hence its effect on the test statistic, vanishes asymptotically. To maintain tractability, the resulting procedure avoids resampling methods such as the bootstrap or subsampling. Instead, by applying sample splitting and cross-fitting, we obtain critical values directly from standard normal quantiles. This yields a tractable inference procedure in practice, as illustrated in our simulation exercises.

Our approach differs from the existing literature. In addressing the measurement error embedded in $\hat{Z}$, prior work has largely proceeded along two lines:

One line imposes structural restrictions on the measurement error or asymptotic requirements on the proxy itself. A standard example is a conditional independence assumption, under which the proxy $\hat{Z}$ is independent of auxiliary measurements or covariates conditional on the true latent variable $Z$ \autocite[see, e.g.,][]{hu_instrumental_2008, hu_illuminating_2022}. Related identifying assumptions also appear in recent work by \textcite{rambachan_program_2024} when a joint sample of $(W,Z,\hat{Z})$ is unavailable. A different version of this approach assumes that the measurement error in $\hat{Z}$ vanishes at a rate comparable to sampling error \autocite{battaglia_inference_2025}. In the present setting, however, these assumptions can be restrictive. Conditional independence is difficult to justify when the unstructured input $X$ contains information that may remain correlated with other covariates in the economic model even after conditioning on $Z$, while rate conditions are hard to verify when $\hat{Z}$ is generated by complex ML procedures whose asymptotic properties are analytically intractable.

A second line assumes access to a \emph{complete} validation sample containing joint observations on $(W,Z,\hat{Z})$ that is sufficient for point identification of $\theta_0$ on its own \autocite[see, e.g.,][]{chen_unified_2000, chen_measurement_2005, allon_machine_2023, angelopoulos_prediction-powered_2023, angelopoulos_ppi_2024, zhang_debiasing_2023, zrnic_cross-prediction-powered_2023, miao_task-agnostic_2024, rambachan_program_2024, kluger_prediction-powered_2025, sanford_adversarial_2025}. Within this framework, the downstream sample of $(W,\hat{Z})$ serves mainly to improve estimation efficiency. Although this approach avoids both structural assumptions on the measurement error and formal asymptotic requirements on the ML procedure, it imposes substantially stronger data requirements. In particular, this approach requires either that the downstream model not depend on the covariates $W$, or that the validation sample contain precisely the covariates $W$ needed for the downstream analysis. In practice, such complete validation samples are not always available: upstream ML researchers typically train prediction rules and release held-out validation data on $(Z,\hat{Z})$ without knowing which covariates $W$ will matter for a particular downstream application. Conversely, downstream researchers often face prohibitive costs or practical barriers to collecting the true values of $Z$ needed to obtain a joint sample of $(W,Z,\hat{Z})$.

The first contribution of this paper is to provide an alternative framework that avoids the limitations of existing approaches. Unlike the first line of the literature, our framework does not impose structural assumptions on the measurement error $Z-\hat{Z}$, nor does it require formal statistical guarantees, such as convergence rates or consistency, for the proxy $\hat{Z}$. Instead, we rely on validation data to learn the joint distribution of the true latent variable $Z$ and the proxy $\hat{Z}$. In this respect, our framework is closer in spirit to the second line of the literature. The key difference is that we do not require a \emph{complete} validation sample that, by itself, is sufficient for consistent estimation of the parameter of interest. The trade-off is that our framework generally delivers partial rather than point identification for the model parameter. The informativeness of our partial-identification framework depends on the quality of the proxy. When the ML-generated proxy is highly accurate, the identified set for the parameter would be tight. In the extreme case in which $\hat{Z}=Z$, the parameter is point identified within our framework. At the same time, the validity of our approach does not depend on the accuracy of the proxy. When $\hat{Z}$ is only a crude measure of $Z$, the resulting identified set remains valid, although it may be wide and therefore only weakly informative about the parameter. This feature faithfully reflects the informational content of a poor proxy: when the proxy contains little reliable information about the latent variable, one should \emph{not} expect the data to deliver sharp conclusions about the parameter of interest.

Beyond bypassing restrictive structural or asymptotic assumptions and strong data requirements, our framework offers several practical advantages. Because we treat $\hat{Z}$ not as a plug-in substitute for $Z$, but rather as a linking variable between the validation and downstream samples, the framework naturally accommodates settings in which $Z$ and $\hat{Z}$ take values in different spaces. For example, if $Z$ is a discrete group identity, $\hat{Z}$ need not itself be a categorical label. It may instead take the form of a vector of predicted probabilities, a record of the model's top-two guesses, a likelihood ranking, or some other statistical summary. When available, such richer representations of $\hat{Z}$ can carry more information about $Z$ and therefore lead to a more informative identified set. This flexibility also provides a simple way to combine proxies produced by multiple competing ML procedures: one can define a multidimensional $\hat{Z}$ whose components collect the outputs of the different procedures. More broadly, our perspective suggests that, in empirical economics, the quality of an ML procedure should be judged not solely by its out-of-sample predictive accuracy, but by the extent to which its output preserves the economically relevant information contained in the unstructured input $X$ for downstream moment conditions.

As a second contribution, the identification strategy developed in this paper applies more broadly to general data combination problems \autocite{cross_regressions_2002,ridder_chapter_2007} and is therefore of independent interest. Because our framework recasts the proxy as a variable that links the downstream and validation samples, our identification strategy carries over directly to classical data combination environments. Such problems are pervasive in empirical work and remain of substantial econometric interest. Recent applications include assessing algorithmic fairness \autocite{kallus_assessing_2022}, measuring intergenerational mobility \autocite{santavirta_name-based_2024}, estimating long-term treatment effects \autocite{athey_surrogate_2025, obradovic_identification_2026}, and combining stated and revealed preferences \autocite{meango_combining_2025}.

Recent advances in this literature have used OT tools to characterize data combination problems \autocite[see, e.g.,][]{hwang_bounding_2025, dhaultfoeuille_partially_2025, fan_partial_2025}. The paper most closely related to ours is \textcite{fan_partial_2025}, who develop a \emph{conditional} OT characterization for the same class of general moment models studied here. However, implementing their approach requires solving a separate transport problem for each realization of the overlap variable shared across the two datasets---in our setting, $\hat{Z}$. When this linking variable has large support or includes continuous components, the conditional approach can become computationally prohibitive. The \emph{unconditional} OT characterization developed in this paper complements their result by providing an alternative formulation. Rather than solving a continuum of conditional OT problems and estimating the corresponding conditional distributions, we solve a single OT problem defined over unconditional distributions. The trade-off is that the resulting unconditional OT problem is higher-dimensional than each conditional problem in \textcite{fan_partial_2025}. We view this as a favorable trade-off when the support of $\hat{Z}$ is large or when $\hat{Z}$ contains continuous components.

The remainder of the paper is organized as follows. Section \ref{sec:framework} introduces the analytical framework and discusses several empirical applications that fit into this setup. Section \ref{sec:validation_identification} establishes the identification results when an auxiliary validation sample containing both the proxy and the true latent target variable is available, while Section \ref{sec:validation_inference} develops the corresponding statistical inference procedures. Section \ref{sec:simulation} demonstrates the finite-sample performance of the proposed methods through Monte Carlo simulations. Section \ref{sec:extensions} explores an extension of the benchmark model to settings in which the validation sample is only available to identify the marginal distributions of subvectors of $(Z, \hat{Z})$. Finally, Section \ref{sec:discussion} concludes.

\subsection*{Notation} Throughout the paper, capital letters (e.g., $Z$, $\hat{Z}$, $W$, and $S$) denote random variables or random vectors, and lower-case letters (e.g., $z$, $\hat{z}$, $w$, and $s$) denote their corresponding realizations. All random variables are defined on a common complete probability space. Vectors are written as column vectors, and the superscript $\top$ denotes transpose. For vectors in finite-dimensional Euclidean spaces, $\norm{\cdot}$ denotes the Euclidean norm. 
We write $\expt_F[\cdot]$ for expectation under the distribution $F$; when no ambiguity arises, we simply write $\expt[\cdot]$.
 For a distribution $F$, we write $\mathcal{L}^1(F)$ for the space of $F$-integrable functions. We use $\coloneqq$ to denote definitions. Unless stated otherwise, all equalities and inequalities involving random variables are understood to hold almost surely.

\section{Analytical Framework and Examples}\label{sec:framework}
We study a two-stage empirical workflow consisting of an upstream machine-learning (ML) prediction step and a downstream
econometric inference step. In the downstream step, the applied researcher seeks to learn about a finite-dimensional
parameter $\theta_0\in\Theta$, characterized by the population moment conditions
\begin{equation}\label{eq:downstream_model}
	\expt\!\left[q(W,Z;\theta_0)\right] = 0,
\end{equation}
where $W$ is a vector of observed covariates, $Z$ is a vector of latent target variables that are \emph{not} observed in the downstream dataset, and $q(\cdot)$ is a known moment function. We assume that $W$ and $Z$ take values in finite-dimensional Euclidean spaces.

Because $Z$ is unobserved, the downstream researcher relies on a proxy generated by an upstream prediction rule,
\begin{equation}\label{eq:upstream_model}
	\hat{Z} = g(X),
\end{equation}
where $X$ denotes observed predictors. The prediction rule $g(\cdot)$ is produced upstream, using separate
training data on $(X,Z)$. In this paper, $X$ may consist of high-dimensional structured
covariates or unstructured inputs such as text or images. 

We impose minimal structure on the upstream learning step. The training procedure is left completely unrestricted, and the prediction rule $g(\cdot)$ could be misspecified for the conditional distribution of $Z$ given $X$. In particular, our analysis does not rely on any consistency, rate-of-convergence properties, or asymptotic distribution results of the estimated predictor. Dispensing with such asymptotic requirements is particularly appealing for empirical work, as formal convergence analysis are often intractable for the complex, off-the-shelf machine learning algorithms deployed in practice. For expositional simplicity in this section, we assume that the proxy $\hat{Z}$ and the latent target $Z$ share the same dimensionality. Nevertheless, our methodology readily accommodates settings in which $\hat{Z}$ and $Z$ take values in spaces of differing dimensions, as we discuss later in Section~\ref{sec:validation_identification}.

Our framework distinguishes the \emph{informativeness} of the proxy from the \emph{validity} of the econometric procedure.
The predictive accuracy of $g(\cdot)$ governs the informativeness of the resulting bounds. Intuitively, stronger predictive performance yields tighter restrictions linking $Z$ and $\hat Z$ and therefore a smaller identified set for $\theta_0$ in the downstream problem. By contrast, the validity of our approach does not hinge on $g(\cdot)$ being highly accurate or statistically consistent. Instead, validity is ensured whenever the downstream researcher has access to an auxiliary validation sample containing joint observations of $(Z,\hat Z)$. The joint distribution of $(Z,\hat Z)$ quantifies both the predictive performance of $g(\cdot)$ and the nature of the discrepancy between the proxy $\hat Z$ and the latent target variable $Z$. This data requirement is often modest. For example, the validation sample may be a subset of labeled observations held out from the data used to train $g(\cdot)$, or it may be obtained through a separate validation exercise after $g(\cdot)$ is trained. In the next section, we formalize the assumptions governing the validation sample and show how to combine this information with the downstream observations and moment conditions to construct \emph{sharp} bounds on $\theta_0$ that account for the proxy-induced measurement error.

To summarize, in this two-stage empirical setting, the downstream researcher observes a sample drawn from the joint distribution of $(W,X)$ but lacks direct observations of the target variable $Z$. Consequently, the sample analogs of the moment conditions in \eqref{eq:downstream_model} cannot be directly evaluated. This motivate us to develop an approach that leverages the observed proxy $\hat{Z}$, alongside the joint distribution of $(Z, \hat{Z})$ recovered from the validation sample, to bound the moments and partially identify $\theta_0$. 

We conclude this section with several examples illustrating how  empirical applications map into the framework of downstream econometric model \eqref{eq:downstream_model} and the upstream prediction \eqref{eq:upstream_model}.

\begin{example}[Machine-learned measures as dependent variables]\label{ex:pm25}
Consider the empirical application in \textcite{sager_clean_2025}, which studies the effects of U.S.\ Clean Air Act standards on fine particulate matter (PM$_{2.5}$). For concreteness, we focus on their baseline specification (equation (2) in \textcite{sager_clean_2025}),
\begin{equation}\label{eq:pm25_reg}
	\Delta \mathit{PM}_i = \alpha_0 + \beta_0 \Delta \mathit{NA}_i + \varepsilon_i,
\end{equation}
where $\Delta \mathit{PM}_i$ denotes the change in true PM$_{2.5}$ concentration for census tract $i$ between a pre-treatment period and a post-treatment period, and $\Delta \mathit{NA}_i$ is a treatment indicator capturing the change in nonattainment status.\footnote{In \textcite{sager_clean_2025}, $\Delta \mathit{NA}_i=1$ for tracts that become subject to regulatory treatment from 2005 onward.} The true parameter of interest is $\theta_0=(\alpha_0,\beta_0)'$.

To connect \eqref{eq:pm25_reg} to the generic formulation in \eqref{eq:downstream_model}, suppose for simplicity that there are two periods $t\in\{0,1\}$ and write $\Delta \mathit{PM}_i = \mathit{PM}_{i,1}-\mathit{PM}_{i,0}$. Under the standard exogeneity assumption that $\expt[\varepsilon_i \mid \Delta \mathit{NA}_i] = 0$, the linear regression model \eqref{eq:pm25_reg} implies the population moment conditions
\begin{equation}\label{eq:pm25_moments}
\expt\!\left[
\begin{pmatrix}
1\\[2pt]
\Delta \mathit{NA}_i
\end{pmatrix}
\big\{(\mathit{PM}_{i,1}-\mathit{PM}_{i,0})-\alpha_0-\beta_0\Delta \mathit{NA}_i\big\}
\right]=0,
\end{equation}
which matches the structure of \eqref{eq:downstream_model} with the observed covariates $W_i=\Delta \mathit{NA}_i$ and the latent target variables given by the true PM$_{2.5}$ levels $Z_i = (\mathit{PM}_{i,0},\mathit{PM}_{i,1})'$.

In the empirical setting of \textcite{sager_clean_2025}, the downstream dataset contains the treatment indicator $\Delta \mathit{NA}_i$ but lacks direct observations of true PM$_{2.5}$ at the tract-by-year level. Instead, the analysis relies on high-resolution historical PM$_{2.5}$ estimates constructed by \textcite{meng_estimated_2019}, which combine information from chemical transport modeling, satellite remote sensing, and ground-based monitoring. Letting $\widehat{\mathit{PM}}_{i,t}$ denote the published estimate for tract $i$ in year $t$, this setting maps cleanly into our framework by defining the ML-generated proxy as $\hat{Z}_i=(\widehat{\mathit{PM}}_{i,0},\widehat{\mathit{PM}}_{i,1})'$, which serves as the upstream prediction for the unobserved vector $Z_i$.

By comparing their predicted PM$_{2.5}$ levels to actual ground-based monitor measurements, \textcite{meng_estimated_2019} report that their estimates have a mean bias of $-0.33~\mu\mathrm{g}/\mathrm{m}^3$ and a root mean squared error (RMSE) of $1.94~\mu\mathrm{g}/\mathrm{m}^3$. Given that typical ambient PM$_{2.5}$ concentrations range from $2~\mu\mathrm{g}/\mathrm{m}^3$ to $20~\mu\mathrm{g}/\mathrm{m}^3$, these predictions, while highly informative, are subject to non-negligible prediction errors. Neglecting these errors by naively treating the estimated PM$_{2.5}$ levels as the ground truth has a risk of inducing bias and invalidating downstream statistical inference. Our partial identification framework explicitly accounts for these errors by leveraging the reported performance metrics to construct valid bounds on the structural parameters.
\end{example}

\begin{example}[Machine-learned measures as regressors]\label{ex:remote}
	\hspace{-1em}\footnote{This example is adapted from the ``remote work and wage inequality'' application in \textcite{battaglia_inference_2025}.}
Remote work has risen markedly since the onset of the COVID-19 pandemic \textcite{barrero_why_2021,aksoy_working_2022}. An important empirical question is how remote-work arrangements affect posted wages. Consider an applied researcher estimating the linear regression model
\begin{equation}\label{eq:remote_wage_equation}
	Y_i = \beta_0 Z_i + \alpha_0^\top C_i + \varepsilon_i,
\end{equation}
where $i$ indexes a job posting, $Y_i$ is the posted log wage, $Z_i$ is a latent binary indicator for whether the job allows remote work, and $C_i$ is a vector of observed controls (including an intercept). Assuming standard exogeneity, $\expt[\varepsilon_i\mid Z_i,C_i]=0$, the model implies the population moment conditions
\begin{equation}\label{eq:remote_downstream_model}
	\expt\!\left[
	\begin{pmatrix}
		Z_i\\[2pt]
		C_i
	\end{pmatrix}
	\big(Y_i-\beta_0 Z_i-\alpha_0^\top C_i\big)
	\right]=0.
\end{equation}
In many job-posting datasets, $Y_i$ and $C_i$ are directly observed, whereas the true remote-work status $Z_i$ is not explicitly recorded and must be inferred from the unstructured text of the job description. This environment maps cleanly into the analytical framework of \eqref{eq:downstream_model}, with the observed variables $W_i=(Y_i,C_i^\top)^\top$, the latent regressor $Z_i$, the unstructured inputs $X_i$ given by the job-description text, and the structural parameter vector $\theta_0=(\beta_0,\alpha_0^\top)^\top$.

To infer $Z_i$ from $X_i$, \textcite{hansen_remote_2023} train a text classifier $g(\cdot)$ that maps the job-posting text $X_i$ into a predicted remote-work label $\hat Z_i=g(X_i)$. They obtain human annotations from Amazon Mechanical Turk for a subset of postings and use them to fine-tune a pre-trained language model (DistilBERT; \textcite{sanh_distilbert_2020}). Using an auxiliary validation sample of human-labeled postings, they report an out-of-sample classification error rate of $2\%$. As a benchmark, remote job postings account for approximately $28\%$ of all postings in their sample. 

While $\hat{Z}_i$ is a highly precise estimate, the residual prediction error remains a first-order econometric concern. In this context, the unstructured job-posting text $X_i$ need not be exogenous with respect to the wage shock $\varepsilon_i$. Because job descriptions naturally contain rich information about firm characteristics and job amenities, $X_i$ is highly likely to be correlated with unobserved determinants of wages. This implies that the ML proxy $\hat Z_i=g(X_i)$ could be endogenous, even if the true status $Z_i$ is strictly exogenous. Consequently, a naive regression that directly substitutes $\hat Z_i$ for $Z_i$ in \eqref{eq:remote_wage_equation} suffers from non-classical measurement error and endogeneity bias. 

Our partial identification approach would circumvent this concern. The moment conditions in \eqref{eq:remote_downstream_model} are stated in terms of the \emph{true} regressor $Z_i$. By leveraging an auxiliary validation sample or valid summary performance measures, our procedure disciplines the relationship between $Z_i$ and $\hat{Z}_i$ without imposing exogeneity of $\hat{Z}_i$.
 Finally, if the exogeneity of the true $Z_i$ itself is in doubt, one can readily replace \eqref{eq:remote_downstream_model} with moments that exploit valid instrumental variables for $Z_i$.
\end{example}

\begin{example}[Media slant]\label{ex:media_slant}
Consider a setting in the spirit of \textcite{groseclose_measure_2005} and \textcite{gentzkow_what_2010}, in which newspaper ideology affects its consumer demand. To fix ideas, suppose newspaper demand is modeled using the BLP framework of \textcite{berry_automobile_1995}. A market $t$ may be defined by time, location, or both. Let $\mathcal{J}_t$ be the set of newspapers availabe in market $t$, and let $j=0$ denote the outside option. 
The utility that potential reader $i$ derives from choosing newspaper $j$ in market $t$ is
\begin{align*}
	u_{ijt} &= \alpha_i Z_{jt} + \beta_i^\top C_{jt} + \xi_{jt} + \epsilon_{ijt}, \qquad j\in \mathcal{J}_t,\\
	u_{i0t} &= 0,
\end{align*}
where $Z_{jt}$ denotes the ideological position, or slant index, of newspaper $j$ in market $t$. We normalize $Z_{jt}$ so that $Z_{jt}=0$ corresponds to the most liberal position and $Z_{jt}=1$ to the most conservative position. The vector $C_{jt}$ collects other observed newspaper characteristics (including price), $\xi_{jt}$ is an unobserved product characteristic, and $\epsilon_{ijt}$ is an i.i.d.\ Type~I extreme value shock.

Let $F(\alpha,\beta;\theta_0)$ denote the distribution of heterogeneous tastes $(\alpha_i,\beta_i)$, indexed by the parameter vector $\theta_0$. Writing $Z_t=(Z_{jt})_{j\in\mathcal{J}_t}$, $C_t=(C_{jt})_{j\in\mathcal{J}_t}$, and $\xi_t=(\xi_{jt})_{j\in\mathcal{J}_t}$, the implied market share for product $j\in\mathcal{J}_t$ is
\[
	\sigma_j(Z_t,C_t,\xi_t;\theta_0)
	=
	\int
	\frac{\exp\!\big(\alpha Z_{jt} + \beta^\top C_{jt} + \xi_{jt}\big)}
	{1+\sum_{j'\in\mathcal{J}_t}\exp\!\big(\alpha Z_{j't} + \beta^\top C_{j't} + \xi_{j't}\big)}
	\, \ud F(\alpha,\beta;\theta_0).
\]
Let $S_t=(S_{jt})_{j\in\mathcal{J}_t}$ denote observed market shares, and let $\xi_t(S_t,Z_t,C_t;\theta_0)$ denote the unique vector solving the standard BLP inversion equations
\[
	S_{jt}=\sigma_j(Z_t,C_t,\xi_t;\theta_0),\qquad j\in\mathcal{J}_t.
\]
Suppose, as in BLP, that there exists a valid instrument $V_{jt}$. Then the demand system implies the moment conditions
\begin{equation}\label{eq:BLP_demand_newsppaper}
	\expt\!\left[V_{jt}\xi_{jt}(S_t,Z_t,C_t;\theta_0)\right]=0,
	\qquad j\in\mathcal{J}_t.
\end{equation}

The key departure from the standard BLP setup is that the ideology index $Z_t$ is not directly observed. Following \textcite{groseclose_measure_2005} and \textcite{gentzkow_what_2010}, however, one can construct a proxy from newspaper text. Let $X_{jt}$ denote the text content of newspaper $j$ in market $t$. Using text-based ML methods, we can train a predictor of ideological position from an external labeled sample. For example, \textcite{groseclose_measure_2005} use speeches by members of the U.S.\ Congress as training data, taking each politician's Americans for Democratic Action (ADA) score---a measure of left--right ideology based on congressional voting records---as the label and the speech transcript as the text input. Their text-analysis procedure yields a prediction rule $\tilde{g}$ that maps text into an ideology score. Applying this rule to newspaper text $X_{jt}$ produces a proxy $\hat{Z}_{jt}$ for each latent $Z_{jt}$. Collecting the newspaper texts into $X_t=(X_{jt})_{j\in\mathcal{J}_t}$ and the corresponding predictions into a vector, we obtain
\begin{equation}\label{eq:ideology_prediction}
	\hat{Z}_t = g(X_t) = \big(\tilde{g}(X_{jt}) : j\in \mathcal{J}_t\big).
\end{equation}

In this application, \eqref{eq:BLP_demand_newsppaper} and \eqref{eq:ideology_prediction} play the roles of \eqref{eq:downstream_model} and \eqref{eq:upstream_model}, respectively, with $(V_t, S_t, C_t)$ being the covariates only observed in the downstream dataset.  The only complication is that the available validation data typically contain observations on $(Z_{jt},\hat{Z}_{jt})$, rather than joint observations on the entire vector $(Z_t,\hat{Z}_t)$. Because the prediction rule is trained to recover each newspaper's ideology score separately, rather than the full vector of scores within a market, validation is available componentwise rather than at the market level. Section \ref{sec:extensions} shows that this case can be accommodated as a straightforward extension of the baseline framework.
\end{example}

\section{Identification with Validation Data}\label{sec:validation_identification}
In this section, we study partial identification of the structural parameter $\theta_0$ when an auxiliary validation sample is available. We assume that this sample contains joint observations of the latent target variable $Z$, the ML-generated proxy $\hat{Z}$, and, potentially, a low-dimensional vector of observable characteristics $h(X)$ extracted from the raw predictors $X$. This data environment arises naturally when an upstream data provider releases a held-out sample of $(Z,X)$ together with the trained prediction rule, allowing the downstream researcher to compute both $\hat{Z} = g(X)$ and $h(X)$ and hence to obtain validation data on $(Z, \hat{Z},h(X))$.

Let $S \coloneqq h(X)$ denote a low-dimensional attribute of $X$. Its role is to capture observable strata along which the relationship between $Z$ and $\hat{Z}$ may vary. For example, in Example~\ref{ex:pm25}, $S$ might indicate the geographic region or time period associated with the remote-sensing data; in Example~\ref{ex:remote}, it might indicate the language or industry of the job posting. At one extreme, $S$ may simply be a constant and therefore carry no additional information. More informative choices of $S$, however, could tighten the identified set by exploiting variation in the conditional joint distribution of $(Z,\hat{Z})$ across subpopulations indexed by $S$.

Throughout the identification analysis, we treat the prediction rule $g(\cdot)$ and the feature map $h(\cdot)$ as fixed and deterministic functions. This treatment is justified under the standard assumption that the upstream training sample used to train $g(\cdot)$ is independent of both the validation sample and the downstream empirical sample. Under this assumption, the population analysis can be interpreted conditional on the realized $(g,h)$, so that the only randomness in $\hat{Z}=g(X)$ and $S=h(X)$ arises from the randomness in $X$.

A useful feature of our setup is that the validation sample need not contain the downstream covariates $W$. This is empirically important. Upstream data collection is typically organized around obtaining labeled pairs $(Z,X)$ for prediction, whereas the covariates $W$ are application-specific and often unavailable to the upstream provider. Moreover, from the downstream researcher's perspective, obtaining ground-truth measurements of $Z$ in the downstream sample is often costly or infeasible. We therefore assume the validation data to contain $(Z,\hat{Z},S)$ but not $W$.

\subsection{Compatibility and the identified set}\label{subsec:validation_idset}
The downstream and validation samples are informative about the same latent target only if they can be viewed as marginals of a common population distribution. We formalize this requirement as a compatibility condition. Let $F$ denote the population distribution of the downstream observables $(W,\hat Z,S)$, where $\hat Z=g(X)$ and $S=h(X)$ are computed using the fixed maps $(g,h)$. Let $G$ denote the population distribution of the validation observables $(Z,\hat Z,S)$.

\begin{assumption}[Compatibility of Marginal Distributions]\label{assu:consistent_validation}
	Let $H_0$ denote the true joint distribution of $(W,Z,\hat Z,S)$. Assume that (\emph{i}) the marginal distribution of $(W,\hat Z,S)$ under $H_0$ is $F$, and (\emph{ii}) the marginal distribution of $(Z,\hat Z,S)$ under $H_0$ is $G$.
\end{assumption}

Assumption~\ref{assu:consistent_validation} requires the downstream and validation samples to agree on the overlap variables $(\hat Z,S)$. In practice, however, a validation sample constructed as a held-out subset of an upstream training dataset may exhibit covariate shift relative to the downstream population of interest, in which case the implied distributions of $(\hat Z,S)$ need not coincide across the two samples. When the support of the validation sample covers that of the downstream sample, such discrepancies can often be addressed using standard reweighting methods, such as inverse probability weighting, to align the distribution of $X$ and, hence, that of $(\hat Z,S)$ in the validation sample with that in the downstream population. We abstract from this issue here to focus on the core identification problem, and leave a formal treatment of reweighting to future work.

Under Assumption~\ref{assu:consistent_validation}, the identified set for $\theta_0$ is defined as the set of parameters consistent with the observed marginal distributions and the population moment conditions. Formally, let $\mathcal{H}(F,G)$ denote the set of all joint distributions over $(W,Z,\hat Z,S)$ whose marginals on $(W,\hat Z,S)$ and $(Z,\hat Z,S)$ are $F$ and $G$, respectively.

\begin{definition}[Identified Set with Validation Data]\label{def:id_set_validation}
	The identified set $\Theta_I(F, G)$ is the set of parameters $\theta \in \Theta$ for which there exists a joint distribution $H \in \mathcal{H}(F, G)$ such that the structural moment conditions are satisfied: $\expt_{H}[q(W, Z;\theta)] = 0$, i.e.,
	\[
	\Theta_I(F,G)
	\coloneqq
	\left\{
	\theta\in\Theta :
	\exists\, H\in\mathcal{H}(F,G)
	\text{ such that }
	\expt_{H}\!\left[q(W,Z;\theta)\right]=0
	\right\}.
	\]
\end{definition}

\begin{remark}[Relation to Data Combination]\label{rem:data_combination}
	If the full predictor vector $X$ were low-dimensional and observed in both samples, one could instead link the downstream and validation samples directly through $X$, as in the classical data combination literature \autocite[e.g.][]{cross_regressions_2002,ridder_chapter_2007}. Formally, one would impose compatibility at the level of the full distributions of $(W,X)$ and $(Z,X)$. The identified set implied by this stronger compatibility condition would generally be weakly tighter than the set in Definition~\ref{def:id_set_validation}, as it exploits the full information contained in $X$. However, when $X$ is high-dimensional or unstructured (e.g., text or images), nonparametrically modeling and transporting the joint distribution of $(W, X)$ and $(Z, W)$ across samples is practically impossible due to the curse of dimensionality. Our approach circumvents this by treating the $(g,h)$ as dimension-reduction devices, adapting the data combination perspective to settings where only low-dimensional summary features are computationally tractable.
\end{remark}

\subsection{an unconditional optimal transport characterization}\label{sec:uncond_OT_characterization}
Characterizing the identified set $\Theta_I(F, G)$ computationally is a non-trivial task. Recently, \textcite{fan_partial_2025} proposed a sharp identification method that relies on \emph{conditional} optimal transport. However, when the proxy $\hat{Z}$ or the stratifying variables $S$ contain continuous components, their approach requires solving a continuum of optimal transport problems---specifically, one for every realization of $(\hat{Z}, S)$ in the joint support. Solving those conditional optimal transport problems or approximating them by fine discretization could be burdensome in practice.

We develop an alternative identification method that reduces the characterization to an \emph{unconditional} optimal transport problem while preserving sharpness. The key device, adapted from \textcite{li_finite_2025}, is to introduce auxiliary copies of the overlap variables and move the exact-matching restrictions from the coupling class into moment conditions. Specifically, let $(\hat{Z}',S')$ be auxiliary random vectors and consider the augmented random vector
\begin{equation*}
	(W,Z,\hat{Z},S,\hat{Z}',S').
\end{equation*}
For the purpose of this representation, we view the downstream sample as providing observations on $(W,\hat{Z},S)$ with distribution $F$, and the validation sample as providing observations on $(Z,\hat{Z}',S')$ with distribution $G$. Let $\mathcal{H}'(F,G)$ denote the \emph{unconditional} Fr\'{e}chet class of all joint distributions $H'$ over $(W,Z,\hat{Z},S,\hat{Z}',S')$ whose marginal distribution on $(W,\hat{Z},S)$ is $F$ and whose marginal distribution on $(Z,\hat{Z}',S')$ is $G$.

The advantage of this construction is that the almost-sure restrictions $\hat{Z}=\hat{Z}'$ and $S=S'$, which are implicit in $\mathcal{H}(F,G)$ and give rise to a computationally demanding conditional coupling problem, are no longer imposed directly in $\mathcal{H}'(F,G)$. Instead, we reintroduce these exact-matching requirements through additional moment restrictions. This yields the following equivalent representation of the identified set.

\begin{definition}[Equivalent Representation of Identified Set]\label{def:alias_id_set}
Define $\Theta'_I(F, G)$ as the set of parameters $\theta\in \Theta$ for which there exists an augmented joint distribution $H'\in \mathcal{H}'(F, G)$ such that 
\begin{align}
	\expt_{H'}[q(W, Z; \theta)] &= 0, \label{eq:alias_structural}\\
	\expt_{H'}\big[|\hat{Z}_j - \hat{Z}'_j|\big] &= 0, \quad \text{for all } j=1,\dots,d_{\hat{Z}}, \label{eq:alias_Z}\\
	\expt_{H'}\big[|S_k - S'_k | \big] &= 0, \quad \text{for all } k=1,\dots,d_S, \label{eq:alias_S}
\end{align}
where $d_{\hat{Z}}$ and $d_S$ are the dimensions of $\hat{Z}$ and $S$, respectively.
\end{definition}

Because equations \eqref{eq:alias_Z} and \eqref{eq:alias_S} force $\hat{Z}' = \hat{Z}$ and $S' = S$ almost surely under $H'$, any valid distribution in $\mathcal{H}'(F, G)$ that satisfies \eqref{eq:alias_Z} and \eqref{eq:alias_S} collapses to an element in $\mathcal{H}(F, G)$, establishing that $\Theta'_I(F, G) = \Theta_I(F, G)$. While theoretically equivalent, this augmented representation successfully isolates the binding coupling constraints into the moment restrictions, allowing us to characterize the identified set using purely unconditional optimal transport.

To state this formally, let $\tilde{q}(W, Z, \hat{Z}, S, \hat{Z}', S'; \theta)$ denote the augmented moment vector stacking the structural moments and the squared proxy differences:
\begin{equation}\label{eq:stacked_q}
	\tilde{q}(W, Z, \hat{Z}, S, \hat{Z}', S';\theta) 
	\coloneqq 
	\begin{pmatrix}
		q(W, Z;\theta) \\[4pt]
		|\hat{Z}_1 - \hat{Z}'_1| \\
		\vdots \\
		|\hat{Z}_{d_{\hat{Z}}} - \hat{Z}'_{d_{\hat{Z}}}| \\[4pt]
		|S_1 - S'_1| \\
		\vdots \\
		|S_{d_S} - S'_{d_S}|
	\end{pmatrix}.
\end{equation}
The system of equations \eqref{eq:alias_structural}--\eqref{eq:alias_S} can be written compactly as the single vector moment condition $\expt_{H'}[\tilde{q}(\cdot; \theta)] = 0$. 


Let $d_q$ be the dimension of $q$, and let $\mathbb{B} \subset \real^{d_q + d_{\hat{Z}} + d_S}$ be any compact convex set containing an open neighborhood of the origin (for example, the closed unit ball $\mathbb{B} = \{\lambda : \norm{\lambda} \le 1\}$). By standard properties of moments, $\expt_{H'}[\tilde{q}(\cdot; \theta)] = 0$ holds if and only if
\begin{equation*}
	\sup_{\lambda \in \mathbb{B}} \, \expt_{H'}\big[\lambda^{\top}\tilde{q}(W, Z, \hat{Z}, S, \hat{Z}', S'; \theta)\big] = 0.
\end{equation*}
Therefore, for any valid parameter $\theta \in \Theta'_I(F, G)$, there must exist an $H' \in \mathcal{H}'(F, G)$ such that the supremum over $\lambda$ is exactly zero. Consequently, taking the minimum over all joint distributions in $\mathcal{H}'(F, G)$ yields:
\begin{equation*}
	\min_{H'\in \mathcal{H}'(F, G)} \sup_{\lambda \in \mathbb{B}} \, \expt_{H'}\big[\lambda^{\top}\tilde{q}(W, Z, \hat{Z}, S, \hat{Z}', S'; \theta)\big] \le 0.
\end{equation*}
Because the supremum of a minimum is always bounded above by the minimum of a supremum (i.e., the weak duality), swapping the order of the operators therefore establishes a necessary condition for identification:
\begin{equation*}
	\sup_{\lambda \in \mathbb{B}} \min_{H'\in \mathcal{H}'(F, G)} \expt_{H'}\big[\lambda^{\top}\tilde{q}(W, Z, \hat{Z}, S, \hat{Z}', S'; \theta)\big] \le 0.
\end{equation*}
Our main identification theorem establishes that this relationship is, in fact, not only necessary but also sufficient under mild regularity conditions.  

Let $\mathcal{X}_d$ and $\mathcal{X}_v$ denote the supports of the downstream observables $(W, \hat{Z}, S)$ and the validation observables $(Z, \hat{Z}, S)$, respectively. Let $\mathcal{W}$ and $\mathcal{Z}$ denote the support of $W$ and $Z$ respectively.

\begin{assumption}[Regularity conditions]\label{assu:regularity_minimax}
	For any parameter $\theta \in \Theta$, assume that
	\begin{enumerate}
		\item\label{enu:polish} The supports $\mathcal{X}_d$ and $\mathcal{X}_v$ are closed subsets of finite-dimensional Euclidean spaces.
		\item\label{enu:continous} The structural moment function $q(w, z; \theta)$ is continuous jointly in $(w, z)$.
		\item\label{enu:fintie_variance}  $\hat{Z}$ and $S$ have finite first-order moments under both distributions $F$ and $G$. That is, $\expt_F\big[\|\hat{Z}\| + \norm{S}\big] < \infty$ and $\expt_G\big[\|\hat{Z}\| + \norm{S}\big] < \infty$. 
		\item \label{enu:envelop}There exist non-negative continuous functions $b_1 \in \mathcal{L}^1(F)$ and $b_2 \in \mathcal{L}^1(G)$ such that the structural moment function is bounded by an additively separable envelope:
		\[
			\|q(w, z; \theta)\| \le b_1(w) + b_2(z) \quad \text{for all } w \in \mathcal{W} \text{ and } z \in \mathcal{Z}.
		\]
	\end{enumerate}
\end{assumption}
Assumption \ref{assu:regularity_minimax}\ref{enu:polish} ensures $\mathcal{X}_d$ and $\mathcal{X}_v$ are polish spaces so that 
probability measures $F$ and $G$ are Radon measures. Assumption \ref{assu:regularity_minimax}\ref{enu:continous} imposes continuity, which automatically holds when $(w, z)$ are discrete. Finally, Assumptions~\ref{assu:regularity_minimax}\ref{enu:fintie_variance} and \ref{enu:envelop} provide the necessary envelope bounds to guarantee that the expected moment conditions are finite and continuous with respect to the weak topology on the space of joint distributions.

Under these regularity conditions, we can establish the following result, the proof of which is in Appendix \ref{sec:proof_id_validation}.

\begin{theorem}\label{thm:id_validation}
	Suppose Assumptions~\ref{assu:consistent_validation} and~\ref{assu:regularity_minimax} hold. Then $\theta \in \Theta_I(F,G)$ if and only if
	\begin{equation}\label{eq:maxmin_characterization}
		\max_{\lambda \in \mathbb{B}} \min_{H'\in \mathcal{H}'(F, G)} \expt_{H'}\big[\lambda^{\top}\tilde{q}(W, Z, \hat{Z}, S, \hat{Z}', S'; \theta)\big] \le 0.
	\end{equation}
	Moreover, the maximum and minimum in \eqref{eq:maxmin_characterization} are both attained.
\end{theorem}

Theorem \ref{thm:id_validation} provides a sharp characterization of the identified set $\Theta_I(F, G)$. For a fixed pair $(\theta,\lambda)$, the inner minimization problem in Theorem~\ref{thm:id_validation} is a standard unconditional optimal transport problem between the two observed marginal distributions $F$ and $G$, with the cost function given by
\begin{equation}\label{eq:cost_function}
	c_{\theta,\lambda}(w,\hat{z},s; z,\hat{z}',s')
	\coloneqq \lambda^{\top}\tilde{q}(w,z,\hat{z},s,\hat{z}',s';\theta).
\end{equation}
As we demonstrate in Section~\ref{sec:validation_inference}, by exploiting the Kantorovich duality of this optimal transport problem, the resulting max-min characterization can be implemented as a convex optimization problem. This greatly simplifies computation and provides a tractable foundation for estimation and inference.

The identification analysis we developed here changes the role of the proxy $\hat{Z}$. Rather than treating $\hat{Z}$ primarily as a plug-in substitute for $Z$, the theorem treats it as a bridge linking the downstream and validation samples. Accordingly, the analysis shifts attention away from the approximation error $Z-\hat{Z}$ itself and toward the compatibility of the two samples through the common variables $(\hat{Z},S)$. This perspective is what allows our results to avoid structural assumptions on the prediction error and to accommodate highly general upstream prediction rules. For the applied researcher, the implication is straightforward: one need not restrict attention to ML methods for which a full theoretical analysis is available, but may instead choose the method that is most effective in the empirical setting at hand.

This perspective also suggests a different criterion for evaluating upstream ML methods. In standard ML practice, predictors are typically judged by how accurately they approximate $Z$. From the standpoint of downstream identification, however, a useful predictor is one that preserves as much of the information in $X$ about $Z$ as possible. In the extreme case in which the conditional distribution of $Z$ given $X$ depends on $X$ only through $(\hat{Z},S)$, the pair $(\hat{Z},S)$ is as informative as $X$ itself for the downstream problem: replacing $X$ with $(\hat{Z},S)$ entails no loss of information about $Z$. A predictor $\hat{Z}$ that closely approximates $Z$ is therefore sufficient, but not necessary, for $\hat{Z}$ to serve as an effective dimension-reduction device.

An immediate implication of this viewpoint is that $Z$ and $\hat{Z}$ need not take values in the same space. For example, when $Z$ is a binary label, $\hat{Z}$ may be the predicted probability that $Z=1$ rather than a hard binary classification. Likewise, when $Z$ is a multinomial group label, $\hat{Z}$ may record the model's top two predicted classes, or even a vector of class scores, rather than only its single best guess. These richer forms of $\hat{Z}$ retain information that would be lost if the proxy were collapsed to a plug-in estimate of $Z$, and Theorem~\ref{thm:id_validation} shows how such information can be incorporated into the downstream analysis.

More generally, viewing $\hat{Z}$ as a dimension-reduction device suggests a natural way to combine predictions from multiple ML methods. Suppose, for example, that two ML prediction rules produce proxies $\hat{Z}_1$ and $\hat{Z}_2$ respectively. One can then form the proxy $\hat{Z}$ as $\hat{Z} = (\hat{Z}_1,\hat{Z}_2)$, and carry out the analysis based on this combined proxy. This can be accommodated  because $\hat{Z}$ need not reside in the same space as $Z$. In this way, our framework can exploit complementary information from multiple prediction methods without requiring the researcher to commit to a single proxy \emph{ex ante}.

\begin{remark}
	More broadly, Theorem~\ref{thm:id_validation} delivers a new identification result for classical data combination problems based on unconditional optimal transport. It provides a sharp and computationally tractable characterization for general moment models in which the variables entering the moment restrictions are not jointly observed in a single dataset, but instead are distributed across two distinct samples. As such, the result is of independent theoretical interest and contributes more generally to the data combination literature. In particular, it complements recent work that uses optimal transport methods to study identification and partial identification in related data combination settings. See, for example, \textcite{hwang_bounding_2025}, \textcite{dhaultfoeuille_partially_2025} and \textcite{fan_partial_2025}.
\end{remark}

\section{Inference with Validation Data}\label{sec:validation_inference}

This section develops inference for testing whether a candidate parameter value $\theta$ belongs to the identified set $\Theta_I(F,G)$, building on the characterization in Theorem~\ref{thm:id_validation}. This inference problem is nonstandard for two reasons. First, the criterion is defined through an optimal-transport problem, and its sample analog exhibits nonstandard asymptotic behavior. Second, the identified-set characterization involves an outer max--min operation, which further complicates inference.

We address these challenges in two steps. First, we rewrite the max--min problem in \eqref{eq:maxmin_characterization} using the Kantorovich dual representation, so that both $\lambda$ and the dual functions are jointly determined by a convex optimization problem. We then approximate the infinite-dimensional class of dual functions by sieve spaces, thereby reducing the problem to a finite-dimensional convex program. Details are provided in Subsection~\ref{sec:convex_prog_sieve_approx}. Building on this representation, we develop a testing procedure based on sample splitting and cross-fitting; see Subsection~\ref{sec:cross_fit_inference}. A key practical advantage of this approach is that the test can be calibrated using an asymptotically pivotal upper bound on the distribution of the test statistic, so that valid critical values can be obtained from standard reference distributions without bootstrap or other simulation-based methods.

\subsection{Convex programming and sieve approximation}\label{sec:convex_prog_sieve_approx}
Recall that $\mathcal{X}_d$ and $\mathcal{X}_v$ denote the supports of $(W,\hat{Z},S)$ and $(Z,\hat{Z}',S')$, respectively. Let $\mathcal{X}_{w}$ and $\mathcal{X}_{z}$ denote the suport of $W$ and $Z$ respectively. And, let $\mathcal{C}(\mathcal{X}_v)$ denote the class of real-valued continuous functions on $\mathcal{X}_v$ that are integrable with respect to $G$.  The following lemma shows that, under Assumption~\ref{assu:regularity_minimax}, the max--min criterion in \eqref{eq:maxmin_characterization} admits a Kantorovich dual representation.

\begin{lemma}[Kantorovich Duality]\label{lem:kantorovich_duality}
Suppose Assumption~\ref{assu:regularity_minimax} holds, and fix an arbitrary $\theta\in\Theta$. Then
\begin{equation}\label{eq:kantorovich_duality_application}
\max_{\lambda \in \mathbb{B}} \min_{H'\in \mathcal{H}'(F, G)} \expt_{H'}\big[\lambda^{\top}\tilde{q}(W, Z, \hat{Z}, S, \hat{Z}', S'; \theta)\big]
= \mathscr D(F,G;\theta),
\end{equation}
where
\begin{multline}\label{eq:convex_characterization}
\mathscr D(F,G;\theta)
\coloneqq
\sup_{\lambda\in\mathbb B,\ \psi\in \mathcal{C}(\mathcal{X}_v)} \Bigg\{
\expt_F\!\left[
\inf_{(z,\hat z',s')\in\mathcal X_v}
\left\{
\lambda^{\top}\tilde q(W,z,\hat Z,S,\hat z',s';\theta)
-
\psi(z,\hat z',s')
\right\}
\right] \\
+
\expt_G\!\left[\psi(Z,\hat Z',S')\right]
\Bigg\}.
\end{multline}
Moreover, if moment function $q$ is uniformly continuous in $\mathcal{X}_w\times \mathcal{X}_z$, then the supremum in \eqref{eq:convex_characterization} is attained at some $\lambda^*\in \mathbb{B}$ and uniform continuous $\psi^*\in \mathcal{C}(\mathcal{X}_v)$.
\end{lemma}

Combining Lemma~\ref{lem:kantorovich_duality} with Theorem~\ref{thm:id_validation} yields the following equivalent characterization of the identified set:
\[
\theta\in\Theta_I(F,G)
\qquad\text{if and only if}\qquad
\mathscr D(F,G;\theta)\le 0.
\]
Compared with \eqref{eq:maxmin_characterization}, the dual representation in \eqref{eq:convex_characterization} is computationally more convenient, provided that the inner minimization over $\mathcal X_v$ can be evaluated efficiently; see Remark~\ref{rmk:computation_dual}. In particular, it transforms the original max--min problem into a joint maximization over $(\lambda,\psi)$. Moreover, the objective function in \eqref{eq:convex_characterization} is jointly concave in $(\lambda,\psi)$, since the term inside the $\inf\{\dots\}$ is affine in $(\lambda,\psi)$ for each $(z,\hat z',s')$, and the pointwise infimum of affine functions is concave. 

To obtain a computationally implementable finite-dimensional approximation, we replace the infinite-dimensional class $\mathcal{C}(\mathcal{X}_v)$ with a sieve space. Let $\mathcal{S}_K(\mathcal{X}_v)$ be a linear sieve space spanned by $K$ known basis functions:
\[
\mathcal{S}_K(\mathcal{X}_v)
\coloneqq
\left\{
\psi:\ \psi(z,\hat z',s') = \beta^{\top}\varphi(z,\hat z',s')
\text{ for some } \beta\in\mathbb{R}^K
\right\},
\]
where $\varphi(z,\hat z',s') \coloneqq \big(\varphi_1(z,\hat z',s'),\dots,\varphi_K(z,\hat z',s')\big)^{\top}$ is a vector of known basis functions. Restricting the dual function $\psi$ in \eqref{eq:convex_characterization} to this finite-dimensional sieve space yields the sieve approximation
\begin{multline}\label{eq:sieve_dual_value}
\mathscr D_K(F,G;\theta)
\coloneqq
\sup_{\lambda\in\mathbb B,\ \beta\in\mathbb R^K}
\Bigg\{
\expt_F\!\left[
\inf_{(z,\hat z',s')\in\mathcal X_v}
\left\{
\lambda^{\top}\tilde q(W,z,\hat Z,S,\hat z',s';\theta)
-
\beta^{\top}\varphi(z,\hat z',s')
\right\}
\right] \\
+
\expt_G\!\left[\beta^{\top}\varphi(Z,\hat Z',S')\right]
\Bigg\}.
\end{multline}
Equation \eqref{eq:sieve_dual_value} reduces the problem to a finite-dimensional concave maximization problem over $(\lambda,\beta)$, which is the key to the computational tractability of our inference procedure. The next lemma establishes the relation between $\mathscr{D}(F, G;\theta)$ and $\mathscr{D}_K(F, G;\theta)$.

\begin{assumption}\label{assu:sieve_approximation}
Assume the sieve space $\mathcal{S}_K(\mathcal{X}_v)$ can approximate any continuous function on $\mathcal{X}_v$ arbitrarily well in $\mathcal{L}^{\infty}$ norm as $K\to \infty$. That is, for any $\psi^* \in \mathcal{X}_v$, 
\begin{equation*}
\lim_{K\to \infty} \inf_{\psi\in \mathcal{S}_K(\mathcal{X}_v)} \norm{\psi^* - \psi}_{\infty} = 0.
\end{equation*}
\end{assumption}

\begin{lemma}[Sieve approximation error]\label{lem:sieve_approx}
Assume that each basis function is continuous and integrable with respect to $G$, and that Assumptions~\ref{assu:regularity_minimax} hold. Fix an arbitrary $\theta\in \Theta$. Then,
\begin{equation}\label{eq:sieve_validity}
\mathscr D_K(F,G;\theta)\le \mathscr D(F,G;\theta).
\end{equation}
Moreover, under Assumption \ref{assu:sieve_approximation},
\begin{equation}\label{eq:sieve_power}
\lim_{K\to \infty}\mathscr D_K(F,G;\theta) = \mathscr D(F,G;\theta)
\end{equation}
\end{lemma}

Lemma~\ref{lem:sieve_approx} has two important implications for inference. First, \eqref{eq:sieve_validity} shows that the computationally tractable condition
\begin{equation}\label{eq:convex_necessary_condition}
\mathscr{D}_K(F,G;\theta) \le 0
\end{equation}
is a valid \emph{necessary} condition for $\theta\in\Theta_I(F,G)$. Our inference procedure therefore proceeds by testing \eqref{eq:convex_necessary_condition}. Because \eqref{eq:convex_necessary_condition} may be weaker than \eqref{eq:convex_characterization} for fixed $K$, the resulting test is valid but potentially conservative. Second, \eqref{eq:sieve_power} shows that the gap between the finite-dimensional criterion $\mathscr{D}_K(F,G;\theta)$ and the population criterion $\mathscr{D}(F,G;\theta)$ vanishes as $K\to\infty$, provided that the underlying sieve space can approximate any continuous function arbitrarily well.

\begin{remark}\label{rmk:computation_dual}
A computational bottleneck in evaluating $\mathscr D_K(F,G;\theta)$ is the inner minimization problem over $\mathcal X_v$. When $(z,\hat z',s')$ has a discrete and finite support, this problem is straightforward to solve. By contrast, when $(z,\hat z',s')$ includes continuous components, evaluating the inner $\inf$ can be computationally costly, depending on the application. This consideration provides an additional motivation for the resampling-free inference procedure developed below: by avoiding bootstrap or other resampling methods, it also avoids repeatedly solving the inner minimization problem across resampled observations.

In some applications, it may be preferable to work with the alternative one-sided dual representation that places the dual potential, and hence the sieve approximation, on functions of $(w,\hat z,s)\in\mathcal X_d$. This formulation is especially attractive when the support of $(W,\hat Z,S)$ is discrete, or more generally when optimization over $\mathcal X_d$ is computationally easier than optimization over $\mathcal X_v$. In such cases, one can equivalently work with the alternative one-sided dual representation
\begin{multline}
\mathscr D^\dagger(F,G;\theta)
\coloneqq
\sup_{\lambda\in\mathbb B,\ \phi\in \mathcal{C}(\mathcal X_d)} \Bigg\{
\expt_F\!\left[\phi(W,\hat Z,S)\right]
\\
+
\expt_G\!\left[
\inf_{(w,\hat z,s)\in\mathcal X_d}
\left\{
\lambda^{\top}\tilde q(w,Z,\hat z,s,\hat Z',S';\theta)
-
\phi(w,\hat z,s)
\right\}
\right]
\Bigg\}
\le 0.
\end{multline}
Approximating $\phi$ by a sieve space then yields
\begin{multline}\label{eq:sieve_dual_value_alt}
\mathscr D^\dagger_K(F,G;\theta)
\coloneqq
\sup_{\lambda\in\mathbb B,\ \beta\in\mathbb R^K}
\Bigg\{
\expt_F\!\left[\beta^{\top}\varphi(W,\hat Z,S)\right]
\\
+
\expt_G\!\left[
\inf_{(w,\hat z,s)\in\mathcal X_d}
\Big(
\lambda^{\top}\tilde q(w,Z,\hat z,s,\hat Z',S';\theta)
-
\beta^{\top}\varphi(w,\hat z,s)
\Big)
\right]
\Bigg\},
\end{multline}
where the sieve basis functions $\varphi$ are now defined on $\mathcal X_d$. Thus, the choice between the two one-sided representations should be guided by the computational structure of the application: depending on whether $\mathcal X_d$ or $\mathcal X_v$ is easier to optimize over, one formulation may be more tractable than the other.
\end{remark}

\subsection{Sample Splitting and Cross-Fitted Inference}\label{sec:cross_fit_inference}
One possible route to inference based on \eqref{eq:convex_necessary_condition} would be to study the empirical process associated with the objective function in $\mathscr D_K(F,G;\theta)$ and establish a uniform Gaussian approximation over $(\lambda,\beta)$. Such an approach would rely on Donsker-type conditions for the relevant class of objective functions. While feasible in principle, it has two drawbacks. First, the required empirical-process conditions may be restrictive and difficult to verify. Second, the resulting limiting distribution is generally non-pivotal, so critical values typically must be obtained by bootstrap or other simulation-based methods, which can be computationally burdensome in practice.

We therefore develop an alternative inference procedure based on sample splitting and cross-fitting. The primary advantage of this approach is that it places minimal regularity assumptions on the moment function $q(\cdot)$ and yields a test statistic that can be bounded using a pivotal, analytical critical value, thereby avoiding bootstrapping or other simulation methods. The theoretical trade-off is that, because the exact joint asymptotic distribution of the cross-fitted statistics remains unknown, the critical value must account for the least favorable dependence structure. While this minimax bounding renders the resulting inference procedure potentially conservative in finite samples, it guarantees correct asymptotic size control without sacrificing computational tractability.

Assume that, conditional on the trained prediction rule $g$, the downstream sample
\[
\{(W_i,\hat Z_i,S_i): i=1,\ldots,n_d\}
\]
is i.i.d.\ from $F$, the validation sample
\[
\{(Z_j,\hat Z'_j,S'_j): j=1,\ldots,n_v\}
\]
is i.i.d.\ from $G$, and the two samples are mutually independent. 
Define 
\begin{equation*}
n \coloneqq n_d+n_v,\ \underline{n} \coloneqq \min(n_d, n_v)
\end{equation*}
Randomly partition each sample into two approximately equal folds, indexed by $m\in\{1,2\}$. Let $\mathcal I^d_m$ and $\mathcal I^v_m$ denote the index sets for fold $m$ in the downstream and validation samples, respectively.

For an arbitrary candidate parameter $\theta \in \Theta$, fix a sequence of basis dimensions $K_n\to\infty$ and a sequence of bounds $c_n\to\infty$. The proposed inference procedure consists of the following steps:

\vspace{0.5em}
\noindent\textbf{Step 1.} Using only the downstream observations in fold $1$, compute an estimator $\hat{\Omega}_{1,n}$ of the covariance matrix of the vector $q(W, \hat{Z};\theta)$ and an estimator $\hat{\Lambda}_{1,n}$ of the covaraince matrix of the vector $(\hat{Z}^\top, S^\top)^\top$.  Define
\begin{equation*}
\hat{\Sigma}_{1,n} = \begin{pmatrix}
	\hat{\Omega}_{1,n} &\\
	 & \hat{\Lambda}_{1,n}
\end{pmatrix}
\end{equation*}
and
\[
\mathbb B_{1,n}
\coloneqq
\bigl\{\hat\Sigma_{1,n}^{-1/2}\gamma:\ \|\gamma\|\le 1\bigr\},
\]
where $\hat\Sigma_{1,n}^{-1/2}$ denotes an inverse square root (or generalized inverse square root, if needed) for $\hat{\Sigma}_{1,n}$. This step provides an empirical normalization of the components of the augmented moment vector $\tilde{q}$.

\vspace{0.5em}
\noindent\textbf{Step 2.} Using only the downstream and validation observations in fold $1$, solve the empirical analogue of \eqref{eq:sieve_dual_value}:
\begin{multline}\label{eq:empirical_argmax}
(\hat\lambda_{1,n},\hat\beta_{1,n})
\in
\argmax_{\lambda\in\mathbb B_{1,n},\ \beta\in\mathbb C_n}
\Bigg\{
\frac{1}{|\mathcal I^d_1|}\sum_{i\in\mathcal I^d_1}
\inf_{(z,\hat z',s')\in\mathcal X_v}
\Big[
\lambda^{\top}\tilde q(W_i,z,\hat Z_i,S_i,\hat z',s';\theta)
-
\beta^{\top}\varphi(z,\hat z',s')
\Big]
\\
+
\frac{1}{|\mathcal I^v_1|}\sum_{j\in\mathcal I^v_1}
\beta^{\top}\varphi(Z_j,\hat Z'_j,S'_j)
\Bigg\},
\end{multline}
where $\mathbb C_n \coloneqq \{\beta\in\mathbb R^{K_n}:\ \|\beta\|\le c_n\}$ and $\varphi(\cdot)$ denotes the vector of $K_n$ sieve basis functions defined on $\mathcal X_v$. Restricting the coefficients to the compact set $\mathbb C_n$ explicitly ensures a finite empirical maximizer for $\hat{\beta}_{1,n}$ always exists and controls the complexity of the sieve approximation.

\vspace{0.5em}
\noindent\textbf{Step 3.} Hold $(\hat\lambda_{1,n},\hat\beta_{1,n})$ fixed and evaluate the criterion on fold $2$. Define
\begin{align*}
\hat u_i
&\coloneqq
\inf_{(z,\hat z',s')\in\mathcal X_v}
\Big[
\hat\lambda_{1,n}^{\top}\tilde q(W_i,z,\hat Z_i,S_i,\hat z',s';\theta)
-
\hat\beta_{1,n}^{\top}\varphi(z,\hat z',s')
\Big],
\qquad i\in\mathcal I^d_2,
\\
\hat v_j
&\coloneqq
\hat\beta_{1,n}^{\top}\varphi(Z_j,\hat Z'_j,S'_j),
\qquad j\in\mathcal I^v_2.
\end{align*}
The resulting held-out estimate of the sieve dual criterion is
\[
\widehat{\mathscr D}_2(\theta)
\coloneqq
\frac{1}{|\mathcal I^d_2|}\sum_{i\in\mathcal I^d_2}\hat u_i
+
\frac{1}{|\mathcal I^v_2|}\sum_{j\in\mathcal I^v_2}\hat v_j.
\]

\vspace{0.5em}
\noindent\textbf{Step 4.} Conditional on fold $1$, the quantity $\widehat{\mathscr D}_2(\theta)$ is the sum of two independent sample averages, one from the downstream sample and one from the validation sample. Its conditional variance can therefore be estimated by
\[
\hat V_2(\theta)
\coloneqq
\frac{\hat\sigma_{u,2}^2}{|\mathcal I^d_2|}
+
\frac{\hat\sigma_{v,2}^2}{|\mathcal I^v_2|},
\]
where $\hat\sigma_{u,2}^2$ and $\hat\sigma_{v,2}^2$ are the sample variances of $\{\hat u_i:i\in\mathcal I^d_2\}$ and $\{\hat v_j:j\in\mathcal I^v_2\}$, respectively. Since $\hat V_2(\theta)$ shrinks asymptotically at rate $1/\underline{n}$, $\widehat{\mathscr D}_2(\theta)$ should be scaled by $\sqrt{\underline{n}}$. To guard against degeneracy when $(\hat\lambda_{1,n},\hat\beta_{1,n})=(0,0)$, let $\epsilon>0$ be a small fixed constant. The fold-2 test statistic is then
\begin{equation}\label{eq:T2_stat}
T_2(\theta)
\coloneqq
\frac{\sqrt{\underline{n}}\widehat{\mathscr D}_2(\theta)}
{\sqrt{\max\{\underline{n}\hat V_2(\theta),\epsilon\}}}.
\end{equation}

\vspace{0.5em}
\noindent\textbf{Step 5.} Repeat Steps 1--4 after swapping the roles of the two folds. This yields a second statistic,
\begin{equation}\label{eq:T1_stat}
	T_1(\theta)
	\coloneqq
	\frac{\sqrt{\underline{n}}\widehat{\mathscr D}_1(\theta)}
	{\sqrt{\max\{\underline{n}\hat V_1(\theta),\epsilon\}}},
\end{equation}
where $\widehat{\mathscr D}_1(\theta)$ and $\hat V_1(\theta)$ are computed on fold $1$ using the optimizer $(\hat\lambda_{2,n},\hat\beta_{2,n})$ obtained from fold $2$. We then aggregate the two fold-specific statistics by
\[
T(\theta)\coloneqq \max\{T_1(\theta),T_2(\theta)\}.
\]
The associated $p$-value is
\[
p(\theta)
\coloneqq
\min\bigl\{1,\ 2\bigl[1-\Phi\bigl(T(\theta)\bigr)\bigr]\bigr\},
\]
where $\Phi(\cdot)$ denotes the standard normal distribution function. We reject the null hypothesis
\[
H_0:\ \theta\in\Theta_I(F,G)
\]
at nominal significance level $\alpha$ whenever $p(\theta)<\alpha$.

Some remarks are in order.

\begin{remark}
The covariance standardization in Step 1 provides a data-driven normalization for the components of $\tilde{q}$.
 This normalization is not ideal for two reasons. First, the natural target would be the covariance matrix of $q(W,Z;\theta)$, but this object is not directly estimable because $Z$ is not observed in the downstream sample. We therefore replace it with the covariance of $q(W,\hat Z;\theta)$. Since $\hat Z$ is only a proxy for $Z$, this choice may affect efficiency, but it does not affect the validity of the procedure. 
Second, it is not theoretically obvious what the optimal weighting scheme should be for the exact-matching penalties $|\hat{Z} - \hat{Z}'|$ and $|S - S'|$ in $\tilde{q}$. Standardizing these via the empirical covariance of $(\hat{Z}, S)$ provides a heuristic that performs well in simulations. 
Finding optimal weights remains a task for future research.
\end{remark}

\begin{remark}
Ideally, one would characterize the joint asymptotic distribution of $(T_1(\theta),T_2(\theta))$ and then derive the limiting distribution of $T(\theta)$. Without stronger assumptions, however, we only obtain the marginal asymptotic distribution of each fold-specific statistic. The difficulty is that the fold-specific optimizers $(\hat\lambda_{1,n},\hat\beta_{1,n})$ and $(\hat\lambda_{2,n},\hat\beta_{2,n})$ need not converge to a common limit under the null. This is in contrast to standard double/debiased machine-learning settings, where the corresponding estimators converge to a unique probability limit and the leading terms of the fold-specific statistics depend only on the observations in their respective evaluation folds, which in turn yields asymptotic independence. In our partial-identification framework, by contrast, the dual objective may be flat over a non-singleton set of maximizers under the null, so the empirical optimizers need not settle on a common limit.

To guarantee correct size without specifying the joint distribution of $(T_1(\theta),T_2(\theta))$, the critical value for $T(\theta)$ must be calibrated against the least-favorable joint distribution consistent with the known standard normal marginals. This least-favorable bound reduces to the Bonferroni correction used in the definition of $p(\theta)$.
\end{remark}

\begin{remark}
If, for computational reasons, one prefers to work with the alternative one-sided dual representation $\mathscr{D}^{\dagger}(F, G;\theta)$ and its sieve approximation $\mathscr{D}_K^{\dagger}(F, G;\theta)$, as discussed in Remark~\ref{rmk:computation_dual}, the cross-fitted inference procedure developed above continues to apply with only straightforward modifications. Specifically, one replaces the optimization problem in \eqref{eq:empirical_argmax} with
\begin{multline*}
	(\hat\lambda_{1,n},\hat\beta_{1,n})
	\in
	\argmax_{\lambda\in\mathbb B_{1,n},\ \beta\in\mathbb C_n}
	\Bigg\{
	\frac{1}{|\mathcal I^d_1|}\sum_{i\in\mathcal I^d_1}
	\beta^{\top}\varphi(W_i,\hat Z_i,S_i)
	\\
	+
	\frac{1}{|\mathcal I^v_1|}\sum_{j\in\mathcal I^v_1}
	\inf_{(w,\hat z,s)\in\mathcal X_d}
	\Big[
	\lambda^{\top}\tilde q(w,Z_j,\hat z,s,\hat Z'_j,S'_j;\theta)
	-
	\beta^{\top}\varphi(w,\hat z,s)
	\Big]
	\Bigg\},
\end{multline*}
and replaces the definitions of $\hat u_i$ and $\hat v_j$ with
\begin{align*}
	\hat u_i
	&\coloneqq
	\hat\beta_{1,n}^{\top}\varphi(W_i,\hat Z_i,S_i),
	\qquad i\in\mathcal I^d_2,
	\\
	\hat v_j
	&\coloneqq
	\inf_{(w,\hat z,s)\in\mathcal X_d}
	\Big[
	\hat\lambda_{1,n}^{\top}\tilde q(w,Z_j,\hat z,s,\hat Z'_j,S'_j;\theta)
	-
	\hat\beta_{1,n}^{\top}\varphi(w,\hat z,s)
	\Big],
	\qquad j\in\mathcal I^v_2.
\end{align*}
Apart from these substitutions, the remainder of the inference procedure is unchanged.
\end{remark}

To establish asymptotic size control formally, we impose the following primitive regularity conditions.

\begin{assumption}[Regularity for asymptotic inference]\label{assu:inference_regularity}
Assume the following conditions hold:
\begin{enumerate}
	\item\label{enu:fold_size}
	As $n\to\infty$, both $n_d\to\infty$ and $n_v\to\infty$ (and hence $\underline n\to\infty$). Moreover, the fold sizes are asymptotically proportional to their respective sample sizes. That is, there exist constants $\kappa_d,\kappa_v>0$ such that, for each $m\in\{1,2\}$,
	\[
	\liminf_{n\to\infty}\frac{|\mathcal I^d_m|}{n_d}\ge \kappa_d
	\qquad\text{and}\qquad
	\liminf_{n\to\infty}\frac{|\mathcal I^v_m|}{n_v}\ge \kappa_v.
	\]

	\item\label{enu:compact_support}
	The support $\mathcal X_d$ of $(W,\hat Z,S)$ under $F$ is compact. The support $\mathcal X_v$ of $(Z,\hat Z,S)$ under $G$ is also compact.

	\item\label{enu:pd_cov}
	The population covariance matrix $\Omega$ of $q(W,\hat{Z};\theta)$ and the population covariance matrix $\Lambda$ of $(\hat{Z}^\top, S^\top)^\top$ are positive definite.

\item\label{enu:sieve}
Each sieve basis function is uniformly bounded on $\mathcal X_v$. That is, there exists a constant $C<\infty$ such that, for every basis function $\varphi_k(\cdot)$,
\[
\sup_{(z,\hat z',s')\in\mathcal X_v}|\varphi_k(z,\hat z',s')|\le C.
\]
In addition, the sieve dimension $K_n$ and the coefficient bound $c_n$ satisfy
\begin{equation}\label{eq:sieve_complexity}
\lim_{\underline n\to\infty}\frac{c_n\sqrt{K_n}}{\sqrt{\underline n}}=0.
\end{equation}
\end{enumerate}
\end{assumption}

Assumption~\ref{assu:inference_regularity} consists of mild regularity conditions. Assumption~\ref{assu:inference_regularity}\ref{enu:fold_size} requires only that both the downstream and validation sample sizes increase to infinity and that no fold becomes asymptotically negligible. In particular, it places no restriction on the relative magnitudes of $n_d$ and $n_v$. This flexibility is important in applications where the validation sample may be either much smaller or much larger than the downstream sample.
Assumptions~\ref{assu:inference_regularity}\ref{enu:compact_support}--\ref{enu:pd_cov} are standard regularity conditions. In Assumption~\ref{assu:inference_regularity}\ref{enu:sieve}, the uniform bound on the sieve basis functions is largely a normalization: by rescaling the basis, one can always impose a common sup-norm bound, such as $C=1$. Equation \eqref{eq:sieve_complexity} is a rate condition, which restricts the growth of the effective complexity of the sieve space. Intuitively, because the sieve coefficients are constrained by $\|\beta\|\le c_n$, the relevant envelope of the sieve approximation on $\mathcal X_v$ is of order $c_n\sqrt{K_n}$, and \eqref{eq:sieve_complexity} requires this envelope to grow slowly relative to sampling variation.

\begin{theorem}[Asymptotic Size Control]\label{thm:asymptotic_size}
	Suppose Assumptions \ref{assu:consistent_validation}, \ref{assu:regularity_minimax}, and \ref{assu:inference_regularity} hold. For any $\theta \in \Theta_I(F, G)$, the cross-fitted inference procedure asymptotically controls the probability of a false rejection at the nominal significance level $\alpha \in (0, 1)$:
	\begin{equation}\label{eq:size_control}
		\limsup_{\underline{n} \to \infty} \, \prob\big(p(\theta) < \alpha\big) \le \alpha.
	\end{equation}
\end{theorem}

Theorem~\ref{thm:asymptotic_size} establishes pointwise asymptotic size control for the proposed procedure. This result can be strengthened to uniform size control over suitably restricted classes of pairs $(F,G)$, for example, classes with uniformly bounded supports and covariance matrices whose eigenvalues are uniformly bounded away from zero. 

Because the critical value is calibrated against the least favorable joint distribution of $(T_1(\theta),T_2(\theta))$ consistent with their marginal standard normal limits, the resulting procedure may be conservative, so the inequality in \eqref{eq:size_control} can be strict. This conservativeness is the price of maintaining computational tractability under weak regularity conditions. Finally, both the proposed method and Theorem~\ref{thm:asymptotic_size} concern inference on the full parameter vector. A conservative procedure for subvector inference can be obtained by projecting the full-vector confidence region onto the subvector of interest. A sharper treatment of subvector inference is left for future work.

\section{Monte Carlo Simulations}\label{sec:simulation}

In this section, we conduct Monte Carlo simulations to assess the finite-sample performance of the proposed partial-identification and inference procedures. These simulations are designed to illustrate several key features of the methodology. First, we examine the size control of our procedure and compare it with that of a naive plug-in approach that ignores measurement error in the ML-generated proxy. Second, we illustrate the power of the proposed cross-fitted test and how it varies with sample size. Third, we show how incorporating a stratifying variable $S$ can tighten the identified bounds on the structural parameters. Fourth, we compare the results obtained using discrete and continuous proxies $\hat{Z}$. We also examine the role played by the number of sieve basis functions. This final exercise further illustrates how our framework accommodates settings in which the latent target $Z$ and its proxy $\hat{Z}$ take values in different spaces. The next subsection describes the data-generating process used throughout these simulations.

\subsection{Data-Generating Process}

We consider a design motivated by Example~\ref{ex:remote}, in which the downstream researcher is interested in the regression model
\begin{equation*}
	Y = Z\theta_1 + C\theta_2 + \varepsilon,
\end{equation*}
where $Z$ is a binary regressor that is unobserved in the downstream sample. Instead, the researcher observes a machine-learned proxy $\hat Z = g(X)$ constructed from a high-dimensional predictor vector $X$. The data-generating process (DGP) is designed so that $(Z,C)$ are exogenous in the structural equation, whereas $X$ contains both exogenous and endogenous components. Consequently, although the true regressor $Z$ is exogenous, the proxy $\hat Z$ may be endogenous. In the downstream sample, the researcher observes $(Y,C,X)$. For simplicity, we take $C$ to be scalar and set the true parameter vector to $\theta_0 = (1,1)^\top$.

Specifically, let $(V,U,\xi,\nu)$ be mutually independent random variables, each distributed as $N(0,1)$. We then generate the remaining variables as follows:
\begin{itemize}
	\item The regression error is
	\begin{equation*}
		\varepsilon = \xi + U.
	\end{equation*}

	\item The observed regressor $C$ is
	\begin{equation*}
		C = V + \nu.
	\end{equation*}

	\item The high-dimensional predictor vector is partitioned as $X = (X^{ex}, X^{en})$, where the exogenous block $X^{ex}$ is generated according to
	\begin{equation*}
		X^{ex}_j = V + \zeta^{ex}_j,
		\qquad
		j = 1,\dots,\dim(X^{ex}),
	\end{equation*}
	and the endogenous block $X^{en}$ is generated according to
	\begin{equation*}
		X^{en}_j = V + U + \zeta^{en}_j,
		\qquad
		j = 1,\dots,\dim(X^{en}).
	\end{equation*}
	Here, $\zeta^{ex}_j$ and $\zeta^{en}_j$ are i.i.d.\ $N(0,1)$ random variables, independent of $(V,U,\xi,\nu)$. We set
	\begin{equation*}
		\dim(X^{ex}) = 10
		\qquad\text{and}\qquad
		\dim(X^{en}) = 490.
	\end{equation*}
	In implementing the estimation and inference procedures, we do not reveal which components of $X$ belong to the exogenous block and which belong to the endogenous block.

	\item The latent binary regressor $Z$ is generated by
	\begin{equation}\label{eq:simu_Z_model}
		Z = \indicator\left\{ \kappa_0 + Q + \sigma_\eta(X^{ex})\eta > 0 \right\},
		\qquad
		Q = \frac{1}{\sqrt{\dim(X^{ex})}} \sum_{j=1}^{\dim(X^{ex})} X^{ex}_j.
	\end{equation}
	Here, $\eta$ is independent of all previously defined variables and follows a Type~I extreme-value distribution. The intercept $\kappa_0$ is calibrated so that the marginal probability $\prob(Z=1)$ equals $0.5$ in every simulation design. The scale function $\sigma_\eta(X^{ex})$ governs the difficulty of predicting $Z$ from $X$. In the baseline specification, we consider three homoskedastic designs,
	\begin{equation*}
		\sigma_\eta(X^{ex}) = s,
		\qquad
		s \in \{s_L, s_M, s_H\},
	\end{equation*}
	corresponding to low, medium, and high levels of prediction noise. Specifically, we set
	\begin{equation*}
		(s_L,s_M,s_H) = (0.1,1,3).
	\end{equation*}
	In Subsection~\ref{sec:simu_stratification}, we also consider heteroskedastic specifications for $\sigma_\eta(X^{ex})$ to illustrate the role of stratification through $S$. In the baseline design, however, we set $S$ to be constant and therefore omit it from the analysis.
\end{itemize}

By construction, $(Z,C)$ are exogenous in the downstream regression. At the same time, because $U$ enters both $\varepsilon$ and the endogenous block $X^{en}$, a flexible predictor $\hat Z = g(X)$ may inherit endogeneity through its dependence on $X^{en}$. 

\subsection{Construction of the Proxy $\hat{Z}$}

The prediction rule $g(\cdot)$ is estimated from an i.i.d.\ upstream training sample of $(Z,X)$ of size $n_{\text{train}}=5000$. This sample is generated independently of both the downstream sample and the auxiliary validation sample used to assess the proxy. We simulate the upstream training sample only once and keep both the realized sample and the resulting fitted rule $g(\cdot)$ fixed across all Monte Carlo replications. This design mirrors the theoretical framework in Sections~\ref{sec:validation_identification} and \ref{sec:validation_inference}, in which the downstream econometric analysis is conducted conditional on a pre-trained prediction rule.

In the simulations, we estimate $g(\cdot)$ using an $\ell_1$-penalized logistic regression (logistic LASSO). For a given penalty parameter $\gamma>0$, the estimator solves
\begin{multline*}
	(\hat{\tau}_0(\gamma),\hat{\tau}_1(\gamma))
	=
	\argmin_{\tau_0,\tau_1}
	\Bigg\{
	-\frac{1}{n_{\text{train}}}\sum_{i=1}^{n_{\text{train}}}
	\Big[
	Z_i\log \Lambda(\tau_0+X_i^\top\tau_1)
	\\
	+
	(1-Z_i)\log\big(1-\Lambda(\tau_0+X_i^\top\tau_1)\big)
	\Big]
	+
	\gamma\|\tau_1\|_1
	\Bigg\},
\end{multline*}
where $\|\tau_1\|_1=\sum_j |\tau_{1,j}|$ is the $\ell_1$ norm of the slope vector, and $\Lambda(\cdot)$ denotes the standard logistic cumulative distribution function,
\begin{equation*}
	\Lambda(v)=\frac{\exp(v)}{1+\exp(v)}.
\end{equation*}

To choose $\gamma$, we perform five-fold cross-validation over a dense grid of candidate values. Specifically, we randomly partition the training indices into five equal-sized folds, denoted by $\mathcal I_1,\dots,\mathcal I_5$. For each candidate $\gamma$ and each fold $k\in\{1,\dots,5\}$, we fit the model on the remaining four folds to obtain fold-specific estimates $\big(\hat{\tau}_0^{(-k)}(\gamma),\hat{\tau}_1^{(-k)}(\gamma)\big)$ and evaluate predictive performance on the held-out fold $k$. We then select the tuning parameter
\begin{equation*}
	\gamma^*
	=
	\argmin_{\gamma}
	\left\{
	-\frac{1}{5}\sum_{k=1}^5 \frac{1}{|\mathcal I_k|}
	\sum_{i\in\mathcal I_k}
	\Big[
	Z_i \log \hat{P}_i^{(-k)}(\gamma)
	+
	(1-Z_i)\log\big(1-\hat{P}_i^{(-k)}(\gamma)\big)
	\Big]
	\right\},
\end{equation*}
where
\[
\hat{P}_i^{(-k)}(\gamma)
=
\Lambda\big(\hat{\tau}_0^{(-k)}(\gamma)+X_i^\top \hat{\tau}_1^{(-k)}(\gamma)\big)
\]
is the predicted probability for observation $i$ when fold $k$ is held out.

After selecting $\gamma^*$, we refit the logistic LASSO on the full upstream training sample to obtain the final estimates $\hat{\tau}_0=\hat{\tau}_0(\gamma^*)$ and $\hat{\tau}_1=\hat{\tau}_1(\gamma^*)$. The resulting predicted probability for an observation with covariates $X$ is
\begin{equation*}
	\hat{P}(Z=1\mid X)=\Lambda(\hat{\tau}_0+X^\top \hat{\tau}_1).
\end{equation*}
In the baseline designs, we convert this score into a binary proxy via
\begin{equation}\label{eq:simu_binary_proxy}
	\hat{Z}=g(X)\coloneqq \indicator\big\{\hat{P}(Z=1\mid X)>0.5\big\}.
\end{equation}
The resulting classification accuracy, measured by $\prob(Z=\hat{Z})$, is approximately $98.42\%$ when $\sigma_\eta(X^{ex})=s_L$, $88.57\%$ when $\sigma_\eta(X^{ex})=s_M$, and $73.12\%$ when $\sigma_\eta(X^{ex})=s_H$.
\vspace{-0.5em}

\subsection{Size Control} We conduct the Monte Carlo study using $10{,}000$ replications. In each replication, the downstream sample consists of $n_d$ observations on $(Y,C,\hat Z)$, while the auxiliary validation sample consists of $n_v$ observations on $(Z,\hat Z)$. The baseline simulation design does not include a stratifying variable, so $S$ is omitted throughout this subsection. The role of stratification is examined separately in Subsection~\ref{sec:simu_stratification}.

For each replication, we implement the inference procedure developed in Section~\ref{sec:validation_inference}. Throughout the simulations, we set $c_n=\underline{n}^{2/5}$. We also experimented with the alternative choice $c_n=\underline{n}^{1/3}$ and obtained very similar results.  Because $(Z,\hat Z)$ is discrete and $S$ is absent, the inner $\inf$ problem in \eqref{eq:sieve_dual_value} is easy to solve numerically. Moreover, in this setting the sieve approximation is exact. We therefore take the basis vector $\varphi(Z,\hat Z)$ to be the four cell indicators corresponding to the support of $(Z,\hat Z)$:
\begin{equation}\label{eq:binary_proxy_construction}
\varphi(Z,\hat Z)
=
\begin{pmatrix}
\indicator(Z=1,\hat Z=1)\\
\indicator(Z=1,\hat Z=0)\\
\indicator(Z=0,\hat Z=1)\\
\indicator(Z=0,\hat Z=0)
\end{pmatrix}.
\end{equation}

An attractive feature of our framework is its robustness to asymmetric sample sizes in the downstream and validation datasets. To illustrate this property, we examine the rejection rate of the test under four sample-size configurations:
\[
(n_d,n_v)\in\{(500,500),\ (10000,10000),\ (10000,500),\ (500,10000)\}.
\]
The last two designs explicitly examine the performance of the procedure under substantial sample-size asymmetry. For each configuration, we report rejection frequencies at nominal significance levels of $10\%$, $5\%$, and $1\%$ when testing the true parameter value $\theta_0=(1,1)^\top$. For reference, on a Mac with an M4 Max CPU (16 CPU cores), the case $(n_d,n_v)=(10000,10000)$ takes about 90 seconds for all $10{,}000$ replications combined, or 9 milliseconds per replication.

A second implication of the theory is that the asymptotic validity of the procedure does not depend on the predictive accuracy of the upstream ML method. To illustrate this feature, we report size results for the three levels of prediction noise, $(s_L,s_M,s_H)$, introduced in the data-generating process. Although higher prediction noise naturally widens the identified set, it should not compromise the size control of the test.

As a benchmark, we also report the simulated rejection frequencies of a conventional plug-in $F$-test based on the ordinary least squares regression of $Y$ on $(\hat Z,C)$. This naive procedure treats $\hat Z$ as if it were observed without error and therefore ignores both the prediction error in $\hat Z$ and the endogeneity that $\hat Z$ may inherit from the endogenous components of $X$.

Table~\ref{tab:size_control} reports the results. Two conclusions emerge. First, our procedure controls size well across all sample-size configurations and all levels of prediction noise. Although conservative by construction, the conservativeness is mild in our simulation. For example, when $(n_d,n_v)=(10000,10000)$, the empirical rejection frequency at the $5\%$ level ranges from $4.2\%$ to $4.5\%$, depending on the prediction-noise design, while at the $1\%$ level it ranges from $0.9\%$ to $1.0\%$. Second, the naive plug-in $F$-test fails to control size. It exhibits substantial over-rejection whenever prediction noise is moderate or high, for all sample-size configurations. Even in the low-noise design, where the prediction accuracy is approximately $98.42\%$, over-rejection is still noticeable when $n_d=10{,}000$.

\begin{table}[htbp]
	\centering
	\caption{Empirical Rejection Probabilities Under the True Null Hypothesis}
	\label{tab:size_control}
	\begin{tabular}{ll *{3}{C{1.2cm}} @{\hspace{1em}} *{3}{C{1.2cm}}}
	  \toprule
			& & \multicolumn{3}{c}{\makebox[0pt][c]{Proposed Cross-Fitted Test}} & \multicolumn{3}{c}{\makebox[0pt][c]{Naive Plug-In $F$-Test}} \\
\cmidrule(lr){3-5} \cmidrule(lr){6-8}
		Prediction Noise & $(n_d, n_v)$ & $10\%$ & $5\%$ & $1\%$ &  $10\%$ & $5\%$ & $1\%$ \\
		\midrule
		\multirow{4}{*}{Low ($s_L$)} 
		 & (500,  500) & 7.3\% & 3.8\% & 0.7\% & 10.2\% & 5.0\% & 1.1\%  \\
 & (10000,  10000) & 8.5\% & 4.2\% & 0.9\% & 13.0\% & 7.1\% & 1.7\%  \\
 & (10000,  500) & 9.5\% & 4.7\% & 0.9\% & 13.0\% & 7.1\% & 1.7\%  \\
 & (500,  10000) & 6.9\% & 3.4\% & 0.7\% & 10.2\% & 5.0\% & 1.1\%  \\

 		&  & & &  &  & &   \\
		\multirow{4}{*}{Medium ($s_M$)} 
		 & (500,  500) & 8.6\% & 4.7\% & 0.8\% & 24.9\% & 15.8\% & 5.1\%  \\
 & (10000,  10000) & 9.0\% & 4.3\% & 1.0\% & 99.9\% & 99.5\% & 97.9\%  \\
 & (10000,  500) & 9.9\% & 5.1\% & 1.1\% & 99.9\% & 99.5\% & 97.9\%  \\
 & (500,  10000) & 9.2\% & 4.7\% & 1.0\% & 24.9\% & 15.8\% & 5.1\%  \\

 		&  & & &  &  & &   \\
		\multirow{4}{*}{High ($s_H$)} 
		 & (500,  500) & 9.1\% & 5.2\% & 1.0\% & 68.8\% & 56.9\% & 33.0\%  \\
 & (10000,  10000) & 8.9\% & 4.5\% & 1.0\% & 100\% & 100\% & 100\%  \\
 & (10000,  500) & 10.0\% & 5.1\% & 1.2\% & 100\% & 100\% & 100\%  \\
 & (500,  10000) & 10.0\% & 5.0\% & 1.0\% & 68.8\% & 56.9\% & 33.0\%  \\
\bottomrule
	\end{tabular}
    \begin{minipage}{15.5cm}
			{\footnotesize \textit{Notes:} Entries report empirical rejection frequencies across $N=10000$ Monte Carlo replications when the candidate parameter is set equal to the true value $\theta_0=(1,1)^\top$. The proposed procedure is the cross-fitted test developed in Section~\ref{sec:validation_inference} with sieve defined in \eqref{eq:binary_proxy_construction} and $c_n = \underline{n}^{2/5}$. The naive procedure is the conventional plug-in $F$-test from the regression of $Y$ on $(\hat Z,C)$, which treats $\hat Z$ as if it were observed without error. The proxy $\hat Z$ is constructed using the logistic LASSO prediction rule in \eqref{eq:simu_binary_proxy}. For reference, on a Mac with an M4 Max CPU (16 CPU cores), it takes about 90 seconds to finish all 10,000 replications of the tests when $(n_d,n_v)=(10000,10000)$ regardless of the value of predictio noise.}
	\end{minipage}
\end{table}

\subsection{Power and Informativeness}
Figure~\ref{fig:power} visualizes the finite-sample power of the proposed procedure through heatmaps over the two-dimensional parameter space $(\theta_1,\theta_2)$. For each design, we evaluate the test on a grid of $10{,}000$ parameter values and repeat the experiment over $500$ Monte Carlo replications. At each grid point, we record the fraction of replications in which the null hypothesis is rejected at the $5\%$ significance level. The resulting heatmap therefore summarizes the empirical rejection probability across the parameter space. The purple end of the color scale corresponds to a rejection frequency of $0\%$, meaning that the point lies in the $95\%$ confidence set in every replication, whereas the yellow end corresponds to a rejection frequency of $100\%$, meaning that the point is excluded from the $95\%$ confidence set in every replication.

Figure~\ref{fig:power} reports confidence sets for nine different sample-size combinations, arranged in a $3\times 3$ panel. Across rows, the downstream sample size $n_d$ increases from $500$ to $2000$, while across columns the validation sample size $n_v$ increases from $500$ to $2000$. As the figure shows, increasing either $n_d$ or $n_v$ sharpens the rejection surface and yields a more informative confidence set. The gains in power are most pronounced, however, when both sample sizes increase simultaneously. This pattern is consistent with our asymptotic theory, under which the precision of the procedure is governed by $\underline{n}=\min(n_d,n_v)$.

To illustrate how predictive performance affects the informativeness of our procedure, Figure~\ref{fig:power_noise} plots the empirical $95\%$ confidence sets under the three prediction-noise designs $s\in\{s_L,s_M,s_H\}$. The heatmaps are constructed in the same manner as in Figure~\ref{fig:power} with sample sizes fixed at $(n_d, n_v) = (1000, 1000)$. As expected, lower prediction noise and hence better predictive performance yield more informative inference, producing smaller confidence sets. 

\clearpage
\thispagestyle{plain}

\begin{center}
\includegraphics[width=\textwidth,height=.72\textheight,keepaspectratio]{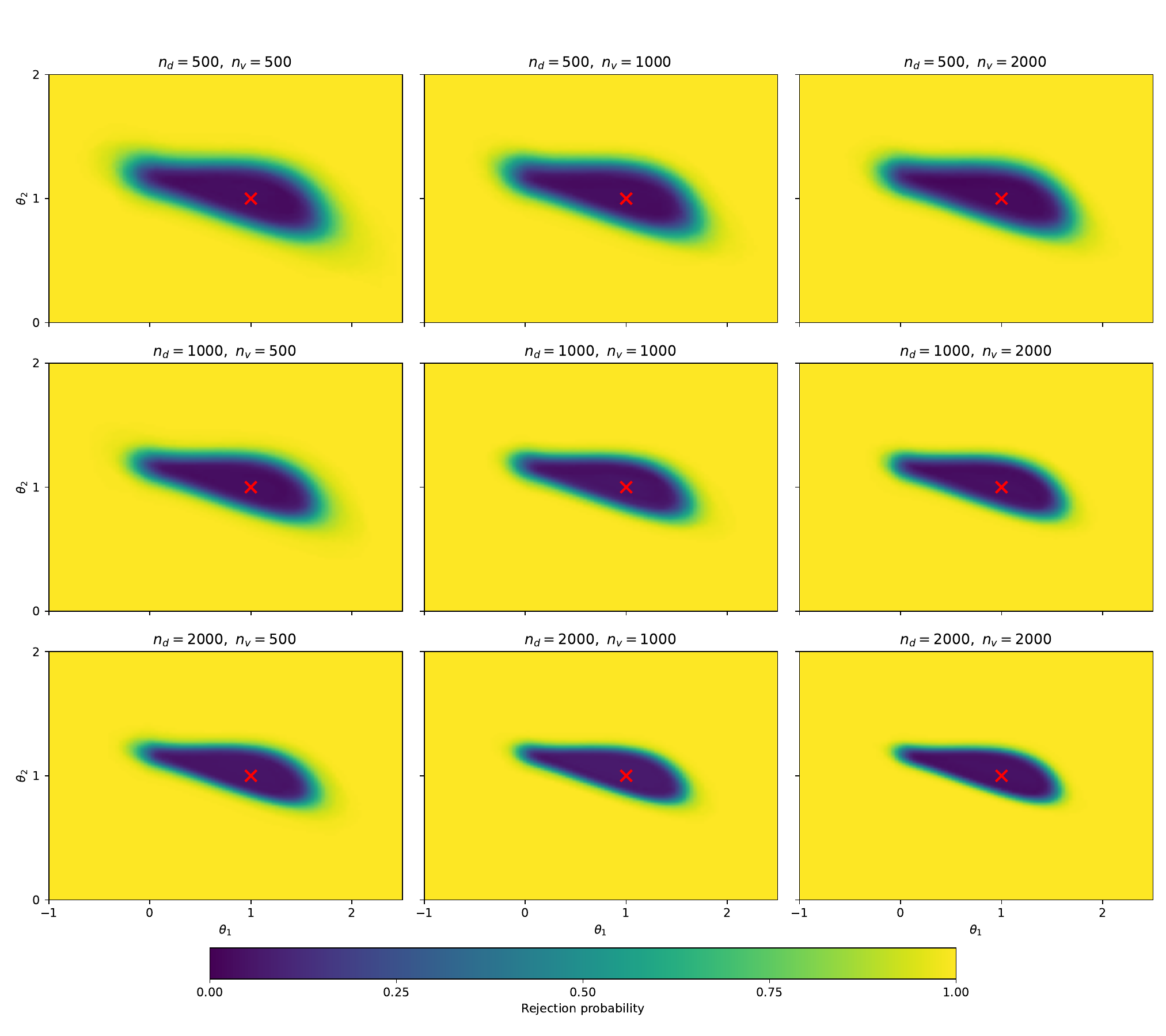}

\captionof{figure}{95\% Confidence Set across Different Sample Sizes}
\label{fig:power}

\begin{minipage}{0.96\linewidth}
\footnotesize
	\textit{Notes:} Each panel reports the empirical rejection frequency of the proposed test at the $5\%$ level over a grid of $10000$ candidate values of $(\theta_1,\theta_2)$, based on $500$ Monte Carlo replications. 
The sieve function is defined in \eqref{eq:binary_proxy_construction} and $c_n = \underline{n}^{2/5}$.
	The implied $95\%$ confidence set in each replication is the set of grid points not rejected by the test. Dark purple corresponds to a rejection frequency of $0\%$, meaning that the point is included in the confidence set in every replication, whereas bright yellow corresponds to a rejection frequency of $100\%$, meaning that the point is excluded in every replication. The panels are arranged in a $3\times 3$ layout: moving down the rows increases the downstream sample size $n_d$ from $500$ to $2000$, and moving across the columns increases the validation sample size $n_v$ from $500$ to $2000$.
\end{minipage}
\end{center}

\clearpage
\thispagestyle{plain}

\begin{center}
\includegraphics[width=.5\textwidth,height=.72\textheight,keepaspectratio]{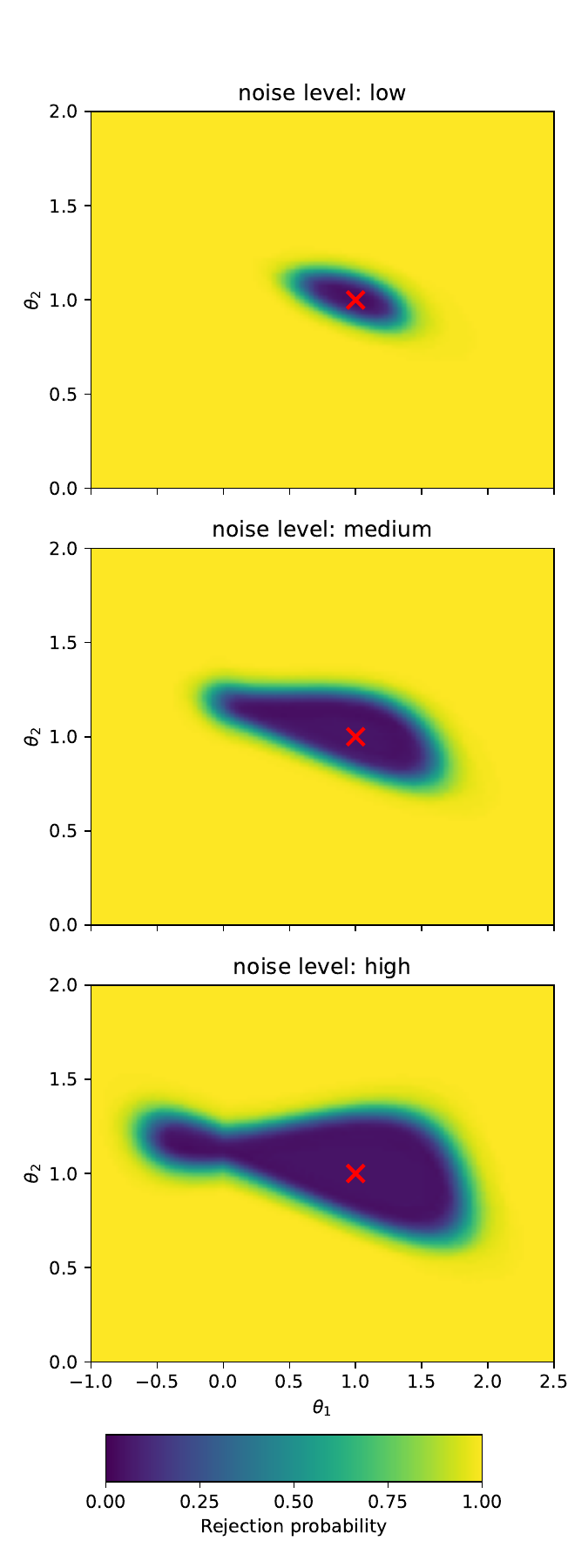}

\captionof{figure}{95\% Confidence Set across Prediction Noises}
\label{fig:power_noise}

\begin{minipage}{0.96\linewidth}
\footnotesize
\textit{Notes:} Each panel reports the empirical rejection frequency of the proposed test at the $5\%$ level under a different prediction-noise design. The heatmaps are constructed in the same manner as in Figure~\ref{fig:power}. Only the level of prediction noise varies across panels.
\end{minipage}
\end{center}
\clearpage

\subsection{Effect of Stratification}\label{sec:simu_stratification}

To illustrate the role of stratification, we modify the baseline design by allowing the scale of the latent shock in \eqref{eq:simu_Z_model} to vary with an observable component of $X$. Specifically, we let
\begin{equation*}
\sigma_\eta(X^{ex})
=
s_1 \indicator\{ |X_1| > 1 \}
+
s_2 \indicator\{ |X_1| \le 1 \},
\end{equation*}
where $X_1$ denotes the first component of $X$.\footnote{To preserve the exogeneity of $Z$ in the structural equation, we take the first component of $X$ to belong to the exogenous block in the data-generating process. This information is used only to define the stratification variable and is not otherwise exploited in the inference procedure.}
We then define the stratification variable by
\begin{equation*}
S = \indicator\{ |X_1| > 1 \}.
\end{equation*}
Under this design, the variance of the prediction noise varies across the two strata indexed by $S$. The setup is meant to mimic empirical environments in which the latent target variable $Z$ is easier to predict for some subpopulations than for others.

\begin{figure}[htbp!]
	\centering
	\includegraphics[width=\linewidth]{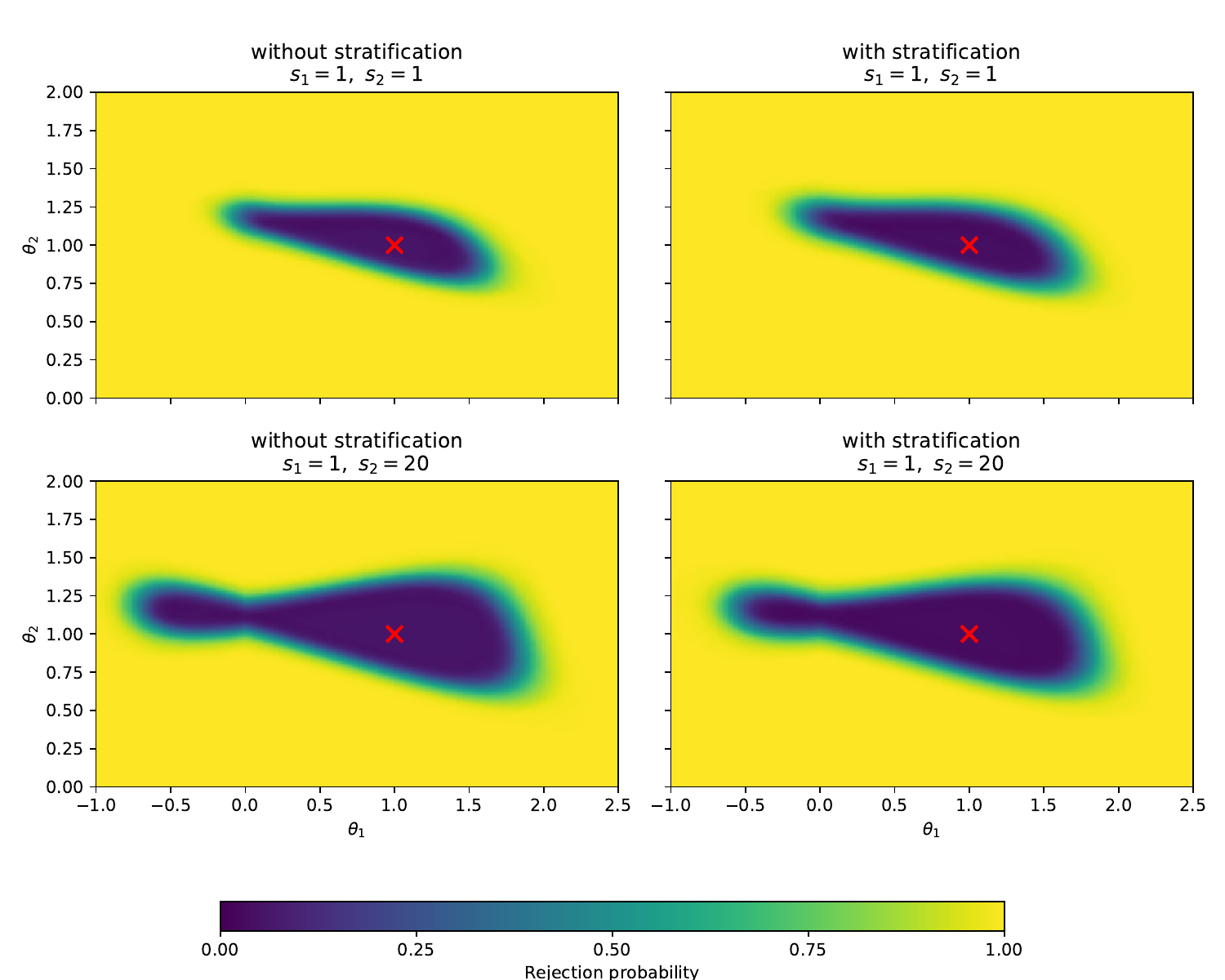}
	\caption{Effect of Stratification on Confidence Sets}
	\label{fig:stratification}
	\vspace{0.5em}
	\begin{minipage}{0.96\linewidth}
	\footnotesize
\textit{Notes:} Each panel reports the empirical rejection frequency of the proposed test at the $5\%$ level across homoskedastic and heteroskedastic prediction-noise designs, with and without stratification. The heatmaps are constructed in the same manner as in Figure~\ref{fig:power}.
	\end{minipage}
\end{figure}

We compare inference with and without incorporating $S$ under two configurations of $(s_1,s_2)$. We begin with the benchmark case $(s_1,s_2)=(1,1)$. In this design, the prediction noise is homoskedastic, so stratification does not provide any additional information for identification. Accordingly, incorporating $S$ should not tighten the identified set for $\theta$. This is what we see in the first row of Figure~\ref{fig:stratification}: adding $S$ delivers essentially no gain. If anything, because conditioning on $S$ increases the dimensionality of the problem, the finite-sample performance is slightly worse when stratification is included.

We next consider a highly heterogeneous design with $(s_1,s_2)=(1,20)$. In this case, conditional on $S=0$, the covariates $X$ are less informative about $Z$, so the classification accuracy of $\hat Z$ in that stratum is much lower. As shown in the lower panels of Figure~\ref{fig:stratification}, incorporating stratification yields noticeably tighter inference. Although $S$ does not improve the predictive performance of $\hat Z$ itself, it identifies the subpopulations in which the proxy is informative and those in which it is not.  More broadly, this exercise illustrates a central theme of our framework. The pair $(\hat Z,S)$ serves as a bridge that transfers information about $Z$ from the validation sample to the downstream moment conditions. As a result, even when $S$ does not directly enhance prediction, it can still improve inference by helping to characterize the conditional distribution of $Z$ given $(\hat Z,S)$.

\subsection{Continuous $\hat{Z}$ and Sieve Approximation}
To illustrate how the inference procedure applies to a continuous proxy and to assess the effect of sieve approximation, we modify the baseline design by replacing the binary proxy $\hat{Z}$ in \eqref{eq:simu_binary_proxy} with the continuous proxy
\begin{equation}\label{eq:simu_continuous_proxy}
    \hat{Z} = g(X) \coloneqq \hat{P}(Z = 1 \mid X).
\end{equation}
Because $Z$ remains binary whereas $\hat{Z}$ is continuous, this specification also shows that the procedure can accommodate settings in which $Z$ and $\hat{Z}$ take values in different spaces.

Since $\hat{Z}$ is continuous, the basis functions in \eqref{eq:binary_proxy_construction} are no longer appropriate. We therefore use the following order-$d$ sieve basis, which has dimension $K = 2(d + 1)$:
\begin{equation}\label{eq:continuous_proxy_sieve_example}
    \varphi_K(Z, \hat{Z}) =
    \begin{pmatrix}
    \indicator(Z = 1)\\
    \indicator(Z = 1)\hat{Z}\\
    \indicator(Z = 1)\hat{Z}^2\\
    \vdots\\
    \indicator(Z = 1)\hat{Z}^{d}\\
    \indicator(Z = 0)\\
    \indicator(Z = 0)\hat{Z}\\
    \indicator(Z = 0)\hat{Z}^2\\
    \vdots\\
    \indicator(Z = 0)\hat{Z}^{d}
    \end{pmatrix}.
\end{equation}
This basis can be viewed as the tensor product of the indicator basis $(\indicator(Z = 1), \indicator(Z = 0))$ and the polynomial basis $(1, \hat{Z}, \hat{Z}^2, \ldots, \hat{Z}^{d})$. We set $c_n = (\min\{n_d, n_v\})^{2/5}$.

\begin{table}[htbp]
    \centering
    \caption{Empirical Null Rejection Probabilities with Continuous Proxy}
    \label{tab:sieve_size}
    \begin{tabular}{cl *{3}{C{2cm}} }
        \toprule
        \multirow{2}{*}{Polynomial Order} & \multirow{2}{*}{$(n_d,n_v)$} & \multicolumn{3}{c}{Empirical Rejection Probability} \\
        \cmidrule(lr){3-5}
        & & $10\%$ & $5\%$ & $1\%$ \\
        \midrule

        \multirow{4}{*}{1}
			    & (500,  500) & 6.3\% & 3.3\% & 0.6\%  \\
 & (10000,  10000) & 4.2\% & 2.0\% & 0.4\%  \\
 & (10000,  500) & 7.8\% & 3.9\% & 0.8\%  \\
 & (500,  10000) & 6.3\% & 3.1\% & 0.5\%  \\

            &         &      &      &      \\

        \multirow{4}{*}{2}
			    & (500,  500) & 7.6\% & 3.9\% & 0.7\%  \\
 & (10000,  10000) & 8.8\% & 4.5\% & 0.9\%  \\
 & (10000,  500) & 9.2\% & 4.7\% & 1.1\%  \\
 & (500,  10000) & 8.1\% & 3.8\% & 0.6\%  \\

            &         &      &      &      \\
        \multirow{4}{*}{3}
			    & (500,  500) & 7.6\% & 3.9\% & 0.7\%  \\
 & (10000,  10000) & 8.9\% & 4.5\% & 0.9\%  \\
 & (10000,  500) & 9.5\% & 4.9\% & 1.1\%  \\
 & (500,  10000) & 8.1\% & 4.1\% & 0.7\%  \\

            &         &      &      &      \\

        \multirow{4}{*}{4}
			    & (500,  500) & 7.9\% & 4.0\% & 0.7\%  \\
 & (10000,  10000) & 9.3\% & 4.6\% & 0.9\%  \\
 & (10000,  500) & 9.6\% & 4.9\% & 1.1\%  \\
 & (500,  10000) & 8.3\% & 4.2\% & 0.8\%  \\

            &         &      &      &      \\

        \multirow{4}{*}{5}
			    & (500,  500) & 7.7\% & 4.0\% & 0.7\%  \\
 & (10000,  10000) & 9.1\% & 4.5\% & 0.9\%  \\
 & (10000,  500) & 9.4\% & 4.9\% & 1.1\%  \\
 & (500,  10000) & 8.2\% & 4.1\% & 0.7\%  \\
\bottomrule
    \end{tabular}

    \vspace{0.5em}
    \begin{minipage}{0.9\linewidth}
			\footnotesize \textit{Notes:} Entries report empirical null rejection frequencies across $N=10000$ Monte Carlo replications. The prediction noise is fixed at the homoskedastic level $s_M = 1.0$. The inference procedure is the cross-fitted test developed in Section~\ref{sec:validation_inference} with sieve defined in \eqref{eq:continuous_proxy_sieve_example} and $c_n = \underline{n}^{2/5}$. The proxy $\hat Z$ is constructed using the logistic LASSO prediction rule in \eqref{eq:simu_continuous_proxy}. For reference, on a Mac with an M4 Max CPU (16 CPU cores), it takes less than 260 seconds to finish all 10,000 replications of the tests when $(n_d,n_v)=(10000,10000)$, regardless of the polynomial order.
    \end{minipage}
\end{table}

Table~\ref{tab:sieve_size} reports empirical rejection probabilities for sieve orders $d = 1,2,\ldots,5$ across different sample-size configurations. In this exercise, the prediction noise is fixed at the homoskedastic medium level $s_M = 1.0$. As Lemma~\ref{lem:sieve_approx} suggests, replacing the infinite-dimensional dual space with a finite-dimensional sieve approximation preserves the validity of the underlying necessary condition. Accordingly, the proposed procedure exhibits good size control across all polynomial orders. At the same time, the polynomial order of the sieve affects the degree of conservativeness. When the sieve order is low, the sieve space may approximate the optimal dual function $\psi$ in \eqref{eq:convex_characterization} only coarsely, so the resulting approximation error can be non-negligible. As discussed in Section~\ref{sec:validation_inference}, this approximation error leads to a gap between the actual dual $\mathscr{D}(F, G;\theta)$ and its finite-dimensional approximation $\mathscr{D}_K(F, G;\theta)$, therefore making the test more conservative. This pattern is visible in Table~\ref{tab:sieve_size}, especially when the polynomial order is $1$.

As the polynomial order increases, the sieve approximation improves and the test becomes less conservative. This is borne out in the simulation results: moving from a first-order to a second-order sieve noticeably reduces conservativeness, and the third-order and the fourth-order sieve yield some further improvements. However, we observe diminishing gain as the order of polynomial increases. This suggests that, once the sieve is sufficiently flexible, the remaining conservativeness is driven primarily by the same force that underlies the results in Table~\ref{tab:size_control}. Specifically, our critical values are calibrated using the least favorable joint distribution consistent with the known standard normal marginals of the fold-specific statistics. This calibration yields a tractable pivotal critical value under relatively weak regularity conditions, but it also introduces an inherent degree of conservativeness.

Finally, we compare the informativeness of the continuous proxy in \eqref{eq:simu_continuous_proxy} with that of the binary proxy in \eqref{eq:simu_binary_proxy}. As discussed in Remark~\ref{rem:data_combination} and Section~\ref{sec:uncond_OT_characterization}, the proxy $\hat{Z}$ is used as a bridge linking the downstream and validation samples. The more information $\hat{Z}$ preserves about the relationship between the high-dimensional input $X$ and the latent target $Z$, the tighter the resulting identified set and confidence set should be. In the present design, the binary proxy is obtained from the continuous proxy by comparing it with $0.5$ and is therefore a coarsening of it. The continuous proxy thus contains more information than the binary proxy, suggesting that it should yield tighter confidence sets.

To examine this implication, we simulate confidence sets for both the binary proxy and the continuous proxy, the latter using polynomial sieve bases of orders $d=1,\ldots,5$. Figure~\ref{fig:sieve_power_1000} reports the resulting confidence sets for the case $n_d=n_v=1000$, while Figure~\ref{fig:sieve_power_5000} reports the corresponding results for the larger sample size $n_d=n_v=5000$. In both exercises, the prediction noise is fixed at the homoskedastic medium level $s_M=1.0$. As in Figure~\ref{fig:power}, we evaluate the test on a grid of $10{,}000$ candidate parameter values and repeat the exercise over $500$ Monte Carlo replications. The results confirm the conjecture. Even with a first-order polynomial sieve, the continuous proxy yields tighter confidence sets than the binary proxy. This gain becomes more pronounced at the larger sample size. Moreover, the gains from using higher-order polynomial sieves are also more evident when the sample size is larger, since the reduction in sieve approximation error then more clearly outweighs the cost of estimating a higher-dimensional set of sieve coefficients.

\clearpage
\thispagestyle{plain}

\begin{center}
\includegraphics[width=\textwidth,height=.8\textheight,keepaspectratio]{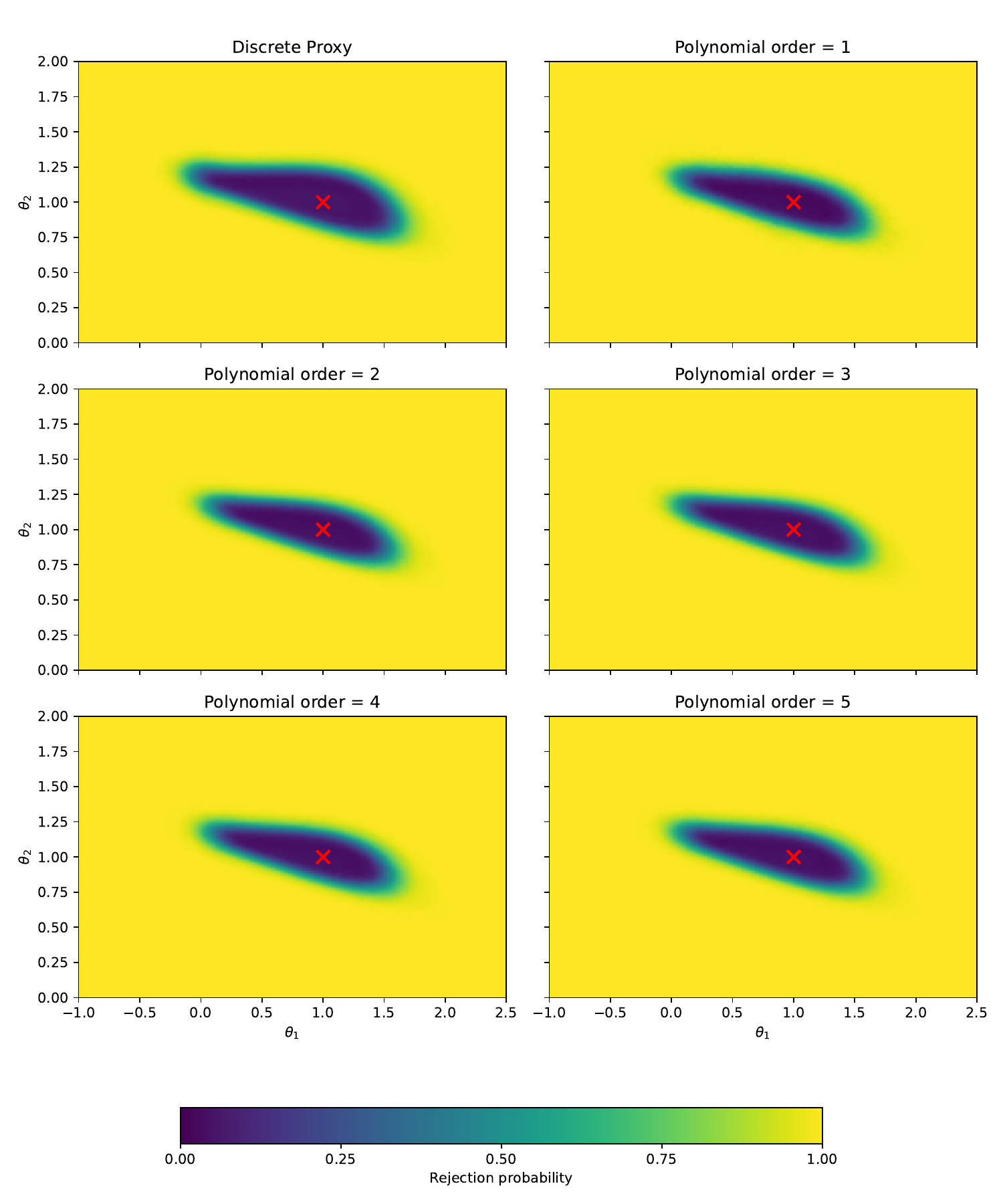}

\captionof{figure}{95\% Confidence Sets with Binary and Continuous Proxies: $(n_d,n_v)=(1000,1000)$}
\label{fig:sieve_power_1000}

\begin{minipage}{0.96\linewidth}
\footnotesize
\textit{Notes:} Each panel reports the empirical rejection frequency of the proposed test at the $5\%$ level over a grid of $10{,}000$ candidate values of $(\theta_1,\theta_2)$, based on $500$ Monte Carlo replications. The upper-left panel corresponds to the binary proxy defined in \eqref{eq:simu_binary_proxy}, using the sieve basis in \eqref{eq:binary_proxy_construction}. The remaining panels correspond to the continuous proxy defined in \eqref{eq:simu_continuous_proxy}, using the polynomial sieve basis in \eqref{eq:continuous_proxy_sieve_example} with orders $d=1,\ldots,5$. The sample sizes are fixed at $n_d=n_v=1000$.
\end{minipage}
\end{center}

\clearpage
\thispagestyle{plain}

\begin{center}
\includegraphics[width=\textwidth,height=.8\textheight,keepaspectratio]{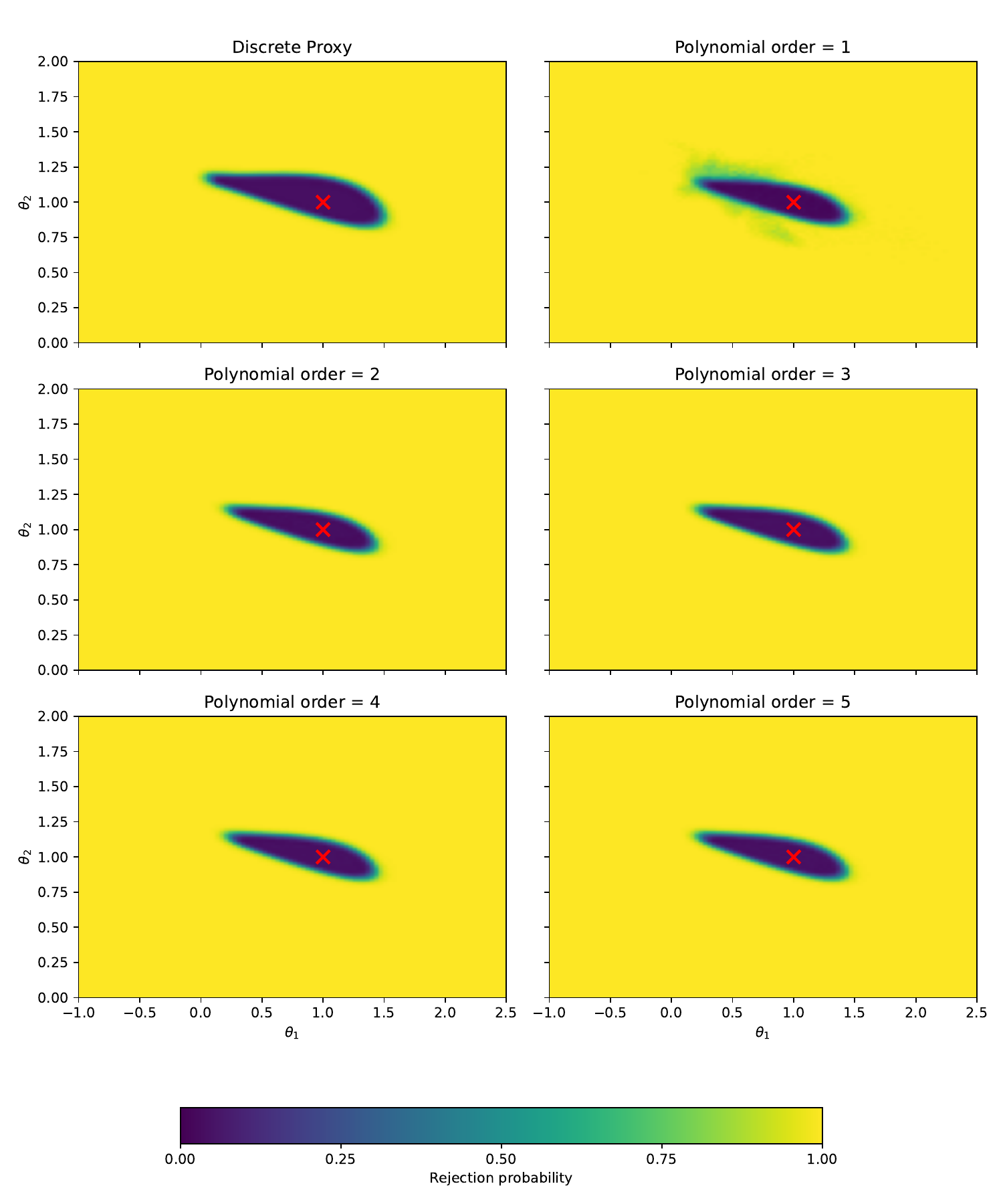}

\captionof{figure}{95\% Confidence Sets with Binary and Continuous Proxies: $(n_d,n_v)=(5000,5000)$}
\label{fig:sieve_power_5000}

\begin{minipage}{0.96\linewidth}
\footnotesize
\textit{Notes:} Each panel reports the empirical rejection frequency of the proposed test at the $5\%$ level over a grid of $10{,}000$ candidate values of $(\theta_1,\theta_2)$, based on $500$ Monte Carlo replications. The upper-left panel corresponds to the binary proxy defined in \eqref{eq:simu_binary_proxy}, using the sieve basis in \eqref{eq:binary_proxy_construction}. The remaining panels correspond to the continuous proxy defined in \eqref{eq:simu_continuous_proxy}, using the polynomial sieve basis in \eqref{eq:continuous_proxy_sieve_example} with orders $d=1,\ldots,5$. The sample sizes are fixed at $n_d=n_v=5000$.
\end{minipage}
\end{center}

\section{Extension}\label{sec:extensions}
We now consider an extension in which $\hat{Z}$ is vector-valued, but the available validation data do not jointly observe the full vector $(Z,\hat{Z},S)$. Instead, the researcher may have access to several validation samples, each of which contains information only on a subvector of $(Z,\hat{Z},S)$. In this case, the validation data identify a collection of lower-dimensional marginal distributions rather than the full joint distribution of $(Z,\hat{Z},S)$. As illustrated in Example~\ref{ex:media_slant}, this situation can arise when $\hat{Z}$ consists of several ML-generated proxies, each constructed separately rather than estimated jointly as a single vector-valued predictor. We show that this setting can be accommodated by a straightforward extension of the framework developed above.

Formally, partition $(Z,\hat{Z},S)$ into $L$ subvectors,
\[
(Z_1,\hat{Z}_1,S_1),\ldots,(Z_L,\hat{Z}_L,S_L).
\]
Let $G_l$ denote the population distribution of the validation observables $(Z_l,\hat{Z}_l,S_l)$ for $l=1,\ldots,L$. The compatibility condition in Assumption~\ref{assu:consistent_validation} is then replaced by the following.

\begin{assumption1prime}[Compatibility of Multiple Marginal Distributions]\label{assu:consistent_validation_multiple_marginal}
	Let $H_0$ denote the true joint distribution of $(W,Z,\hat Z,S)$. Assume that (\emph{i}) the marginal distribution of $(W,\hat Z,S)$ under $H_0$ is $F$, and (\emph{ii}) for each $l=1,\ldots,L$, the marginal distribution of $(Z_l,\hat Z_l,S_l)$ under $H_0$ is $G_l$.
\end{assumption1prime}

Define $\mathcal{H}_L\!\left(F,(G_l)_{l=1}^L\right)$ to be the set of all joint distributions $H$ over $(W,Z,\hat Z,S)$ whose marginal distribution on $(W,\hat Z,S)$ is $F$ and whose marginal distribution on $(Z_l,\hat Z_l,S_l)$ is $G_l$ for each $l=1,\ldots,L$. This leads to the following notion of identified set.

\begin{definition}[Identified Set with Multiple Validation Marginals]\label{def:id_set_validation_multiple_marginal}
	The identified set with multiple validation marginals $\Theta_I\!\left(F,(G_l)_{l=1}^L\right)$ is defined by
	\[
	\Theta_I\!\left(F,(G_l)_{l=1}^L\right)
	\coloneqq
	\left\{
	\theta\in\Theta :
	\exists\, H\in \mathcal{H}_L\!\left(F,(G_l)_{l=1}^L\right)
	\text{ such that }
	\expt_H\!\left[q(W,Z;\theta)\right]=0
	\right\}.
	\]
\end{definition}

As in the baseline case, we use the same decoupling device. Let $(\hat Z',S')$ denote auxiliary random vectors, and for each $l=1,\ldots,L$, let $(\hat Z_l',S_l')$ denote the subvector corresponding to $(\hat Z_l,S_l)$. Define $\mathcal{H}_L'\!\left(F,(G_l)_{l=1}^L\right)$ to be the set of all joint distributions over $(W,Z,\hat Z,S,\hat Z',S')$ whose marginal distribution on $(W,\hat Z,S)$ is $F$ and whose marginal distribution on $(Z_l,\hat Z_l',S_l')$ is $G_l$ for each $l=1,\ldots,L$. The regularity conditions in Assumption \ref{assu:regularity_minimax} then needs to be replaced by the following assumption.

\begin{assumption2prime}[Regularity conditions]\label{assu:regularity_minimax_multi_marginal}
	For any parameter $\theta\in \Theta$, assume that Assumptions \ref{assu:regularity_minimax}\ref{enu:polish}-\ref{enu:fintie_variance} hold, and that there exist non-negative continuous functions $b_1 \in \mathcal{L}^1(F)$ and $b_{2,l} \in \mathcal{L}^1(G_l)$ for each $l=1,...,L$ such that the structural moment function is bounded by an additively separable envelope:
		\[
			\|q(w, z; \theta)\| \le b_1(w) + \sum_{l=1}^L b_{2,l}(z) \quad \text{for all } w \in \mathcal{W} \text{ and } z \in \mathcal{Z}.
		\]
\end{assumption2prime}

We then obtain the following counterpart to Theorem~\ref{thm:id_validation}.

\begin{theorem}\label{thm:id_validation_multi_marginal}
	Suppose Assumptions 1' and 2' hold. Then $\theta \in \Theta_I\!\left(F,(G_l)_{l=1}^L\right)$ if and only if
	\begin{equation}\label{eq:maxmin_characterization_multi_marginal}
		\max_{\lambda \in \mathbb{B}}
		\min_{H'\in \mathcal{H}_L'\left(F,(G_l)_{l=1}^L\right)}
		\expt_{H'}\!\left[\lambda^{\top}\tilde{q}(W,Z,\hat{Z},S,\hat{Z}',S';\theta)\right]
		\le 0.
	\end{equation}
	Moreover, both the maximum and the minimum in \eqref{eq:maxmin_characterization_multi_marginal} are attained.
\end{theorem}

The inner minimization problem in \eqref{eq:maxmin_characterization_multi_marginal} is now a multi-marginal optimal transport problem. For each $l=1,\ldots,L$, let $\mathcal X_{v,l}$ denote the support of $(Z_l,\hat Z_l,S_l)$. This problem admits the dual representation
\begin{align*}
&\min_{H'\in\mathcal H_L'\left(F,(G_l)_{l=1}^L\right)}
\expt_{H'}\!\left[\lambda^{\top}\tilde q(W,Z,\hat Z,S,\hat Z',S';\theta)\right] \\
	=&
\sup_{\psi_l\in \mathcal{L}^1(G_{l}),\, l=1,\ldots,L}
\expt_F\!\left[
\inf_{(z,\hat z',s')\in\mathcal X_v}
\left\{
\lambda^{\top}\tilde q(W,z,\hat Z,S,\hat z',s';\theta)
-
\sum_{l=1}^L \psi_l(z_l,\hat z_l',s_l')
\right\}
\right]
+
\sum_{l=1}^L \expt_{G_l}\!\left[\psi_l(Z_l,\hat Z_l',S_l')\right].
\end{align*}
Consequently, the inequality in \eqref{eq:maxmin_characterization_multi_marginal} is equivalent to
\begin{multline}\label{eq:convex_characterization_multi_marginal}
\mathscr D\!\left(F,(G_l)_{l=1}^L;\theta\right)
\coloneqq
\sup_{\lambda\in\mathbb B,\ \psi_l\in \mathcal{L}^1(G_{l}),\, l=1,\ldots,L}
\Bigg\{\\
\expt_F\!\left[
\inf_{(z,\hat z',s')\in\mathcal X_v}
\left\{
\lambda^{\top}\tilde q(W,z,\hat Z,S,\hat z',s';\theta)
-
\sum_{l=1}^L \psi_l(z_l,\hat z_l',s_l')
\right\}
\right] 
+
\sum_{l=1}^L \expt_{G_l}\!\left[\psi_l(Z_l,\hat Z_l',S_l')\right]
\Bigg\}
\le 0.
\end{multline}

To obtain a computationally tractable approximation, we again replace each infinite-dimensional dual function class with a sieve space. For simplicity, suppose that each sieve has dimension $K$, and define
\[
\mathcal S_{l,K}(\mathcal X_{v,l})
\coloneqq
\left\{
\psi:\ \psi(z_l,\hat z_l',s_l')=\beta_l^\top \varphi_l(z_l,\hat z_l',s_l')
\text{ for some } \beta_l\in\mathbb R^K
\right\},
\]
where $\varphi_l(z_l,\hat z_l',s_l')$ is a vector of $K$ known basis functions on $\mathcal X_{v,l}$. Restricting each dual function $\psi_l$ in \eqref{eq:convex_characterization_multi_marginal} to the corresponding sieve space $\mathcal S_{l,K}(\mathcal X_{v,l})$ yields the finite-dimensional approximation
\begin{multline*}
\mathscr D_K\!\left(F,(G_l)_{l=1}^L;\theta\right)
\coloneqq
\sup_{\lambda\in\mathbb B,\ \beta_l\in\mathbb R^K,\, l=1,\ldots,L}
\Bigg\{\\
\expt_F\!\left[
\inf_{(z,\hat z',s')\in\mathcal X_v}
\left\{
\lambda^{\top}\tilde q(W,z,\hat Z,S,\hat z',s';\theta)
-
\sum_{l=1}^L \beta_l^{\top}\varphi_l(z_l,\hat z_l',s_l')
\right\}
\right]
+
\sum_{l=1}^L \expt_{G_l}\!\left[\beta_l^{\top}\varphi_l(Z_l,\hat Z_l',S_l')\right]
\Bigg\}.
\end{multline*}

The inference procedure developed in Section~\ref{sec:validation_inference} extends straightforwardly to this setting. In essence, one applies the same sample-splitting and cross-fitting strategy as before. One fold is used to solve for the optimizer $(\hat{\lambda},\hat{\beta}_1,\ldots,\hat{\beta}_L)$ in the multi-marginal sieve problem, while the other fold is used to evaluate the corresponding test statistic given $(\hat{\lambda},\hat{\beta}_1,\ldots,\hat{\beta}_L)$. The roles of the two folds are then reversed to obtain a second fold-specific statistic, and the two are aggregated using their maximum. Because these modifications are largely mechanical and introduce no new conceptual issues, we omit the details for brevity.

\section{Discussion}\label{sec:discussion}

The central conceptual innovation of our approach is to treat the ML-generated proxy not as a naive plug-in substitute for the unobserved latent variable, but as a dimension-reduction device that empirically bridges the downstream sample and an auxiliary validation sample. To operationalize this idea, we establish new identification results for data combination and propose a novel, resampling-free inference procedure based on sample splitting and cross-fitting, both of which are of independent econometric interest. 

For the applied researcher, this perspective fully decouples the \emph{validity} of econometric inference from the \emph{statistical consistency} of the upstream ML algorithm. Researchers are free to deploy the most flexible prediction algorithms available---whether they output binary classifications, continuous predicted probabilities, or multiple competing scores---confident that the resulting econometric bounds will remain asymptotically valid. Consequently, empirical practitioners can select whichever ML procedure proves most effective in practice for their specific application, unconstrained by the need for formal statistical guarantees from the prediction step.

For researchers developing upstream ML methods, our framework motivates an alternative criterion for evaluating predictive algorithms. When the ultimate goal is downstream empirical analysis, an optimal prediction rule need not minimize the discrepancy between the prediction and the true target variable. Instead, it should serve as an effective dimension-reduction tool that preserves as much information as possible from the raw, unstructured data regarding the latent target. Although not formally explored here, the design of ML algorithms explicitly tailored for this information-preservation objective represents a promising direction for future research.

Our framework also leaves several important theoretical questions open for future research. First, our baseline identification analysis presumes that the downstream target population and the upstream validation sample exhibit compatible marginal distributions over the overlap variables. In empirical practice, however, the validation data---often a held-out subset of an upstream training corpus---may suffer from covariate shift relative to the downstream population of interest. Formally integrating reweighting techniques, such as inverse probability weighting, into the optimal transport characterization to explicitly account for sample selection and distribution shift remains a highly relevant theoretical extension. Second, the cross-fitted inference procedure developed in this paper provides valid tests for the full structural parameter vector. While applied researchers can construct valid, conservative confidence sets for individual parameters of interest (e.g., a specific treatment effect) via projection, this approach sacrifices statistical power. Developing a localized, profiling-based inference theory that achieves sharp, non-conservative bounds for subvectors within this optimal transport framework is a challenging but vital next step.

More broadly, incorporating ML-generated proxies into empirical economic analysis is an important and increasingly relevant topic. By combining tools from optimal transport, partial identification, and cross-fitting, the framework developed here offers a tractable and versatile way to do so. It allows researchers to use complex machine-learned measures while remaining agnostic about the fine details of the upstream prediction step under mild data requirement. We hope that this perspective proves useful both for applied work that relies on ML-based measurement and for future research at the intersection of econometrics and machine learning.

\newpage 
\begin{appendix}
	\section{Proof for Theorem \ref{thm:id_validation}}\label{sec:proof_id_validation}
	Let $\theta \in \Theta$ be fixed. To streamline notation, let $X_d \coloneqq (W, \hat{Z}, S)$ and $X_v \coloneqq (Z, \hat{Z}', S')$. We define the bilinear objective functional $\Phi : \mathcal{H}'(F, G) \times \mathbb{B} \to \real$ as
	\[
		\Phi(H', \lambda) \coloneqq \expt_{H'}\big[\lambda^{\top}\tilde{q}(X_d, X_v; \theta)\big].
	\]
	The proof proceeds in three steps.
	
	\noindent\textbf{Step 1. }
	By Definition~\ref{def:alias_id_set}, $\theta \in \Theta'_I(F, G)$ if and only if there exists a joint distribution $H^* \in \mathcal{H}'(F, G)$ such that $\expt_{H^*}[\tilde{q}(X_d, X_v; \theta)] = \mathbf{0}$. 
	
	For any coupling $H'$, let $\mu(H') \coloneqq \expt_{H'}[\tilde{q}(\cdot; \theta)]$. Because $\mathbb{B}$ is a compact convex set containing an open neighborhood of the origin, there exists an $\epsilon > 0$ such that the closed ball of radius $\epsilon$ centered at the origin is entirely contained in $\mathbb{B}$. If $\mu(H') \neq \mathbf{0}$, we can select $\lambda^* = \epsilon \frac{\mu(H')}{\|\mu(H')\|} \in \mathbb{B}$, which yields
	\[
		\sup_{\lambda \in \mathbb{B}} \lambda^{\top}\mu(H') \ge (\lambda^*)^{\top}\mu(H') = \epsilon \|\mu(H')\| > 0.
	\]
	Conversely, if $\mu(H') = \mathbf{0}$, then trivially $\sup_{\lambda \in \mathbb{B}} \lambda^\top \mu(H') = 0$. Because $\mathbf{0} \in \mathbb{B}$, the supremum is always bounded below by zero. Thus, $\sup_{\lambda \in \mathbb{B}} \Phi(H', \lambda) \le 0$ if and only if exactly $\mu(H') = \mathbf{0}$. Therefore, $\theta \in \Theta'_I(F, G)$ if and only if there exists an $H^* \in \mathcal{H}'(F, G)$ such that $\sup_{\lambda \in \mathbb{B}} \Phi(H^*, \lambda) \le 0$, i.e.,
	\begin{equation}\label{eq:proof_min_sup}
		\inf_{H'\in \mathcal{H}'(F, G)} \sup_{\lambda \in \mathbb{B}} \, \Phi(H', \lambda) \le 0.
	\end{equation}
	
	\noindent\textbf{Step 2. }
	Because $\mathcal{X}_d$ and $\mathcal{X}_v$ are closed subsets of Euclidean spaces (Assumption~\ref{assu:regularity_minimax}(i)), the marginal probability measures $F$ and $G$ are tight. By Prokhorov's Theorem, the set of all joint probability measures on $\mathcal{X}_d \times \mathcal{X}_v$ with fixed tight marginals is relatively compact in the topology of weak convergence. Because the marginal constraints defining the Fr\'{e}chet class $\mathcal{H}'(F, G)$ are closed under weak convergence, $\mathcal{H}'(F, G)$ is weakly closed. Thus, $\mathcal{H}'(F, G)$ is a convex and weakly compact set. Moreover, the parameter space $\mathbb{B} \subset \real^{d_q + d_{\hat{Z}} + d_S}$ is, by definition, a compact and convex set.
	
	Clearly, $\Phi(H', \lambda)$ is linear (hence concave and convex) in both $H'$ and $\lambda$. For a fixed $H'$, $\Phi(\cdot, \lambda)$ is trivially continuous in $\lambda$. We must show that for a fixed $\lambda$, the mapping $H' \mapsto \Phi(H', \lambda)$ is continuous with respect to the weak topology on $\mathcal{H}'(F, G)$.
	
	By Assumption~\ref{assu:regularity_minimax}(ii), the augmented moment vector $\tilde{q}(X_d, X_v; \theta)$ is continuous in the topology of weak convergence. To establish weak continuity of its expectation over the Fr\'{e}chet class, we must bound it by an integrable envelope. Since all norms on finite-dimensional Euclidean spaces are equivalent, it suffices to bound the $L_1$ norm of $\tilde{q}$, which is simply the sum of the absolute values of its components:
	\[
		\|\tilde{q}(X_d, X_v; \theta)\|_1 = \|q(W, Z; \theta)\|_1 + \|\hat{Z} - \hat{Z}'\|_1 + \|S - S'\|_1.
	\]
	Applying the basic inequality $\|a - b\|_1 \le \|a\|_1 + \|b\|_1$ and the structural envelope bound from Assumption~\ref{assu:regularity_minimax}(iv) (which, by norm equivalence, implies a similar bound for the $L_1$ norm of $q$, up to a generic constant $C > 0$), we can construct an additively separable continuous bounding function $A(X_d) + B(X_v)$ such that $\|\tilde{q}(X_d, X_v; \theta)\|_1 \le A(X_d) + B(X_v)$, where
	\begin{align*}
		A(W, \hat{Z}, S) &\coloneqq C [b_1(W) + \|\hat{Z}\| + \|S\|], \\
		B(Z, \hat{Z}', S') &\coloneqq C [ b_2(Z) + \|\hat{Z}'\| + \|S'\|].
	\end{align*}
	Under Assumption~\ref{assu:consistent_validation}, the marginal distributions of $(\hat{Z}, S)$ under $F$ and $(\hat{Z}', S')$ under $G$ are identical. Consequently, Assumption~\ref{assu:regularity_minimax}(iii) implies that both sets of variables possess finite second moments under their respective marginals. It follows immediately that $A \in \mathcal{L}^1(F)$ and $B \in \mathcal{L}^1(G)$. 
	
	Because the marginals of any $H' \in \mathcal{H}'(F, G)$ are exactly $F$ and $G$, the expected value of this bounding envelope is identical and finite for every coupling in the Fr\'{e}chet class. By standard optimal transport theory \autocite[e.g., Lemma 4.3 in][]{villani_optimal_2009}, for any $\lambda\in \mathbb{B}$, $\Phi(H',\lambda)$ is lower semicontinuous in $H'$ with respect to weak convergence. Since $\Phi(H', -\lambda) = -\Phi(H',\lambda)$. This also implies that for any $\lambda\in \mathbb{B}$, $\Phi(H',\lambda)$ is upper semicontinuous in $H'$ with respect to weak convergence.  Thus, $H' \mapsto \Phi(H', \lambda)$ is weakly continuous.
	
\noindent\textbf{Step 3.}
	We have established that $\mathcal{H}'(F, G)$ and $\mathbb{B}$ are compact convex sets, and that $\Phi(H', \lambda)$ is a continuous bilinear functional. Therefore, all conditions for Sion's Minimax Theorem \textcite{sion_general_1958} are satisfied. This yields two crucial results: the infimum over $H'$ can be attained (i.e., $\inf$ is a $\min$), and we can switch the order of these operators:
	\begin{equation}\label{eq:proof_minimax_swap}
		\min_{H'\in \mathcal{H}'(F, G)} \max_{\lambda \in \mathbb{B}} \, \Phi(H', \lambda) 
		= 
		\max_{\lambda \in \mathbb{B}} \min_{H'\in \mathcal{H}'(F, G)} \, \Phi(H', \lambda).
	\end{equation}
	Substituting \eqref{eq:proof_minimax_swap} into \eqref{eq:proof_min_sup} yields that $\theta \in \Theta'_I(F, G)$ if and only if
	\[
		\max_{\lambda \in \mathbb{B}} \min_{H'\in \mathcal{H}'(F, G)} \expt_{H'}\big[\lambda^{\top}\tilde{q}(X_d, X_v; \theta)\big] \le 0.
	\]
	Since we established in the main text that $\Theta'_I(F, G) = \Theta_I(F, G)$, this completes the proof.

	\section{Proof for Lemma \ref{lem:kantorovich_duality}}
\begin{proof}
Fix an arbitrary $\theta\in \Theta$ throughout the proof.  For notational simplicity, let $x_d\coloneqq (w, \hat{z}, s) \in \mathcal{X}_d$ and $x_v\coloneqq (z, \hat{z}', s')\in \mathcal{X}_v$.  For any $\lambda\in \mathbb{B}$, define the optimal transport cost function as 
\begin{equation*}
c_{\theta, \lambda}(x_d, x_v) \coloneqq \lambda^\top \tilde{q}(x_d, x_v;\theta) 
\end{equation*}
Let $V(\lambda)\coloneqq  \min_{H'\in \mathcal{H}'(F, G)}\expt_{H'} c_{\theta, \lambda}(X_d, X_v)$.  We split the proof into two parts.

\noindent\textbf{Part 1: prove equation \eqref{eq:kantorovich_duality_application} hold.} For any $\lambda\in \mathbb{B}$, the cost fucntion is bounded below by 
\begin{equation*}
c_{\theta,\lambda}(x_d, x_v) \ge -M\left( \norm{q(w,z;\theta)} + \sum_{j=1}^{d_{\hat{Z}}}|\hat{z}_j| + M\sum_{j=1}^{d_{\hat{Z}}}|\hat{z}'_j| + \sum_{k=1}^{d_S}|s| + \sum_{k=1}^{d_S}|s'|\right) 
\end{equation*}
where $M\coloneqq \sup_{\lambda\in \mathbb{B}}\norm{\lambda} < \infty$. By Assumption ~\ref{assu:regularity_minimax}\ref{enu:envelop}, $\norm{q(w,z;\theta)} \le b_1(w) + b_2(z)$. Thus,
\begin{equation}\label{eq:aomnwfeiwfe}
\forall \lambda\in \mathbb{B},\quad c_{\theta, \lambda} \ge  -M\left( b_1(w) + b_2(z) + \sum_{j=1}^{d_{\hat{Z}}}|\hat{z}_j| + M\sum_{j=1}^{d_{\hat{Z}}}|\hat{z}'_j| + \sum_{k=1}^{d_S}|s| + \sum_{k=1}^{d_S}|s'|\right) 
\end{equation}
Because $b_1 \in \mathcal{L}^1(F)$, $b_2\in \mathcal{L}^1(G)$ and they are continuous, this establishes that for all $\lambda\in \mathbb{B}$, the continuous cost function $c_{\theta, \lambda}$ is bounded from below by an additively separable continuous integrable envelope. By Theorem 5.10 (\emph{i}) in \textcite{villani_optimal_2009} (see, also, Theorem 3.1 in \textcite{ambrosio_lectures_2021}), we know that for any $\lambda\in \mathbb{B}$,
\begin{eqnarray*}
	&&V(\lambda) \\
	= &&\sup_{\phi\in \mathcal{C}(\mathcal{X}_d), \psi\in \mathcal{C}(\mathcal{X}_v)} \expt_F \phi(X_d) + \expt_G \psi(X_v)\\
							 & s.t.& \phi(x_d) + \psi(x_v) \le c_{\theta, \lambda}(x_d, x_v), \forall x_d\in \mathcal{X}_d,\forall x_v\in \mathcal{X}_v
\end{eqnarray*}
Since $\phi(x_d) + \psi(x_v) \le c_{\theta, \lambda}(x_d, x_v), \forall x_d\in \mathcal{X}_d,\forall x_v\in \mathcal{X}_v$ implies that for any $x_d \in \mathcal{X}_d,$
\begin{equation*}
\phi(x_d) \le  \inf_{x_v \in\mathcal X_v} \left\{ c_{\theta, \lambda}(x_d, x_v) - \psi(x_v) \right\},
\end{equation*}
we know that for any $\lambda\in \mathbb{B}$,
\begin{equation*}
V(\lambda) \le \sup_{\psi\in \mathcal{C}(\mathcal{X}_v)} \expt_F \left\{ \inf_{x_v \in\mathcal X_v} \left\{ c_{\theta, \lambda}(X_d, x_v) - \psi(x_v) \right\} \right\}  + \expt_G \psi(X_v)
\end{equation*}
Moreover, by construction, $ \inf_{x_v \in\mathcal X_v} \left\{ c_{\theta, \lambda}(x_d, x_v) - \psi(x_v) \right\} + \psi(x_v) \le c_{\theta, \lambda}(x_d, x_v)$ for all $x_d\in \mathcal{X}_d$ and $x_v \in \mathcal{X}_v$. Hence, for any $\lambda\in \mathbb{B}$,
\begin{equation}\label{eq:qo34infaq}
V(\lambda) = \sup_{\psi\in \mathcal{C}(\mathcal{X}_v)} \expt_F \left\{ \inf_{x_v \in\mathcal X_v} \left\{ c_{\theta, \lambda}(X_d, x_v) - \psi(x_v) \right\} \right\}  + \expt_G \psi(X_v),
\end{equation}
which completes the proof for \eqref{eq:kantorovich_duality_application}.

\noindent\textbf{Part 2: Attainment.}
By the same argument as in the proof for Theorem \ref{thm:id_validation}, we know there exists $\lambda^*\in \mathbb{B}$ such that 
\begin{equation*}
		\max_{\lambda \in \mathbb{B}} \min_{H'\in \mathcal{H}'(F, G)} \expt_{H'}\big[\lambda^{\top}\tilde{q}(W, Z, \hat{Z}, S, \hat{Z}', S'; \theta)\big]  =  V(\lambda^*).
\end{equation*}
For any function $\psi: \mathcal{X}_v \to \real$, its $c$-transform $\psi^c(\cdot): \mathcal{X}_d \to \real \cup \{-\infty\}$ is defined as
\[
    \psi^c(x_d) \coloneqq \inf_{x_v \in \mathcal{X}_v} \big\{ c_{\theta,\lambda^*}(x_d, x_v) - \psi(x_v) \big\}.
\]
Analogously, for a function $\phi: \mathcal{X}_d \to \real$, its $c$-transform $\phi^c(\cdot;\lambda): \mathcal{X}_v \to \real \cup \{-\infty\}$ is defined as
\[
    \phi^c(x_v) \coloneqq \inf_{x_d \in \mathcal{X}_d} \big\{ c_{\theta,\lambda^*}(x_d, x_v) - \phi(x_d) \big\}.
\]
With this notation, \eqref{eq:qo34infaq} implies that
\begin{equation*}
V(\lambda^*) = \sup_{\psi\in \mathcal{C}_b(\mathcal{X}_v)} \expt_F \psi^c(X_d;\lambda) + \expt_G \psi(X_v).
\end{equation*}
Given \eqref{eq:aomnwfeiwfe}, Theorem 5.10(\emph{iii}) in \textcite{villani_optimal_2009} implies that
\begin{equation}\label{eq:nfoq2i394r}
    V(\lambda^*) = \max_{\psi\in \mathcal{L}^1(G)}\expt_F[\psi^c(X_d)] + \expt_G[\psi(X_v)],
\end{equation}
where the $\max$ is attained at some $\psi \in \mathcal{L}^1(G)$. 

Now, assume $q(\cdot;\theta)$ is uniformly continuous in $\mathcal{X}_w \times \mathcal{X}_z$. Since functions like $|\hat{z}_j - \hat{z}'_j|$ and $|s_k - s'_k|$ is Lipschitz continuous in $\mathcal{X}_{d}\times \mathcal{X}_{v}$. We know $\tilde{q}(\cdot;\theta)$ is uniformly continuous and, hence,  $c_{\theta,\lambda^*}(x_d, x_v)$ is also uniformly continuous. 

Let $\psi \in \mathcal{L}^1(G)$ be the optimal potential that attains the maximum value in \eqref{eq:nfoq2i394r}, and consider the double $c$-transform $\psi^{cc} \coloneqq (\psi^c)^c$, which maps from $\mathcal{X}_v$ to $\real \cup \{-\infty\}$. A fundamental property of $c$-transforms is that $\psi^{cc} \ge \psi$, and $(\psi^{cc})^c = \psi^c$. Consequently, replacing $\psi$ with $\psi^{cc}$ weakly increases the value of $\expt_G[\psi(X_v)]$ without changing $\expt_F[\psi^c(X_d)]$. Thus, $\psi^{cc}$ is also an optimal dual potential.

We now show that $\psi^{cc}$ is continuous. For any $x_v, y_v \in \mathcal{X}_v$, we have
\begin{align*}
    \psi^{cc}(x_v) - \psi^{cc}(y_v) &= \inf_{x_d \in \mathcal{X}_d} \big\{ c_{\theta,\lambda^*}(x_d, x_v) - \psi^c(x_d) \big\} - \inf_{x_d \in \mathcal{X}_d} \big\{ c_{\theta,\lambda^*}(x_d, y_v) - \psi^c(x_d) \big\} \\
    &\le \sup_{x_d \in \mathcal{X}_d} \big\{ c_{\theta,\lambda^*}(x_d, x_v) - c_{\theta,\lambda^*}(x_d, y_v) \big\}.
\end{align*}
By symmetry, $|\psi^{cc}(x_v) - \psi^{cc}(y_v)| \le \sup_{x_d \in \mathcal{X}_d} |c_{\theta,\lambda}(x_d, x_v) - c_{\theta,\lambda^*}(x_d, y_v)|$. Because $c_{\theta,\lambda^*}(x_d, x_v)$ is uniformly continuous on $\mathcal{X}_d \times \mathcal{X}_v$, it admits a modulus of continuity $\omega(\cdot)$ independent of $x_d$. Thus,
\[
    \big|\psi^{cc}(x_v) - \psi^{cc}(y_v)\big| \le \omega(\|x_v - y_v\|).
\]
This establishes that $\psi^{cc}$ is uniformly continuous. Therefore, $\psi^{cc} \in \mathcal{C}(\mathcal{X}_v)$, and the maximum is attained in this space.  This completes the proof.
\end{proof}

\section{Proof for Lemma \ref{lem:sieve_approx}}

\begin{proof}
We first establish \eqref{eq:sieve_validity}. By assumption, each basis function $\varphi_k(\cdot)$ for $k=1,\dots,K$ is continuous and integrable with respect to $G$. Thus, any function $\psi$ in the linear sieve space $\mathcal{S}_K(\mathcal{X}_v)$ is also continuous and $G$-integrable, which implies that $\mathcal{S}_K(\mathcal{X}_v) \subseteq \mathcal{C}(\mathcal{X}_v)$. Because $\mathscr D_K(F,G;\theta)$ maximizes the exact same objective function as $\mathscr D(F,G;\theta)$ but over a restricted set of dual functions, it immediately follows that $\mathscr D_K(F,G;\theta) \le \mathscr D(F,G;\theta)$.

Next, we establish \eqref{eq:sieve_power} under Assumption \ref{assu:sieve_approximation}. By the definition of $\mathscr{D}(F, G;\theta)$, there must exists a sequence $\{(\lambda_n, \psi_n): n=1,2,...\}$ such that, $\forall n \ge 1$, $\lambda_n \in \mathbb{B}$ and $\psi_n\in \mathcal{C}(\mathcal{X}_v)$, and 
\begin{equation}\label{eq:qio44frjn}
	\mathscr{D}(F, G;\theta) - \frac{1}{n} \le \mathcal{D}_n(\lambda_n, \psi_n) \le \mathscr{D}(F, G;\theta) 
\end{equation}
where
\begin{equation*}
\mathcal{D}_n(\lambda_n, \psi_n)\coloneqq
\expt_F\!\left[
\inf_{(z,\hat z',s')\in\mathcal X_v}
\left\{
\lambda_n^{\top}\tilde q(W,z,\hat Z,S,\hat z',s';\theta)
-
\psi_n(z,\hat z',s')
\right\}
\right] 
+
\expt_G\!\left[\psi_n(Z,\hat Z',S')\right]
\end{equation*}
Under Assumption \ref{assu:sieve_approximation}, there must exist some a sequence  $\{(K_n, \psi'_n):n=1,2,...\}$ such that (\emph{i}) $K_n$ is weakly increasing in $n$; (\emph{ii}) for each $n\ge 1$, $\psi'_n \in \mathcal{S}_{K_n}(\mathcal{X}_v)$, and $\norm{\psi_n - \psi'_n}_\infty < \frac{1}{n}$. Thus, for any $(z, \hat{z}', s')\in \mathcal{X}_v$,
\begin{equation*}
	| \{\lambda_n^{\top}\tilde q(W,z,\hat Z,S,\hat z',s';\theta) - \psi_n(z,\hat z',s')\} - 
 \{\lambda_n^{\top}\tilde q(W,z,\hat Z,S,\hat z',s';\theta) - \psi'_n(z,\hat z',s')\} | < \frac{1}{n}.
\end{equation*}
As a result,
\begin{multline*}
\Bigg| \expt_F\!\left[
\inf_{(z,\hat z',s')\in\mathcal X_v}
\left\{
\lambda_n^{\top}\tilde q(W,z,\hat Z,S,\hat z',s';\theta)
-
\psi_n(z,\hat z',s')
\right\}
\right] - \\
\expt_F\!\left[
\inf_{(z,\hat z',s')\in\mathcal X_v}
\left\{
\lambda_n^{\top}\tilde q(W,z,\hat Z,S,\hat z',s';\theta)
-
\psi'_n(z,\hat z',s')
\right\}
\right] \Bigg| \le \frac{1}{n}. 
\end{multline*}
This implies that 
\begin{equation}\label{eq:gnq9io4r}
|\mathcal{D}_n(\lambda_n, \psi_n) - \mathcal{D}_n(\lambda_n, \psi'_n)| \le \frac{2}{n}.
\end{equation}
On the other hand, since $\psi'_n \in \mathcal{S}_{K_n}(\mathcal{X}_v)$, we must have 
\begin{equation}\label{eq:ngvqo2i43}
\mathcal{D}_n(\lambda_n, \psi'_n) \le \mathscr{D}_{K_n}(F, G;\theta)
\end{equation}
Combining \eqref{eq:qio44frjn}-\eqref{eq:ngvqo2i43}, we must have 
\begin{equation*}
\mathscr{D}(F, G;\theta) - \frac{3}{n} \le \mathscr{D}_{K_n}(F, G;\theta).
\end{equation*}
Since we have already proved that $ \mathscr{D}_{K}(F, G;\theta) \le \mathscr{D}(F, G;\theta)$ for any $K$, we know 
\begin{equation*}
\lim_{n\to\infty} \mathscr{D}_{K_n}(F, G;\theta) = \mathscr{D}(F, G;\theta).
\end{equation*}
Since $\mathscr{D}_{K}(F, G;\theta)$ is weakly increasing in $K$, the above result implies \eqref{eq:sieve_power}.

\end{proof}

\section{Proof of Theorem \ref{thm:asymptotic_size}}\label{app:proof_thm_size}
\begin{proof}
Fix a parameter $\theta \in \Theta_I(F, G)$ and a nominal level $\alpha \in (0,1)$. Let $z_{1-\alpha/2} \coloneqq \Phi^{-1}(1-\alpha/2) > 0$. The decision rule rejects $H_0$ if $p(\theta) < \alpha$, which holds if and only if $\max\{T_1(\theta), T_2(\theta)\} > z_{1-\alpha/2}$. Because
\begin{equation}\label{eq:proof_union_bound}
\prob\Big(\max\big\{T_1(\theta),T_2(\theta)\big\}>z_{1-\alpha/2}\Big)
\le
\prob\big(T_1(\theta)>z_{1-\alpha/2}\big)
+
\prob\big(T_2(\theta)>z_{1-\alpha/2}\big),
\end{equation}
it suffices to show that for $m\in\{1,2\}$ and any fixed $c>0$,
\[
\limsup_{\underline n\to\infty}\prob\big(T_m(\theta)>c\big)\le 1-\Phi(c).
\]
Without loss of generality, we focus on $T_2(\theta)$. The proof for $T_1(\theta)$ is symmetric.

Let $\mathfrak S_1$ denote the $\sigma$-algebra generated by all observations in fold 1. Conditional on $\mathfrak S_1$, the pair $(\hat\lambda_{1,n},\hat\beta_{1,n})$ is deterministic. For any compact convex set $\mathbb A\subset\mathbb R^{d_q+d_{\hat Z}+d_S}$ containing an open neighborhood of the origin, define
\begin{multline*}
\mathscr D_{K,\mathbb A}(F,G;\theta)
\coloneqq
\sup_{\lambda\in\mathbb A,\ \beta\in\mathbb R^{K_n}}
\Bigg\{
\expt_F\!\left[
\inf_{(z,\hat z',s')\in\mathcal X_v}
\Big(
\lambda^\top\tilde q(W,z,\hat Z,S,\hat z',s';\theta)
-
\beta^\top\varphi(z,\hat z',s')
\Big)
\right]
\\
+
\expt_G\!\left[\beta^\top\varphi(Z,\hat Z',S')\right]
\Bigg\}.
\end{multline*}
Since $\theta\in \Theta_I(F, G)$, conditional on $\mathfrak S_1$, the weak duality implies that the population objective evaluated at the feasible point $(\hat\lambda_{1,n},\hat\beta_{1,n})$ satisfies
\begin{multline}
\mu_n
\coloneqq
\expt_F\!\left[
\inf_{(z,\hat z',s')\in\mathcal X_v}
\Big\{
\hat\lambda_{1,n}^{\top}\tilde q(W,z,\hat Z,S,\hat z',s';\theta)
-
\hat\beta_{1,n}^{\top}\varphi(z,\hat z',s')
\Big\}
\right]
\\
+
\expt_G\!\left[\hat\beta_{1,n}^{\top}\varphi(Z,\hat Z',S')\right]
\le 0
\label{eq:mu_n_nonpositive}
\end{multline}
almost surely.

Let
\[
S_n
\coloneqq
\sqrt{\underline n}\big(\widehat{\mathscr D}_2(\theta)-\mu_n\big).
\]
Because $\mu_n\le 0$, we have $\sqrt{\underline n}\widehat{\mathscr D}_2(\theta)\le S_n$. Therefore,
\begin{equation}\label{eq:proof_T2_bound}
T_2(\theta)
\le
R_n
\coloneqq
\frac{S_n}{\sqrt{\max\{\underline n\hat V_2(\theta),\,\epsilon\}}}.
\end{equation}

Conditional on $\mathfrak S_1$, $S_n$ is the sum of two independent sample averages centered at their conditional means. Let
\[
\sigma_n^2
\coloneqq
\var(S_n\mid\mathfrak S_1)
=
\frac{\underline n}{|\mathcal I^d_2|}\sigma_{u,n}^2
+
\frac{\underline n}{|\mathcal I^v_2|}\sigma_{v,n}^2,
\]
where
\[
\sigma_{u,n}^2\coloneqq \var_F(\hat u_i\mid\mathfrak S_1),
\qquad
\sigma_{v,n}^2\coloneqq \var_G(\hat v_j\mid\mathfrak S_1).
\]
Define the regularized conditional variance
\[
\tau_n^2\coloneqq \max\{\sigma_n^2,\epsilon\}.
\]
By construction, $\tau_n^2\ge\epsilon>0$ and $\sigma_n^2/\tau_n^2\le 1$.

Next, we uniformly bound the absolute centered magnitudes of the components comprising $S_n$.  Let
\[
\zeta_n
\coloneqq
\sup_{(z,\hat z',s')\in\mathcal X_v}\|\varphi(z,\hat z',s')\|.
\]
Under Assumption~\ref{assu:inference_regularity}\ref{enu:sieve}, we have $\zeta_n\le C\sqrt{K_n}$.  For the validation component $\hat{v}_j$, we have
\[
|\hat v_j|
\le
\|\hat\beta_{1,n}\|\,\|\varphi(Z_j,\hat Z'_j,S'_j)\|
\le
c_n\zeta_n.
\]
For the downstream component, because 
\begin{equation*}
\hat\Sigma_{1,n}\to\Sigma \coloneqq \begin{pmatrix}
	\Omega & \\
				 & \Lambda
\end{pmatrix}
\end{equation*}
is almost surely and $\Sigma$ is positive definite, the multipliers $\hat\lambda_{1,n}\in\mathbb B_{1,n}$ are uniformly almost surely bounded for all sufficiently large $n$. Combined with the continuity of $q(\cdot)$ and the compactness of $\mathcal X_d\times\mathcal X_v$, there exists a finite constant $M>0$ such that
\[
\sup_{(w,\hat z,s)\in\mathcal X_d}\sup_{(z,\hat z',s')\in\mathcal X_v}
\left|
\hat\lambda_{1,n}^{\top}\tilde q(w,z,\hat z,s,\hat z',s';\theta)
\right|
\le M
\]
almost surely for all sufficiently large $n$. Therefore,
\begin{equation}\label{eq:4iun09213}
|\hat u_i| \le M+c_n\zeta_n.
\end{equation}
Under Assumption~\ref{assu:inference_regularity}\ref{enu:fold_size}, for sufficiently large $\underline n$,
\[
|\mathcal I^d_2|\ge (\kappa_d/2)n_d,
\qquad
|\mathcal I^v_2|\ge (\kappa_v/2)n_v.
\]
Since $\underline n\le n_d$, the downstream scaling factor satisfies:
\begin{equation*}
	\frac{\sqrt{\underline{n}}}{|\mathcal{I}^d_2|} \le \frac{\sqrt{\underline{n}}}{(\kappa_d / 2) n_d} \le \frac{\sqrt{n_d}}{(\kappa_d / 2) n_d}  = \frac{1}{(\kappa_d /2) \sqrt{n_d}} \le \frac{1}{(\kappa_d/2) \sqrt{\underline{n}}} = \mathcal O(\underline n^{-1/2})
\end{equation*}
An identifical symmetric bound of $\mathcal O(\underline n^{-1/2})$ applies to the validation factor $\sqrt{\underline{n}}/ |\mathcal{I}^v_2|$.  Hence, if we write
\begin{equation}\label{eq:fniuoi3}
S_n
=
\sum_{i\in\mathcal I^d_2} X_i
+
\sum_{j\in\mathcal I^v_2} Y_j,
\end{equation}
where
\[
X_i
\coloneqq
\frac{\sqrt{\underline n}}{|\mathcal I^d_2|}
\Big(
\hat u_i-\expt_F[\hat u_i\mid\mathfrak S_1]
\Big),
\qquad
Y_j
\coloneqq
\frac{\sqrt{\underline n}}{|\mathcal I^v_2|}
\Big(
\hat v_j-\expt_G[\hat v_j\mid\mathfrak S_1]
\Big),
\]
then the centered variables $X_i$ and $Y_j$ are uniformly bounded by a nonrandom sequence
\[
B_n=\mathcal O\!\left(\frac{c_n\zeta_n}{\sqrt{\underline n}}\right)
\]
conditional on $\mathfrak{S}_1$ almost surely.

	Next, we evaluate the uniform consistency of the regularized sample variance estimator $\underline{n}\hat{V}_2(\theta)$. Conditional on $\mathfrak{S}_1$, $\underline{n}\hat{V}_2(\theta)$ is a weighted sum of two sample variances. Focusing on the downstream component, its conditional variance is bounded by the standard property of sample variances of bounded independent random variables: 
\begin{equation*}
\mathrm{Var}\left(\frac{\underline{n}}{|\mathcal{I}^d_2|} \hat{\sigma}_{u,2}^2 \Big| \mathfrak{S}_1\right) \le \left(\frac{\underline{n}}{|\mathcal{I}^d_2|}\right)^2 \frac{\expt_F[(\hat{u}_i - \expt_F[\hat{u}_i \mid \mathfrak{S}_1])^4 \mid \mathfrak{S}_1]}{|\mathcal{I}^d_2|}.
\end{equation*}
Because of \eqref{eq:4iun09213}, we know $| \hat{u}_i - \expt_F[\hat{u}_i \mid \mathfrak{S}_1] | \le \mathcal{O}(c_n \zeta_n)$, so that the fourth central moment is bounded by $\mathcal{O}(c_n^2 \zeta_n^2) \sigma_u^2$. Thus, the variance of this term is bounded by $\big(\frac{\underline{n}}{|\mathcal{I}^d_2|}\big)^2 \frac{ \mathcal{O}(c_n^2 \zeta_n^2)}{|\mathcal{I}^d_2|} \sigma_u^2$. Because $|\mathcal{I}^d_2| \ge (\kappa_d/2)\underline{n}$ for large $\underline{n}$, we have $\frac{\underline{n}}{|\mathcal{I}^d_2|^2} \le \mathcal{O}(\underline{n}^{-1})$. Hence, this variance is bounded by $\mathcal{O}\big(\frac{c_n^2 \zeta_n^2}{\underline{n}}\big) \big(\frac{\underline{n}}{|\mathcal{I}^d_2|} \sigma_u^2\big)$. Applying an identical argument to the validation component $\hat{v}_j$ yields the combined bound: $\mathrm{Var}(\underline{n}\hat{V}_2(\theta) \mid \mathfrak{S}_1) \le \mathcal{O}\big( \frac{c_n^2 \zeta_n^2}{\underline{n}} \big) \sigma_n^2$. 
	
	By Chebyshev's inequality, we bound the relative deviation of $\underline{n}\hat{V}_2(\theta)$ from $\sigma_n^2$:
	\begin{equation}
		\prob\left( \frac{|\underline{n}\hat{V}_2(\theta) - \sigma_n^2|}{\tau_n^2} > \delta \mathrel{\Bigg|} \mathfrak{S}_1 \right) \le \frac{\mathrm{Var}(\underline{n}\hat{V}_2(\theta) \mid \mathfrak{S}_1)}{\delta^2 \tau_n^4} \le \mathcal{O}\left( \frac{c_n^2 \zeta_n^2}{\underline{n} \delta^2} \frac{\sigma_n^2}{\tau_n^4} \right).
	\end{equation}
	Because $\sigma_n^2 / \tau_n^2 \le 1$ and $1 / \tau_n^2 \le 1 / \epsilon$, the bound simplifies to $\mathcal{O}\big( \frac{c_n^2 \zeta_n^2}{\underline{n} \epsilon \delta^2} \big)$. Under Assumption~\ref{assu:inference_regularity}\ref{enu:sieve}, $c_n \zeta_n / \sqrt{\underline{n}} \to 0$. Because $\epsilon$ is a fixed positive constant, the entire bound vanishes, yielding $\underline{n}\hat{V}_2(\theta) - \sigma_n^2 = o_p(\tau_n^2)$. Because the function $f(x) = \max\{x, \epsilon\}$ is $1$-Lipschitz, the absolute deviation $|\max\{\underline{n}\hat{V}_2(\theta), \epsilon\} - \tau_n^2|$ is bounded by $|\underline{n}\hat{V}_2(\theta) - \sigma_n^2| = o_p(\tau_n^2)$. Consequently, we know that the ratio 
\begin{equation*}
\frac{\max\{\underline{n}\hat{V}_2(\theta), \epsilon\}}{\tau_n^2}  \xrightarrow{p} 1
\end{equation*}
conditionally on $\mathfrak{S}_1$.
	
	By Slutsky's Theorem, the conditional asymptotic distribution of $T_2(\theta)$ is governed entirely by the regularized sum $S_n / \tau_n$. To establish normality, we verify the Lyapunov condition for the row-independent triangular array. Given \eqref{eq:fniuoi3}, let $\Gamma_n$ denote the sum of the absolute third central moments of the components comprising $S_n$, i.e., 
	\begin{equation*}
	\Gamma_n \coloneqq \sum_{i\in \mathcal{I}^d_2}\expt[|X_i|^3] + \sum_{j\in \mathcal{I}^v_2}\expt[|Y_j|^3].
	\end{equation*}
	Because $|X_i|^3 \le \max|X_i| \cdot X_i^2$, $|Y_i|^3 \le \max|Y_i| \cdot Y_i^2$, we bound the third moment sum by factoring out the maximum absolute deviation $B_n$, leaving the sum of the conditional variances $\sigma_n^2$. The Lyapunov ratio therefore  satisfies:
	\begin{equation}
		\frac{\Gamma_n}{\tau_n^3} \le \frac{B_n \sigma_n^2}{\tau_n^3} = B_n \left(\frac{\sigma_n^2}{\tau_n^2}\right) \frac{1}{\tau_n} \le \mathcal{O}\left( \frac{c_n \zeta_n}{\sqrt{\underline{n}}} \right) \frac{1}{\sqrt{\epsilon}} = \mathcal{O}\left( \frac{c_n \zeta_n}{\sqrt{\underline{n}}} \right) \to 0.
	\end{equation}
	Let us then invoke Lyapunov central limit theorem for subsequences. For any arbitrary sequence of sample sizes, the bounded variance ratio $\sigma_n^2 / \tau_n^2 \in [0,1]$ must contain a convergent sub-subsequence where $\sigma_{n_k}^2 / \tau_{n_k}^2 \to \rho^2 \in [0,1]$.
	
	By the Lyapunov Central Limit Theorem, along this sub-subsequence, $S_{n_k} / \tau_{n_k} \xrightarrow{d} \mathcal{N}(0, \rho^2)$ conditional on $\mathfrak{S}_1$. (If $\rho^2 = 0$, this yields a degenerate limit at $0$, which trivially bounds the exceedance probability). For all $\rho^2 \in [0,1]$, the limit distribution $\mathcal{N}(0, \rho^2)$ is stochastically dominated by $\mathcal{N}(0,1)$. Hence, along any convergent sub-subsequence, the conditional exceedance probability satisfies $\lim_{k \to \infty} \prob(T_2(\theta) > c \mid \mathfrak{S}_1) \le 1 - \Phi(c)$. Since this holds for every possible subsequence, the entire sequence converges uniformly almost surely:
	\begin{equation}\label{eq:proof_conditional_limit}
		\limsup_{\underline{n} \to \infty} \, \prob\big(T_2(\theta) > c \mid \mathfrak{S}_1\big) \le 1 - \Phi(c) \quad \text{a.s.}
	\end{equation}
	Therefore,
	\begin{align}\label{eq:proof_unconditional_DCT}
		\limsup_{\underline{n} \to \infty} \, \prob\big(T_2(\theta) > c\big) 
		&= \limsup_{\underline{n} \to \infty} \, \expt_{\mathfrak{S}_1}\Big[ \prob\big(T_2(\theta) > c \mid \mathfrak{S}_1\big) \Big] \nonumber \\
		&\le \expt_{\mathfrak{S}_1}\left[ \limsup_{\underline{n} \to \infty} \, \prob\big(T_2(\theta) > c \mid \mathfrak{S}_1\big) \right] \nonumber \\
		&\le \expt_{\mathfrak{S}_1}\big[ 1 - \Phi(c) \big] = 1 - \Phi(c).
	\end{align}
	Setting $c = z_{1-\alpha/2}$ yields $\limsup_{\underline{n} \to \infty} \prob(T_2(\theta) > z_{1-\alpha/2}) \le \alpha/2$. By symmetry, we also have $\limsup_{\underline{n} \to \infty} \prob(T_1(\theta) > z_{1-\alpha/2}) \le \alpha/2$. This completes the proof.
\end{proof}

\section{Proof for Theorem \ref{thm:id_validation_multi_marginal}}\label{sec:proof_id_validation_multi_marginal}

Let $\theta \in \Theta$ be fixed. To streamline notation, let $X_d \coloneqq (W, \hat{Z}, S)$ and, for each $l=1,\dots,L$, let $X_{v,l} \coloneqq (Z_l, \hat{Z}'_l, S'_l)$. We denote $X_v \coloneqq (Z, \hat{Z}', S')$, which is the concatenation of $(X_{v,l})_{l=1}^L$. We define the bilinear objective functional $\Phi : \mathcal{H}_L'\!\left(F, (G_l)_{l=1}^L\right) \times \mathbb{B} \to \real$ as
\[
	\Phi(H', \lambda) \coloneqq \expt_{H'}\big[\lambda^{\top}\tilde{q}(X_d, X_v; \theta)\big].
\]
The proof proceeds in three steps, mirroring the structure of the proof for Theorem~\ref{thm:id_validation}.

\noindent\textbf{Step 1. }
By Definition~\ref{def:id_set_validation_multiple_marginal} and a decoupling argument identical to that in the baseline setting, we introduce the equivalent augmented representation of the identified set. Let $\Theta'_I\!\left(F, (G_l)_{l=1}^L\right)$ be the set of parameters $\theta \in \Theta$ such that there exists a joint distribution $H' \in \mathcal{H}_L'\!\left(F, (G_l)_{l=1}^L\right)$ satisfying $\expt_{H'}[\tilde{q}(X_d, X_v; \theta)] = \mathbf{0}$. The exact-matching conditions $\expt_{H'}\big[|\hat{Z}_j - \hat{Z}'_j|\big] = 0$ and $\expt_{H'}\big[|S_k - S'_k|\big] = 0$ embedded in $\tilde{q}$ imply that under any such $H'$, $\hat{Z}' = \hat{Z}$ and $S' = S$ almost surely. Any distribution in $\mathcal{H}_L'\!\left(F, (G_l)_{l=1}^L\right)$ satisfying these exact-matching conditions collapses to a valid distribution in $\mathcal{H}_L\!\left(F, (G_l)_{l=1}^L\right)$. It follows that $\Theta'_I\!\left(F, (G_l)_{l=1}^L\right) = \Theta_I\!\left(F, (G_l)_{l=1}^L\right)$. Thus, $\theta \in \Theta_I\!\left(F, (G_l)_{l=1}^L\right)$ if and only if there exists an $H^* \in \mathcal{H}_L'\!\left(F, (G_l)_{l=1}^L\right)$ such that $\expt_{H^*}[\tilde{q}(X_d, X_v; \theta)] = \mathbf{0}$.

For any coupling $H'$, let $\mu(H') \coloneqq \expt_{H'}[\tilde{q}(\cdot; \theta)]$. Because $\mathbb{B}$ is a compact convex set containing an open neighborhood of the origin, there exists an $\epsilon > 0$ such that the closed ball of radius $\epsilon$ centered at the origin is entirely contained in $\mathbb{B}$. If $\mu(H') \neq \mathbf{0}$, we can select $\lambda^* = \epsilon \frac{\mu(H')}{\|\mu(H')\|} \in \mathbb{B}$, which yields
\[
	\sup_{\lambda \in \mathbb{B}} \lambda^{\top}\mu(H') \ge (\lambda^*)^{\top}\mu(H') = \epsilon \|\mu(H')\| > 0.
\]
Conversely, if $\mu(H') = \mathbf{0}$, then trivially $\sup_{\lambda \in \mathbb{B}} \lambda^\top \mu(H') = 0$. Because $\mathbf{0} \in \mathbb{B}$, the supremum is always bounded below by zero. Thus, $\sup_{\lambda \in \mathbb{B}} \Phi(H', \lambda) \le 0$ if and only if exactly $\mu(H') = \mathbf{0}$. Therefore, $\theta \in \Theta_I\!\left(F, (G_l)_{l=1}^L\right)$ if and only if there exists an $H^* \in \mathcal{H}_L'\!\left(F, (G_l)_{l=1}^L\right)$ such that $\sup_{\lambda \in \mathbb{B}} \Phi(H^*, \lambda) \le 0$, i.e.,
\begin{equation}\label{eq:proof_min_sup_multi}
	\inf_{H'\in \mathcal{H}_L'\left(F, (G_l)_{l=1}^L\right)} \sup_{\lambda \in \mathbb{B}} \, \Phi(H', \lambda) \le 0.
\end{equation}

\noindent\textbf{Step 2. }
Because $\mathcal{X}_d$ and $\mathcal{X}_{v,l}$ for all $l=1,\dots,L$ are closed subsets of Euclidean spaces (Assumption~\ref{assu:regularity_minimax}(i)), the marginal probability measures $F$ and $G_1, \dots, G_L$ are tight. By Prokhorov's Theorem, the set of all joint probability measures on $\mathcal{X}_d \times \prod_{l=1}^L \mathcal{X}_{v,l}$ with fixed tight marginals is relatively compact in the topology of weak convergence. Because the marginal constraints defining the Fr\'{e}chet class $\mathcal{H}_L'\!\left(F, (G_l)_{l=1}^L\right)$ are closed under weak convergence, $\mathcal{H}_L'\!\left(F, (G_l)_{l=1}^L\right)$ is weakly closed. Thus, $\mathcal{H}_L'\!\left(F, (G_l)_{l=1}^L\right)$ is a convex and weakly compact set. Moreover, the parameter space $\mathbb{B} \subset \real^{d_q + d_{\hat{Z}} + d_S}$ is, by definition, a compact and convex set.

Clearly, $\Phi(H', \lambda)$ is linear (hence concave and convex) in both $H'$ and $\lambda$. For a fixed $H'$, $\Phi(\cdot, \lambda)$ is trivially continuous in $\lambda$. We must show that for a fixed $\lambda$, the mapping $H' \mapsto \Phi(H', \lambda)$ is continuous with respect to the weak topology on $\mathcal{H}_L'\!\left(F, (G_l)_{l=1}^L\right)$.

By Assumption~\ref{assu:regularity_minimax}(ii), the augmented moment vector $\tilde{q}(X_d, X_v; \theta)$ is continuous in the topology of weak convergence. To establish weak continuity of its expectation over the Fr\'{e}chet class, we must bound it by an integrable envelope. Since all norms on finite-dimensional Euclidean spaces are equivalent, it suffices to bound the $L_1$ norm of $\tilde{q}$, which is bounded by the sum of the absolute values of its components:
\[
	\|\tilde{q}(X_d, X_v; \theta)\|_1 = \|q(W, Z; \theta)\|_1 + \|\hat{Z} - \hat{Z}'\|_1 + \|S - S'\|_1.
\]
Applying the basic inequality $\|a - b\|_1 \le \|a\|_1 + \|b\|_1$ and the structural envelope bound from Assumption~2', we have:
\[
    \|\tilde{q}(X_d, X_v; \theta)\|_1 \le b_1(W) + \sum_{l=1}^L b_{2,l}(Z) + \|\hat{Z}\|_1 + \|S\|_1 + \|\hat{Z}'\|_1 + \|S'\|_1
\]
Thus, there exists a constant $C> 0$ such that we can construct an additively separable continuous bounding function $A(X_d) + \sum_{l=1}^L B_l(X_{v,l})$ such that $\|\tilde{q}(X_d, X_v; \theta)\|_1 \le A(X_d) + \sum_{l=1}^L B_l(X_{v,l})$, where
\begin{align*}
	A(W, \hat{Z}, S) &\coloneqq C \big[b_1(W) + \|\hat{Z}\| + \|S\|\big], \\
	B_l(Z_l, \hat{Z}'_l, S'_l) &\coloneqq C \big[b_{2,l}(Z_l) + \|\hat{Z}'_l\| + \|S'_l\|\big].
\end{align*}
Under Assumption 1', the marginal distributions of $(\hat{Z}_l, S_l)$ under $F$ and $(\hat{Z}'_l, S'_l)$ under $G_l$ are identical. Consequently, Assumption~\ref{assu:regularity_minimax}(iii) implies that both sets of variables possess finite first moments under their respective marginals. It follows immediately that $A \in \mathcal{L}^1(F)$ and $B_l \in \mathcal{L}^1(G_l)$ for each $l=1,\dots,L$.

Because the marginals of any $H' \in \mathcal{H}_L'\!\left(F, (G_l)_{l=1}^L\right)$ are exactly $F$ and $G_l$ for $l=1,\dots,L$, the expected value of this bounding envelope is identical and finite for every coupling in the Fr\'{e}chet class. By standard optimal transport theory \autocite[e.g., a slight extension of Lemma 4.3 in][]{villani_optimal_2009}, 
 for any $\lambda\in \mathbb{B}$, $\Phi(H',\lambda)$ is lower semicontinuous in $H'$ with respect to weak convergence. Since $\Phi(H', -\lambda) = -\Phi(H',\lambda)$. This also implies that for any $\lambda\in \mathbb{B}$, $\Phi(H',\lambda)$ is upper semicontinuous in $H'$ with respect to weak convergence.  Thus, $H' \mapsto \Phi(H', \lambda)$ is weakly continuous.

\noindent\textbf{Step 3. }
We have established that $\mathcal{H}_L'\!\left(F, (G_l)_{l=1}^L\right)$ and $\mathbb{B}$ are compact convex sets, and that $\Phi(H', \lambda)$ is a continuous bilinear functional. Therefore, all conditions for Sion's Minimax Theorem \textcite{sion_general_1958} are satisfied. This yields two crucial results: the infimum over $H'$ can be attained (i.e., $\inf$ is a $\min$), and we can switch the order of these operators:
\begin{equation}\label{eq:proof_minimax_swap_multi}
	\min_{H'\in \mathcal{H}_L'\left(F, (G_l)_{l=1}^L\right)} \max_{\lambda \in \mathbb{B}} \, \Phi(H', \lambda) 
	= 
	\max_{\lambda \in \mathbb{B}} \min_{H'\in \mathcal{H}_L'\left(F, (G_l)_{l=1}^L\right)} \, \Phi(H', \lambda).
\end{equation}
Substituting \eqref{eq:proof_minimax_swap_multi} into \eqref{eq:proof_min_sup_multi} yields that $\theta \in \Theta_I\!\left(F, (G_l)_{l=1}^L\right)$ if and only if
\[
	\max_{\lambda \in \mathbb{B}} \min_{H'\in \mathcal{H}_L'\left(F, (G_l)_{l=1}^L\right)} \expt_{H'}\big[\lambda^{\top}\tilde{q}(W, Z, \hat{Z}, S, \hat{Z}', S'; \theta)\big] \le 0.
\]
Moreover, the maximum and minimum in \eqref{eq:maxmin_characterization_multi_marginal} are both attained. This completes the proof.

\end{appendix}

\newpage

\printbibliography

@article{hu_instrumental_2008,
	title = {Instrumental {Variable} {Treatment} of {Nonclassical} {Measurement} {Error} {Models}},
	volume = {76},
	issn = {0012-9682, 1468-0262},
	url = {http://doi.wiley.com/10.1111/j.0012-9682.2008.00823.x},
	doi = {10.1111/j.0012-9682.2008.00823.x},
	language = {en},
	number = {1},
	urldate = {2021-09-25},
	journal = {Econometrica},
	author = {Hu, Yingyao and Schennach, Susanne M.},
	month = jan,
	year = {2008},
	pages = {195--216},
}

@book{villani_optimal_2009,
	address = {Berlin, Heidelberg},
	series = {Grundlehren der mathematischen {Wissenschaften}},
	title = {Optimal {Transport}: {Old} and {New}},
	volume = {338},
	isbn = {978-3-540-71049-3 978-3-540-71050-9},
	url = {http://link.springer.com/10.1007/978-3-540-71050-9},
	doi = {10.1007/978-3-540-71050-9},
	urldate = {2021-04-19},
	publisher = {Springer Berlin Heidelberg},
	author = {Villani, Cédric},
	editor = {Berger, M. and Eckmann, B. and de la Harpe, P. and Hirzebruch, F. and Hitchin, N. and Hörmander, L. and Kupiainen, A. and Lebeau, G. and Ratner, M. and Serre, D. and Sinai, Ya. G. and Sloane, N. J. A. and Vershik, A. M. and Waldschmidt, M.},
	year = {2009},
}

@article{chen_measurement_2005,
	title = {Measurement {Error} {Models} with {Auxiliary} {Data}},
	volume = {72},
	issn = {0034-6527, 1467-937X},
	url = {https://academic.oup.com/restud/article-lookup/doi/10.1111/j.1467-937X.2005.00335.x},
	doi = {10.1111/j.1467-937X.2005.00335.x},
	language = {en},
	number = {2},
	urldate = {2022-03-07},
	journal = {Review of Economic Studies},
	author = {Chen, Xiaohong and Hong, Han and Tamer, Elie},
	month = apr,
	year = {2005},
	pages = {343--366},
}

@article{sion_general_1958,
	title = {On general minimax theorems.},
	volume = {8},
	number = {1},
	journal = {Pacific Journal of mathematics},
	publisher = {Pacific Journal of Mathematics, A Non-profit Corporation},
	author = {Sion, Maurice},
	year = {1958},
	pages = {171--176},
}

@book{ambrosio_lectures_2021,
	address = {Cham},
	series = {Unitext - {La} matematica per il 3 + 2},
	title = {Lectures on optimal transport},
	isbn = {978-3-030-72161-9},
	doi = {10.1007/978-3-72162-6},
	language = {eng},
	number = {volume 130},
	publisher = {Springer},
	author = {Ambrosio, Luigi and Brué, Elia and Semola, Daniele},
	year = {2021},
}

@article{cross_regressions_2002,
	title = {Regressions, {Short} and {Long}},
	volume = {70},
	issn = {0012-9682, 1468-0262},
	url = {http://doi.wiley.com/10.1111/1468-0262.00279},
	doi = {10.1111/1468-0262.00279},
	language = {en},
	number = {1},
	urldate = {2024-01-05},
	journal = {Econometrica},
	author = {Cross, Philip J. and Manski, Charles F.},
	month = jan,
	year = {2002},
	pages = {357--368},
}

@article{mullainathan_machine_2017,
	title = {Machine {Learning}: {An} {Applied} {Econometric} {Approach}},
	volume = {31},
	doi = {10.1257/jep.31.2.87},
	number = {2},
	journal = {Journal of Economic Perspectives},
	publisher = {American Economic Association},
	author = {Mullainathan, Sendhil and Spiess, Jann},
	month = may,
	year = {2017},
	pages = {87--106},
}

@article{berry_automobile_1995,
	title = {Automobile {Prices} in {Market} {Equilibrium}},
	volume = {63},
	doi = {10.2307/2171802},
	number = {4},
	journal = {Econometrica},
	publisher = {JSTOR},
	author = {Berry, Steven and Levinsohn, James and Pakes, Ariel},
	year = {1995},
	pages = {841},
}

@misc{battaglia_inference_2025,
	title = {Inference for {Regression} with {Variables} {Generated} by {AI} or {Machine} {Learning}},
	url = {https://arxiv.org/abs/2402.15585},
	author = {Battaglia, Laura and Christensen, Timothy and Hansen, Stephen and Sacher, Szymon},
	year = {2025},
	note = {\_eprint: 2402.15585},
}

@article{sager_clean_2025,
	title = {Clean {Identification}? {The} {Effects} of the {Clean} {Air} {Act} on {Air} {Pollution}, {Exposure} {Disparities}, and {House} {Prices}},
	volume = {17},
	issn = {1945-7731, 1945-774X},
	shorttitle = {Clean {Identification}?},
	url = {https://pubs.aeaweb.org/doi/10.1257/pol.20220745},
	doi = {10.1257/pol.20220745},
	language = {en},
	number = {1},
	urldate = {2026-02-19},
	journal = {American Economic Journal: Economic Policy},
	author = {Sager, Lutz and Singer, Gregor},
	month = feb,
	year = {2025},
	pages = {1--36},
}

@article{meng_estimated_2019,
	title = {Estimated {Long}-{Term} (1981–2016) {Concentrations} of {Ambient} {Fine} {Particulate} {Matter} across {North} {America} from {Chemical} {Transport} {Modeling}, {Satellite} {Remote} {Sensing}, and {Ground}-{Based} {Measurements}},
	volume = {53},
	copyright = {http://pubs.acs.org/page/policy/authorchoice\_termsofuse.html},
	issn = {0013-936X, 1520-5851},
	url = {https://pubs.acs.org/doi/10.1021/acs.est.8b06875},
	doi = {10.1021/acs.est.8b06875},
	language = {en},
	number = {9},
	urldate = {2026-02-19},
	journal = {Environmental Science \& Technology},
	author = {Meng, Jun and Li, Chi and Martin, Randall V. and Van Donkelaar, Aaron and Hystad, Perry and Brauer, Michael},
	month = may,
	year = {2019},
	pages = {5071--5079},
}

@techreport{hansen_remote_2023,
	title = {Remote {Work} across {Jobs}, {Companies}, and {Space}},
	url = {http://www.jstor.org/stable/resrep64667},
	urldate = {2026-02-21},
	institution = {IZA - Institute of Labor Economics},
	author = {Hansen, Stephen and Lambert, Peter John and Bloom, Nick and Davis, Steven J. and Sadun, Raffaella and Taska, Bledi},
	year = {2023},
}

@techreport{barrero_why_2021,
	address = {Cambridge, MA},
	title = {Why {Working} from {Home} {Will} {Stick}},
	url = {http://www.nber.org/papers/w28731.pdf},
	doi = {10.3386/w28731},
	language = {en},
	number = {w28731},
	urldate = {2026-02-21},
	institution = {National Bureau of Economic Research},
	author = {Barrero, Jose Maria and Bloom, Nicholas and Davis, Steven},
	month = apr,
	year = {2021},
	pages = {w28731},
}

@techreport{aksoy_working_2022,
	address = {Cambridge, MA},
	title = {Working from {Home} {Around} the {World}},
	url = {http://www.nber.org/papers/w30446.pdf},
	doi = {10.3386/w30446},
	language = {en},
	number = {w30446},
	urldate = {2026-02-21},
	institution = {National Bureau of Economic Research},
	author = {Aksoy, Cevat Giray and Barrero, Jose Maria and Bloom, Nicholas and Davis, Steven and Dolls, Mathias and Zarate, Pablo},
	month = sep,
	year = {2022},
	pages = {w30446},
}

@misc{sanh_distilbert_2020,
	title = {{DistilBERT}, a distilled version of {BERT}: smaller, faster, cheaper and lighter},
	shorttitle = {{DistilBERT}, a distilled version of {BERT}},
	url = {http://arxiv.org/abs/1910.01108},
	doi = {10.48550/arXiv.1910.01108},
	urldate = {2026-02-21},
	publisher = {arXiv},
	author = {Sanh, Victor and Debut, Lysandre and Chaumond, Julien and Wolf, Thomas},
	month = mar,
	year = {2020},
	note = {arXiv:1910.01108 [cs]},
}

@misc{fan_partial_2025,
	title = {Partial {Identification} in {Moment} {Models} with {Incomplete} {Data}--{A} {Conditional} {Optimal} {Transport} {Approach}},
	url = {http://arxiv.org/abs/2503.16098},
	doi = {10.48550/arXiv.2503.16098},
	urldate = {2026-02-26},
	publisher = {arXiv},
	author = {Fan, Yanqin and Park, Hyeonseok and Pass, Brendan and Shi, Xuetao},
	month = oct,
	year = {2025},
	note = {arXiv:2503.16098 [econ]},
}

@incollection{ridder_chapter_2007,
	title = {Chapter 75 {The} {Econometrics} of {Data} {Combination}},
	volume = {6},
	copyright = {https://www.elsevier.com/tdm/userlicense/1.0/},
	isbn = {978-0-444-53200-8},
	url = {https://linkinghub.elsevier.com/retrieve/pii/S1573441207060758},
	doi = {10.1016/S1573-4412(07)06075-8},
	language = {en},
	urldate = {2026-02-27},
	booktitle = {Handbook of {Econometrics}},
	publisher = {Elsevier},
	author = {Ridder, Geert and Moffitt, Robert},
	year = {2007},
	pages = {5469--5547},
}

@misc{li_finite_2025,
	title = {Finite {Sample} {Inference} in {Incomplete} {Models}},
	url = {http://arxiv.org/abs/2204.00473},
	doi = {10.48550/arXiv.2204.00473},
	urldate = {2026-02-27},
	publisher = {arXiv},
	author = {Li, Lixiong and Henry, Marc},
	month = oct,
	year = {2025},
	note = {arXiv:2204.00473 [econ]},
}

@article{dhaultfoeuille_partially_2025,
	title = {Partially {Linear} {Models} under {Data} {Combination}},
	volume = {92},
	copyright = {https://academic.oup.com/pages/standard-publication-reuse-rights},
	issn = {0034-6527, 1467-937X},
	url = {https://academic.oup.com/restud/article/92/1/238/7637571},
	doi = {10.1093/restud/rdae022},
	language = {en},
	number = {1},
	urldate = {2026-02-28},
	journal = {Review of Economic Studies},
	author = {D’Haultfœuille, X and Gaillac, C and Maurel, A},
	month = jan,
	year = {2025},
	pages = {238--267},
}

@article{hwang_bounding_2025,
	title = {Bounding {Omitted} {Variable} {Bias} {Using} {Auxiliary} {Data}: {With} an {Application} to {Estimate} {Neighborhood} {Effects}},
	issn = {0735-0015, 1537-2707},
	shorttitle = {Bounding {Omitted} {Variable} {Bias} {Using} {Auxiliary} {Data}},
	url = {https://www.tandfonline.com/doi/full/10.1080/07350015.2025.2606063},
	doi = {10.1080/07350015.2025.2606063},
	language = {en},
	urldate = {2026-02-28},
	journal = {Journal of Business \& Economic Statistics},
	author = {Hwang, Yujung},
	month = dec,
	year = {2025},
	pages = {1--25},
}

@article{gentzkow_what_2010,
	title = {What {Drives} {Media} {Slant}? {Evidence} {From} {U}.{S}. {Daily} {Newspapers}},
	volume = {78},
	copyright = {http://doi.wiley.com/10.1002/tdm\_license\_1.1},
	issn = {0012-9682},
	shorttitle = {What {Drives} {Media} {Slant}?},
	url = {http://doi.wiley.com/10.3982/ECTA7195},
	doi = {10.3982/ECTA7195},
	language = {en},
	number = {1},
	urldate = {2026-03-25},
	journal = {Econometrica},
	author = {Gentzkow, Matthew and Shapiro, Jesse},
	year = {2010},
	pages = {35--71},
}

@article{groseclose_measure_2005,
	title = {A {Measure} of {Media} {Bias}},
	volume = {120},
	issn = {0033-5533, 1531-4650},
	url = {https://academic.oup.com/qje/article-lookup/doi/10.1162/003355305775097542},
	doi = {10.1162/003355305775097542},
	language = {en},
	number = {4},
	urldate = {2026-03-25},
	journal = {The Quarterly Journal of Economics},
	author = {Groseclose, T. and Milyo, J.},
	month = nov,
	year = {2005},
	pages = {1191--1237},
}

@article{hassan_firm-level_2019,
	title = {Firm-{Level} {Political} {Risk}: {Measurement} and {Effects}},
	volume = {134},
	copyright = {https://academic.oup.com/journals/pages/open\_access/funder\_policies/chorus/standard\_publication\_model},
	issn = {0033-5533, 1531-4650},
	shorttitle = {Firm-{Level} {Political} {Risk}},
	url = {https://academic.oup.com/qje/article/134/4/2135/5531768},
	doi = {10.1093/qje/qjz021},
	language = {en},
	number = {4},
	urldate = {2026-03-27},
	journal = {The Quarterly Journal of Economics},
	author = {Hassan, Tarek A and Hollander, Stephan and Van Lent, Laurence and Tahoun, Ahmed},
	month = nov,
	year = {2019},
	pages = {2135--2202},
}

@article{hu_illuminating_2022,
	title = {Illuminating economic growth},
	volume = {228},
	issn = {03044076},
	url = {https://linkinghub.elsevier.com/retrieve/pii/S0304407621001767},
	doi = {10.1016/j.jeconom.2021.05.007},
	language = {en},
	number = {2},
	urldate = {2026-03-27},
	journal = {Journal of Econometrics},
	author = {Hu, Yingyao and Yao, Jiaxiong},
	month = jun,
	year = {2022},
	pages = {359--378},
}

@article{gentzkow_text_2019,
	title = {Text as {Data}},
	volume = {57},
	issn = {0022-0515},
	url = {https://pubs.aeaweb.org/doi/10.1257/jel.20181020},
	doi = {10.1257/jel.20181020},
	language = {en},
	number = {3},
	urldate = {2026-03-27},
	journal = {Journal of Economic Literature},
	author = {Gentzkow, Matthew and Kelly, Bryan and Taddy, Matt},
	month = sep,
	year = {2019},
	pages = {535--574},
}

@misc{kluger_prediction-powered_2025,
	title = {Prediction-{Powered} {Inference} with {Imputed} {Covariates} and {Nonuniform} {Sampling}},
	url = {http://arxiv.org/abs/2501.18577},
	doi = {10.48550/arXiv.2501.18577},
	urldate = {2026-03-28},
	publisher = {arXiv},
	author = {Kluger, Dan M. and Lu, Kerri and Zrnic, Tijana and Wang, Sherrie and Bates, Stephen},
	month = nov,
	year = {2025},
	note = {arXiv:2501.18577 [stat]},
}

@article{chen_unified_2000,
	title = {A {Unified} {Approach} to {Regression} {Analysis} {Under} {Double}-{Sampling} {Designs}},
	volume = {62},
	copyright = {https://academic.oup.com/journals/pages/open\_access/funder\_policies/chorus/standard\_publication\_model},
	issn = {1369-7412, 1467-9868},
	url = {https://academic.oup.com/jrsssb/article/62/3/449/7083278},
	doi = {10.1111/1467-9868.00243},
	language = {en},
	number = {3},
	urldate = {2026-03-28},
	journal = {Journal of the Royal Statistical Society Series B: Statistical Methodology},
	author = {Chen, Yi-Hau and Chen, Hung},
	month = sep,
	year = {2000},
	pages = {449--460},
}

@misc{meango_combining_2025,
	title = {Combining stated and revealed preferences},
	url = {http://arxiv.org/abs/2507.13552},
	doi = {10.48550/arXiv.2507.13552},
	urldate = {2026-03-28},
	publisher = {arXiv},
	author = {Meango, Romuald and Henry, Marc and Mourifie, Ismael},
	month = nov,
	year = {2025},
	note = {arXiv:2507.13552 [econ]},
}

@article{athey_surrogate_2025,
	title = {The {Surrogate} {Index}: {Combining} {Short}-{Term} {Proxies} to {Estimate} {Long}-{Term} {Treatment} {Effects} {More} {Rapidly} and {Precisely}},
	copyright = {https://creativecommons.org/licenses/by-nc/4.0/},
	issn = {0034-6527, 1467-937X},
	shorttitle = {The {Surrogate} {Index}},
	url = {https://academic.oup.com/restud/advance-article/doi/10.1093/restud/rdaf087/8268796},
	doi = {10.1093/restud/rdaf087},
	language = {en},
	urldate = {2026-03-28},
	journal = {Review of Economic Studies},
	author = {Athey, Susan and Chetty, Raj and Imbens, Guido W and Kang, Hyunseung},
	month = sep,
	year = {2025},
	pages = {rdaf087},
}

@article{santavirta_name-based_2024,
	title = {Name-{Based} {Estimators} of {Intergenerational} {Mobility}},
	volume = {134},
	copyright = {https://academic.oup.com/journals/pages/open\_access/funder\_policies/chorus/standard\_publication\_model},
	issn = {0013-0133, 1468-0297},
	url = {https://academic.oup.com/ej/article/134/663/2982/7673076},
	doi = {10.1093/ej/ueae035},
	language = {en},
	number = {663},
	urldate = {2026-03-28},
	journal = {The Economic Journal},
	author = {Santavirta, Torsten and Stuhler, Jan},
	month = sep,
	year = {2024},
	pages = {2982--3016},
}

@article{kallus_assessing_2022,
	title = {Assessing {Algorithmic} {Fairness} with {Unobserved} {Protected} {Class} {Using} {Data} {Combination}},
	volume = {68},
	issn = {0025-1909, 1526-5501},
	url = {https://pubsonline.informs.org/doi/10.1287/mnsc.2020.3850},
	doi = {10.1287/mnsc.2020.3850},
	language = {en},
	number = {3},
	urldate = {2026-03-28},
	journal = {Management Science},
	author = {Kallus, Nathan and Mao, Xiaojie and Zhou, Angela},
	month = mar,
	year = {2022},
	pages = {1959--1981},
}

@misc{sanford_adversarial_2025,
	title = {Adversarial {Debiasing} for {Unbiased} {Parameter} {Recovery}},
	copyright = {Creative Commons Attribution 4.0 International},
	url = {https://arxiv.org/abs/2502.12323},
	doi = {10.48550/ARXIV.2502.12323},
	urldate = {2026-03-29},
	publisher = {arXiv},
	author = {Sanford, Luke C and Ayers, Megan and Gordon, Matthew and Stone, Eliana},
	year = {2025},
	note = {Version Number: 1},
}

@misc{miao_task-agnostic_2024,
	title = {Task-{Agnostic} {Machine}-{Learning}-{Assisted} {Inference}},
	copyright = {Creative Commons Attribution 4.0 International},
	url = {https://arxiv.org/abs/2405.20039},
	doi = {10.48550/ARXIV.2405.20039},
	urldate = {2026-03-29},
	publisher = {arXiv},
	author = {Miao, Jiacheng and Lu, Qiongshi},
	year = {2024},
	note = {Version Number: 3},
}

@misc{zrnic_cross-prediction-powered_2023,
	title = {Cross-{Prediction}-{Powered} {Inference}},
	copyright = {arXiv.org perpetual, non-exclusive license},
	url = {https://arxiv.org/abs/2309.16598},
	doi = {10.48550/ARXIV.2309.16598},
	urldate = {2026-03-29},
	publisher = {arXiv},
	author = {Zrnic, Tijana and Candès, Emmanuel J.},
	year = {2023},
	note = {Version Number: 3},
}

@article{zhang_debiasing_2023,
	title = {Debiasing {Machine}-{Learning}- or {AI}-{Generated} {Regressors} in {Partial} {Linear} {Models}},
	issn = {1556-5068},
	url = {https://www.ssrn.com/abstract=4636026},
	doi = {10.2139/ssrn.4636026},
	language = {en},
	urldate = {2026-03-29},
	journal = {SSRN Electronic Journal},
	author = {Zhang, Jingwen and Xue, Wendao and Yu, Yifan and Tan, Yong},
	year = {2023},
}

@misc{angelopoulos_ppi_2024,
	title = {{PPI}++: {Efficient} {Prediction}-{Powered} {Inference}},
	shorttitle = {{PPI}++},
	url = {http://arxiv.org/abs/2311.01453},
	doi = {10.48550/arXiv.2311.01453},
	urldate = {2026-03-29},
	publisher = {arXiv},
	author = {Angelopoulos, Anastasios N. and Duchi, John C. and Zrnic, Tijana},
	month = mar,
	year = {2024},
	note = {arXiv:2311.01453 [stat]},
}

@article{angelopoulos_prediction-powered_2023,
	title = {Prediction-powered inference},
	volume = {382},
	number = {6671},
	journal = {Science},
	publisher = {American Association for the Advancement of Science},
	author = {Angelopoulos, Anastasios N and Bates, Stephen and Fannjiang, Clara and Jordan, Michael I and Zrnic, Tijana},
	year = {2023},
	pages = {669--674},
}

@article{allon_machine_2023,
	title = {Machine learning and prediction errors in causal inference},
	journal = {The Wharton School Research Paper Forthcoming},
	author = {Allon, Gad and Chen, Daniel and Jiang, Zhenling and Zhang, Dennis},
	year = {2023},
}

@misc{rambachan_program_2024,
	title = {Program {Evaluation} with {Remotely} {Sensed} {Outcomes}},
	copyright = {Creative Commons Zero v1.0 Universal},
	url = {https://arxiv.org/abs/2411.10959},
	doi = {10.48550/ARXIV.2411.10959},
	urldate = {2026-03-29},
	publisher = {arXiv},
	author = {Rambachan, Ashesh and Singh, Rahul and Viviano, Davide},
	year = {2024},
	note = {Version Number: 3},
}

@article{obradovic_identification_2026,
	title = {Identification of {Long}-{Term} {Treatment} {Effects} via {Temporal} {Links}, {Observational}, and {Experimental} {Data}},
	journal = {Working Paper},
	author = {Obradović, Filip},
	year = {2026},
}

\end{document}